\documentclass[pra, onecolumn,superscriptaddress,nofootinbib]{revtex4}
\usepackage[a4paper, left=1in, right=1in, top=1in, bottom=1in]{geometry}

\usepackage{cmap} 
\usepackage[utf8]{inputenc}
\usepackage[english]{babel}
\usepackage[T1]{fontenc}
\usepackage{appendix}
\usepackage{url, amsfonts, tikz, physics, listings, xcolor, graphicx, float}
\usepackage[shortlabels]{enumitem}
\usepackage{dsfont}
\usepackage{amsmath, amssymb, amstext, amscd, amsthm, makeidx, graphicx, url, mathrsfs, mathtools, longdivision, polynom, bbm, complexity}

\definecolor{blueviolet}{rgb}{0.2, 0.2, 0.6}
\definecolor{webgreen}{rgb}{0,.5,0}
\definecolor{webbrown}{rgb}{.6,0,0}
\usepackage[pdftex,
  bookmarks=false,
  colorlinks=true, 
  urlcolor=webbrown,
  linkcolor=blueviolet,
  citecolor=webgreen,
  pdfstartpage=1,
  pdfstartview={FitH},  
  bookmarksopen=false
  ]{hyperref}

\makeatletter
\newcommand\RedeclareMathOperator{%
  \@ifstar{\def\rmo@s{m}\rmo@redeclare}{\def\rmo@s{o}\rmo@redeclare}%
}
\newcommand\rmo@redeclare[2]{%
  \begingroup \escapechar\m@ne\xdef\@gtempa{{\string#1}}\endgroup
  \expandafter\@ifundefined\@gtempa
     {\@latex@error{\noexpand#1undefined}\@ehc}%
     \relax
  \expandafter\rmo@declmathop\rmo@s{#1}{#2}}
\newcommand\rmo@declmathop[3]{%
  \DeclareRobustCommand{#2}{\qopname\newmcodes@#1{#3}}%
}
\@onlypreamble\RedeclareMathOperator
\makeatother

\allowdisplaybreaks

\RedeclareMathOperator*{\E}{{\mathbb{E}}}

\DeclareMathOperator*{\argmax}{arg\,max}

\numberwithin{equation}{section}

\newtheorem{theorem}{Theorem}
\newtheorem{prop}{Proposition}
\newtheorem{lemma}{Lemma}
\newtheorem{corollary}{Corollary}
\newtheorem{definition}{Definition}
\newtheorem{fact}{Fact}

\newcommand{\indicator}{\mathds{1}}

\begin{document}

\title{Improved machine learning algorithm for\\ predicting ground state properties}
\date{\today}
\author{Laura Lewis}
\affiliation{California Institute of Technology, Pasadena, CA, USA}
\author{Hsin-Yuan Huang}
\affiliation{California Institute of Technology, Pasadena, CA, USA}
\author{Viet T. Tran}
\affiliation{Johannes Kepler University, Linz, Austria}
\author{Sebastian Lehner}
\affiliation{Johannes Kepler University, Linz, Austria}
\author{Richard Kueng}
\affiliation{Johannes Kepler University, Linz, Austria}
\author{John Preskill}
\affiliation{California Institute of Technology, Pasadena, CA, USA}
\affiliation{AWS Center for Quantum Computing, Pasadena, CA, USA}

\begin{abstract}
Finding the ground state of a quantum many-body system is a fundamental problem in quantum physics.
In this work, we give a classical machine learning (ML) algorithm for predicting ground state properties with an inductive bias encoding geometric locality.
The proposed ML model can efficiently predict ground state properties of an $n$-qubit gapped local Hamiltonian after learning from only $\mathcal{O}(\log(n))$ data about other Hamiltonians in the same quantum phase of matter.
This improves substantially upon previous results that require $\mathcal{O}(n^c)$ data for a large constant $c$.
Furthermore, the training and prediction time of the proposed ML model scale as $\mathcal{O}(n \log n)$ in the number of qubits $n$.
Numerical experiments on physical systems with up to $45$ qubits confirm the favorable scaling in predicting ground state properties using a small training dataset.
\end{abstract}

\maketitle

{\renewcommand\addcontentsline[3]{} \section{Introduction}}

Finding the ground state of a quantum many-body system is a fundamental problem with far-reaching consequences for physics, materials science, and chemistry.
Many powerful methods \cite{HohenbergKohn, NobelKohn, CEPERLEY555,SandvikSSE,becca_sorella_2017, DMRG1,DMRG2} have been proposed,
but classical computers still struggle to solve many general classes of the ground state problem.
To extend the reach of classical computers, classical machine learning (ML) methods have recently been adapted to study this problem~\cite{CarleoRMP,APXReview, dassarma2017, carrasquilla2017nature,Carleo_2017,torlai_learning_2016,Nomura2017, evert2017nature,leiwang2016,gilmer2017neural,torlai_Tomo,vargas2018extrapolating,schutt2019unifying,Glasser2018,caro2022out,rodriguez2019identifying,qiao2020orbnet,choo_fermionicnqs2020,kawai2020predicting,moreno2020deep,Kottmann2021}.
A recent work \cite{huang2021provably} proposes a polynomial-time classical ML algorithm that can efficiently predict ground state properties of gapped geometrically local Hamiltonians, after learning from data obtained by measuring other Hamiltonians in the same quantum phase of matter.
Furthermore, \cite{huang2021provably} shows that under a widely accepted conjecture, no polynomial-time classical algorithm can achieve the same performance guarantee.
However, although the ML algorithm given in \cite{huang2021provably} uses a polynomial amount of training data and computational time, the polynomial scaling $\mathcal{O}(n^c)$ has a very large degree~$c$.
Moreover, when the prediction error $\epsilon$ is small, the amount of training data grows exponentially in $1 / \epsilon$, indicating that a very small prediction error cannot be achieved efficiently.

In this work, we present an improved ML algorithm for predicting ground state properties.
We consider an $m$-dimensional vector $x \in [-1, 1]^{m}$ that parameterizes an $n$-qubit gapped geometrically local Hamiltonian given as
\begin{equation}
    H(x) = \sum_{j} h_j(\vec{x}_j),
\end{equation}
where $x$ is the concatenation of constant-dimensional vectors $\vec{x}_1, \ldots, \vec{x}_L$ parameterizing the few-body interaction $h_j(\vec{x}_j)$.
Let $\rho(x)$ be the ground state of $H(x)$ and $O$ be a sum of geometrically local observables with $\norm{O}_\infty \leq 1$.
We assume that the geometry of the $n$-qubit system is known, but we do not know how $h_j(\vec{x}_j)$ is parameterized or what the observable $O$ is.
The goal is to learn a function $h^*(x)$ that approximates the ground state property $\Tr(O \rho(x))$ from a classical dataset,
\begin{equation}
    \big( x_\ell, y_\ell \big), \quad \forall \ell = 1, \ldots, N,
\end{equation}
where $y_\ell \approx \Tr(O \rho(x_\ell))$ records the ground state property for $x_\ell \in [-1, 1]^m$ sampled from an arbitrary unknown distribution~$\mathcal{D}$.

The setting considered in this work is very similar to that in \cite{huang2021provably}, but we assume the geometry of the $n$-qubit system to be known, which is necessary to overcome the sample complexity lower bound of $N = n^{\Omega(1 / \epsilon)}$ given in \cite{huang2021provably}.
One may compare the setting to that of finding ground states using adiabatic quantum computation \cite{farhi2000quantum,mizel2007simple,childs2001robustness,aharonov2008adiabatic,barends2016digitized,albash2018adiabatic,du2010nmr,wan2020fast}.
To find the ground state property $\Tr(O \rho(x))$ of $H(x)$,
this class of quantum algorithms requires the ground state $\rho_0$ of another Hamiltonian $H_0$ stored in quantum memory, explicit knowledge of a gapped path connecting $H_0$ and $H(x)$, and an explicit description of~$O$.
In contrast, here we focus on ML algorithms that are entirely classical, have no access to quantum state data, and have no knowledge about the Hamiltonian $H(x)$, the observable $O$, or the gapped paths between $H(x)$ and other Hamiltonians.

The proposed ML algorithm uses a nonlinear feature map $x \mapsto \phi(x)$ with a geometric inductive bias built into the mapping.
At a high level, the high-dimensional vector $\phi(x)$ contains nonlinear functions for each geometrically local subset of coordinates in the $m$-dimensional vector~$x$.
Here, the geometry over coordinates of the vector $x$ is defined using the geometry of the $n$-qubit system.
The ML algorithm learns a function $h^*(x) = \mathbf{w}^* \cdot \phi(x)$ by training an $\ell_1$-regularized regression (LASSO) \cite{doi:10.1137/0907087, tibshirani1996regression, mohri2018foundations} in the feature space.
We prove that given $\epsilon = \Theta(1)$, the improved ML algorithm can use a dataset size of
\begin{equation}
    N = \mathcal{O}\left(\log\left(n\right)\right),
\end{equation}
to learn a function $h^*(x)$ with an average prediction error of at most $\epsilon$,
\begin{equation}
    \E_{x \sim \mathcal{D}} \left| h^*(x) - \Tr(O \rho(x)) \right|^2 \leq \epsilon,
\end{equation}
with high success probability.

The sample complexity $N = \mathcal{O}\left(\log\left(n\right)\right)$ of the proposed ML algorithm improves substantially over the sample complexity of $N = \mathcal{O}(n^c)$ in the previously best-known classical ML algorithm \cite{huang2021provably}, where $c$ is a very large constant.
The computational time of both the improved ML algorithm and the ML algorithm in \cite{huang2021provably} is $\mathcal{O}(n N)$.
Hence, the logarithmic sample complexity $N$ immediately implies a nearly linear computational time.
In addition to the reduced sample complexity and computational time,
the proposed ML algorithm works for any distribution over $x$, while the best previously known algorithm \cite{huang2021provably} works only for the uniform distribution over $[-1, 1]^m$.
Furthermore, when we consider the scaling with the prediction error $\epsilon$, the best known classical ML algorithm in \cite{huang2021provably} has a sample complexity of $N = n^{\mathcal{O}(1 / \epsilon)}$, which is exponential in $1 / \epsilon$.
In contrast, the improved ML algorithm has a sample complexity of $N = \log(n) 2^{\mathrm{polylog}(1 / \epsilon)}$, which is quasi-polynomial in $1 / \epsilon$.
In combination with the classical shadow formalism \cite{huang2020predicting, elben2020mixed, elben2022randomized, wan2022matchgate, bu2022classical}, the proposed ML algorithm also yields the same reduction in sample and time complexity compared to \cite{huang2021provably} for predicting ground state representations.

\begin{figure}[t]
    \centering
    \includegraphics[width=1.0\linewidth]{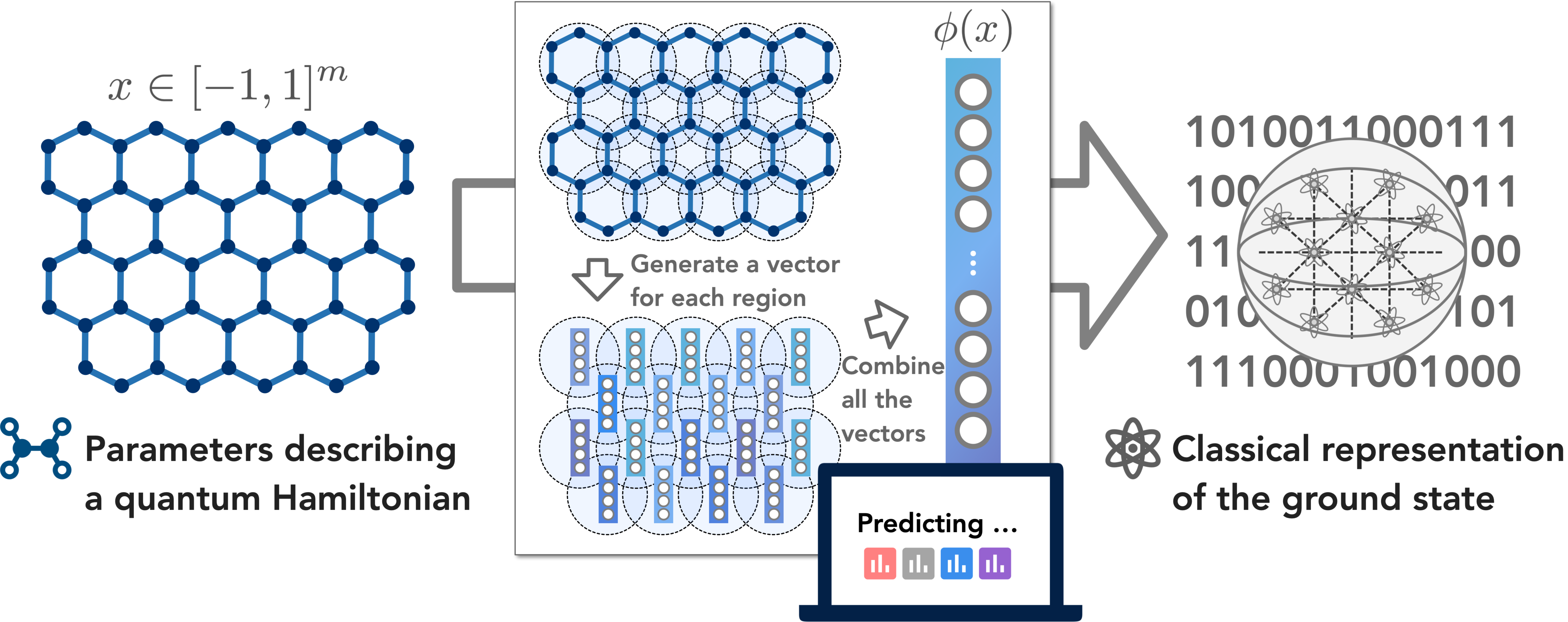}
    \caption{\textbf{Overview of the proposed machine learning algorithm.} Given a vector $x \in [-1, 1]^m$ that parameterizes a quantum many-body Hamiltonian $H(x)$. The algorithm uses a geometric structure to create a high-dimensional vector $\phi(x) \in \mathbb{R}^{m_\phi}$. The ML algorithm then predicts properties or a representation of the ground state $\rho(x)$ of Hamiltonian $H(x)$ using the $m_\phi$-dimensional vector $\phi(x)$. }
    \label{fig:overview}
\end{figure}

\vspace{2em}
{\renewcommand\addcontentsline[3]{}\section{ML algorithm and rigorous guarantee}}

The central component of the improved ML algorithm is the geometric inductive bias built into our feature mapping $x \in [-1, 1]^m \mapsto \phi(x) \in \mathbb{R}^{m_\phi}$.
To describe the ML algorithm, we first need to present some definitions relating to this geometric structure.

\vspace{2em}
{\renewcommand\addcontentsline[3]{} \subsection{Definitions}}

We consider $n$ 
qubits arranged at locations, or sites, in a $d$-dimensional space, e.g., a spin chain ($d=1$), a square lattice ($d=2$), or a cubic lattice ($d=3$).
This geometry is characterized by the distance $d_{\mathrm{qubit}}(i, i')$ between any two qubits $i$ and $i'$.
Using the distance $d_{\mathrm{qubit}}$ between qubits, we can define the geometry of local observables.
Given any two observables $O_A, O_B$ on the $n$-qubit system,
we define the distance $d_{\mathrm{obs}}(O_A, O_B)$ between the two observables as the minimum distance between the qubits that $O_A$ and $O_B$ act on.
We also say an observable is geometrically local if it acts nontrivially only on nearby qubits under the distance metric $d_{\mathrm{qubit}}$.
We then define $S^{\mathrm{(geo)}}$ as the set of all geometrically local Pauli observables, i.e., geometrically local observables that belong to the set $\{I, X, Y, Z\}^{\otimes n}$.
The size of $S^{\mathrm{(geo)}}$ is $\mathcal{O}(n)$, linear in the total number of qubits.

With these basic definitions in place, we now define a few more geometric objects.
The first object is the set of coordinates in the $m$-dimensional vector $x$ that are close to a geometrically local Pauli observable $P$. This is formally given by,
\begin{equation}
  \label{eq:ip}
    I_P \triangleq \left\{ c \in \{1, \dots, m\} : d_{\mathrm{obs}}(h_{j(c)}, P) \leq \delta_1 \right\},
\end{equation}
where $h_{j(c)}$ is the few-body interaction term in the $n$-qubit Hamiltonian $H(x)$ that is parameterized by the variable $x_c \in [-1, 1]$, and $\delta_1$ is an efficiently computable hyperparameter that is determined later.
Note that, by definition, each variable $x_c$ parameterizes one of the interaction terms $h_{j(c)}$.
Intuitively, $I_P$ is the set of coordinates that have the strongest influence on the function $\Tr(P \rho(x))$.

The second geometric object is a discrete lattice over the space $[-1, 1]^m$ associated to each subset $I_P$ of coordinates.
For any geometrically local Pauli observable $P \in S^{\mathrm{(geo)}}$, we define $X_P$ to contain all vectors $x$ that take on value $0$ for coordinates outside $I_P$ and take on a set of discrete values for coordinates inside $I_P$. Formally, this is given by
\begin{equation}
    X_P \triangleq \left.\begin{cases}
    x \in [-1,1]^m : \text{if } c \notin I_P, \,\, x_{c} = 0\\
    \hspace{62pt} \text{if } c \in I_P, \,\, x_{c} \in \left\{0, \pm \delta_2, \pm 2\delta_2,\dots, \pm 1\right\}
    \end{cases}\right\},
\end{equation}
where $\delta_2$ is an efficiently computable hyperparameter to be determined later.
The definition of $X_P$ is meant to enumerate all sufficiently different vectors for coordinates in the subset $I_P \subseteq \{1, \ldots, m\}$.

Now given a geometrically local Pauli observable $P$ and a vector $x$ in the discrete lattice $X_P \subseteq [-1, 1]^m$, the third object is a set $T_{x, P}$ of vectors in $[-1, 1]^m$ that are close to $x$ for coordinates in $I_P$. This is formally defined as,
\begin{equation}
    T_{x, P} \triangleq \left\{ x' \in [-1,1]^m : -\frac{\delta_2}{2} < x_{c} - x_{c}' \leq \frac{\delta_2}{2}, \forall c \in I_P \right\}.
\end{equation}
The set $T_{x, P}$ is defined as a thickened affine subspace close to the vector $x$ for coordinates in $I_P$.
If a vector $x'$ is in $T_{x, P}$, then $x'$ is close to $x$ for all coordinates in $I_P$, but $x'$ may be far away from $x$ for coordinates outside of $I_P$.

\vspace{2em}
{\renewcommand\addcontentsline[3]{} \subsection{Feature mapping and ML model}\label{subsec:ML-model}}

We can now define the feature map $\phi$ taking an $m$-dimensional vector $x$ to an $m_\phi$-dimensional vector $\phi(x)$ using the thickened affine subspaces $T_{x', P}$ for every geometrically local Pauli observable $P \in S^{\mathrm{(geo)}}$ and every vector $x'$ in the discrete lattice $X_P$.
The dimension of the vector $\phi(x)$ is given by $m_{\phi} = \sum_{P \in S^{\mathrm{(geo)}}} |X_P|$.
Each coordinate of the vector $\phi(x)$ is indexed by $x' \in X_P$ and $P \in S^{\mathrm{(geo)}}$ with
\begin{equation} \label{eq:phi-def}
    \phi(x)_{x', P} \triangleq \indicator\left[x \in T_{x', P}\right],
\end{equation}
which is the indicator function checking if $x$ belongs to the thickened affine subspace.
Recall that this means each coordinate of the $m_\phi$-dimensional vector $\phi(x)$ checks if $x$ is close to a point $x'$ on a discrete lattice $X_P$ for the subset $I_P$ of coordinates close to a geometrically local Pauli observable $P$.

The classical ML model we consider is an $\ell_1$-regularized regression (LASSO) over the $\phi(x)$ space.
More precisely, given an efficiently computable hyperparameter $B > 0$, the classical ML model finds an $m_\phi$-dimensional vector $\mathbf{w}^*$ from the following optimization problem,
\begin{equation}
    \min_{\substack{\mathbf{w} \in \mathbb{R}^{m_\phi}\\ \norm{\mathbf{w}}_1 \leq B} } \, \frac{1}{N} \sum_{\ell=1}^N \left| \mathbf{w} \cdot \phi(x_\ell) - y_\ell \right|^2,
\end{equation}
where $\{(x_\ell, y_\ell)\}_{\ell=1}^N$ is the training data.
Here, $x_\ell \in [-1,1]^m$ is an $m$-dimensional vector that parameterizes a Hamiltonian $H(x)$ and $y_\ell$ approximates $\Tr(O\rho(x_\ell))$.
The learned function is given by $h^*(x) = \mathbf{w}^* \cdot \phi(x)$.
The optimization does not have to be solved exactly.
We only need to find a $\mathbf{w}^*$ whose function value is $\mathcal{O}(\epsilon)$ larger than the minimum function value.
There is an extensive literature \cite{efron2004least, daubechies2004iterative, combettes2005signal, cesa2011efficient, friedman2010regularization, hazan2012linear, chen2021quantum} improving the computational time for the above optimization problem.
The best known classical algorithm \cite{hazan2012linear} has a computational time scaling linearly in $m_\phi / \epsilon^2$ up to a log factor, while the best known quantum algorithm \cite{chen2021quantum} has a computational time scaling linearly in $\sqrt{m_\phi} / \epsilon^2$ up to a log factor.

\vspace{2em}
{\renewcommand\addcontentsline[3]{} \subsection{Rigorous guarantee}}

\noindent The classical ML algorithm given above yields the following sample and computational complexity.
This theorem improves substantially upon the result in \cite{huang2021provably}, which requires $N = n^{\mathcal{O}(1 / \epsilon)}$.
The proof idea is given in Section~\ref{sec:proofidea}, and the detailed proof is given in Appendices~\ref{sec:simple},~\ref{sec:norminequality},~\ref{sec:algorithm}.
Using the proof techniques presented in this work, one can show that the sample complexity $N = \log(n / \delta) 2^{\mathrm{polylog}(1 / \epsilon)}$ also applies to any sum of few-body observables $O = \sum_j O_j$ with $\sum_j \norm{O_j}_\infty \leq 1$, even if the operators $\{O_j\}$ are not geometrically local.

\vspace{0.5em}
\begin{theorem}[Sample and computational complexity] \label{thm:rigor-guarantee}
Given $n, \delta > 0$, $\tfrac{1}{e} > \epsilon > 0$ and a training data set $\{x_\ell, y_\ell\}_{\ell = 1}^N$ of size
\begin{equation}
N = \log(n / \delta) 2^{\mathrm{polylog}(1 / \epsilon)},
\end{equation}
where $x_\ell$ is sampled from an unknown distribution $\mathcal{D}$ and $|y_\ell - \Tr(O \rho(x_\ell))| \leq \epsilon$ for any observable~$O$ with eigenvalues between $-1$ and $1$ that can be written as a sum of geometrically local observables.
With a proper choice of the efficiently computable hyperparameters $\delta_1, \delta_2$, and $B$,
the learned function $h^*(x) = \mathbf{w}^* \cdot \phi(x)$ satisfies
\begin{equation}
\E_{x \sim \mathcal{D}} \left| h^*(x) - \Tr(O \rho(x)) \right|^2 \leq \epsilon
\end{equation}
with probability at least $1 - \delta$.
The training and prediction time of the classical ML model are bounded by $\mathcal{O}(n N) = n \log(n / \delta) 2^{\mathrm{polylog}(1 / \epsilon)}$.
\end{theorem}

The output $y_\ell$ in the training data can be obtained by measuring $\Tr(O \rho(x_\ell))$ for the same observable~$O$ multiple times and averaging the outcomes.
Alternatively, we can use the classical shadow formalism~\cite{huang2020predicting, elben2020mixed, elben2022randomized, wan2022matchgate, bu2022classical, van2022hardware} that performs randomized Pauli measurements on $\rho(x_\ell)$ to predict $\Tr(O \rho(x_\ell))$ for a wide range of observables $O$.
Theorem~\ref{thm:rigor-guarantee} and the classical shadow formalism together yield the following corollary for predicting ground state representations.
We present the proof of Corollary~\ref{corollary:ground-state-rep} in Appendix~\ref{sec:rigor-guarantee}.

\begin{corollary}
\label{corollary:ground-state-rep}
Given $n, \delta > 0$, $\tfrac{1}{e} > \epsilon > 0$ and a training data set $\{x_\ell, \sigma_T(\rho(x_\ell))\}_{\ell = 1}^N$ of size
\begin{equation}
N = \log(n / \delta) 2^{\mathrm{polylog}(1 / \epsilon)},
\end{equation}
where $x_\ell$ is sampled from an unknown distribution $\mathcal{D}$ and $\sigma_T(\rho(x_\ell))$ is the classical shadow representation of the ground state $\rho(x_\ell)$ using $T$ randomized Pauli measurements. For $T = \tilde{\mathcal{O}}(\log(n)/\epsilon^2)$, then the proposed ML algorithm can learn a ground state representation $\hat{\rho}_{N, T}(x)$ that achieves
\begin{equation}
\E_{x \sim \mathcal{D}}|\Tr(O\hat{\rho}_{N, T}(x)) - \Tr(O\rho(x))|^2 \leq \epsilon
\end{equation}
for any observable $O$ with eigenvalues between $-1$ and $1$ that can be written as a sum of geometrically local observables with probability at least $1-\delta$.
\end{corollary}

We can also show that the problem of estimating ground state properties for the class of parameterized Hamiltonians $H(x) = \sum_j h_j(\vec{x}_j)$ considered in this work is hard for non-ML algorithms that cannot learn from data.
This is a manifestation of the computational power of data studied in \cite{huang2020power}.
The proof of Proposition~1 in~\cite{huang2021provably} constructs a parameterized Hamiltonian $H(x)$ that belongs to the family of parameterized Hamiltonians considered in this work and hence establishes the following.

\begin{prop}[A variant of Proposition 1 in~\cite{huang2021provably}]
Consider a randomized polynomial-time classical algorithm $\mathcal{A}$ that does not learn from data. Suppose for any smooth family of gapped 2D Hamiltonians $H(x) = \sum_j h_j(\vec{x}_j)$ and any single-qubit observable $O$, $\mathcal{A}$ can compute ground state properties $\Tr(O\rho(x))$ up to a constant error averaged over $x \in [-1, 1]^m$ uniformly.
Then, $\NP$-complete problems can be solved in randomized polynomial time.
\end{prop}

{\renewcommand\addcontentsline[3]{} \section{Proof ideas} \label{sec:proofidea} }

We describe the key ideas behind the proof of Theorem~\ref{thm:rigor-guarantee}.
The proof is separated into three parts.
The first part in Appendix~\ref{sec:simple} describes the existence of a simple functional form that approximates the ground state property $\Tr(O \rho(x))$.
The second part in Appendix~\ref{sec:norminequality} gives a new bound for the $\ell_1$-norm of the Pauli coefficients of the observable $O$ when written in the Pauli basis.
The third part in Appendix~\ref{sec:algorithm} combines the first two parts, using standard tools from learning theory to establish the sample complexity corresponding to the prediction error bound given in Theorem~\ref{thm:rigor-guarantee}.
In the following, we discuss these three parts in detail.

\vspace{2em}
{\renewcommand\addcontentsline[3]{} \subsection{Simple form for ground state property}}

Using the spectral flow formalism \cite{bachmann2012automorphic,hastings2005quasiadiabatic,osborne2007simulating}, we first show that the ground state property can be approximated by a sum of local functions.
First, we write $O$ in the Pauli basis as $O = \sum_{P \in \{I, X, Y, Z\}^{\otimes n}} \alpha_P P$.
Then, we show that for every geometrically local Pauli observable $P$, we can construct a function $f_P(x)$ that depends only on coordinates in the subset $I_P$ of coordinates that parameterizes interaction terms $h_j$ near the Pauli observable $P$.
The function $f_P(x)$ is given by
\begin{equation} \label{eq:fP-def}
f_P(x) = \alpha_P \Tr(P \rho(\chi_P(x))),
\end{equation}
where $\chi_P(x) \in [-1, 1]^m$ is defined as $\chi_P(x)_c = x_c$ for coordinate $c \in I_P$ and $\chi_P(x)_c = 0$ for coordinates $c \not\in I_P$.
The sum of these local functions $f_P$ can be used to approximate the ground state property,
\begin{equation}
    \Tr(O \rho(x)) \approx \sum_{P \in S^{\mathrm{(geo)}}} f_P(x).
\end{equation}
The approximation only incurs an $\mathcal{O}(\epsilon)$ error if we consider $\delta_1 = \Theta(\log^2(1 / \epsilon))$ in the definition of $I_P$.
The key point is that correlations decay exponentially with distance in the ground state of a gapped local Hamiltonian; therefore, the properties of the ground state in a localized region are not sensitive to the details of the Hamiltonian at points far from that localized region.
Furthermore, the local function $f_P$ is smooth.
The smoothness property allows us to approximate each local function $f_P$ by a simple discretization,
\begin{equation}
    f_P(x) \approx \sum_{x' \in X_P} f_P(x') \indicator\left[x \in T_{x', P}\right].
\end{equation}
One could also use other approximations for this step, such as Fourier approximation or polynomial approximation.
For simplicity, we consider a discretization-based approximation with $\delta_2 = \Theta(1 / \epsilon)$
in the definition of $T_{x',P}$ to incur at most an $\mathcal{O}(\epsilon)$ error.
The point is that, for a sufficiently smooth function $f_P(x)$ that depends only on coordinates in $I_P$ and a sufficiently fine lattice over the coordinates in $I_P$, replacing $x$ by the nearest lattice point (based only on coordinates in $I_P$) causes only a small error.
Using the definition of the feature map $\phi(x)$ in Eq.~\eqref{eq:phi-def}, we have
\begin{equation} \label{eq:w-prime-def}
    \Tr(O \rho(x)) \approx \sum_{P \in S^{\mathrm{(geo)}}} \sum_{x' \in X_P} f_P(x') \phi(x)_{x', P} = \mathbf{w}' \cdot \phi(x),
\end{equation}
where $\mathbf{w}'$ is an $m_\phi$-dimensional vector indexed by $x' \in X_P$ and $P \in S^{\mathrm{geo}}$ given by $\mathbf{w}_{x', P}' = f_P(x')$.
The approximation is accurate if we consider $\delta_1 = \Theta(\log^2(1 / \epsilon))$ and $\delta_2 = \Theta(1 / \epsilon)$.
Thus, we can see that the ML algorithm with the proposed feature mapping indeed has the capacity to approximately represent the target function $\Tr(O \rho(x))$.
As a result, we have the following lemma.

\vspace{0.5em}
\begin{lemma}[Training error bound]
\label{lemma:main-training}
The function given by $\mathbf{w}' \cdot \phi(x)$ achieves a small training error:
\begin{equation}
    \frac{1}{N} \sum_{\ell=1}^N \left| \mathbf{w}' \cdot \phi(x_\ell) - y_\ell \right|^2 \leq 0.53 \epsilon.
\end{equation}
\end{lemma}
\noindent This lemma follows from the two facts that $\mathbf{w}' \cdot \phi(x) \approx \Tr(O \rho(x))$ and $\Tr(O \rho(x_\ell)) \approx y_\ell$.

{\renewcommand\addcontentsline[3]{} \subsection{Norm inequality for observables}}

The efficiency of an $\ell_1$-regularized regression depends greatly on the $\ell_1$ norm of the vector $\mathbf{w}'$.
Moreover, the $\ell_1$-norm of $\mathbf{w}'$ is closely related to the observable $O = \sum_j O_j$ given as a sum of geometrically local observables with $\norm{O}_\infty \leq 1$.
In particular, again writing $O$ in the Pauli basis as $O = \sum_{Q \in \{I, X, Y, Z\}^{\otimes n}} \alpha_Q Q$, the $\ell_1$-norm $\norm{\mathbf{w}'}_1$ is closely related to
$\sum_{Q} \left| \alpha_Q \right|,$
which we refer to as the Pauli $1$-norm of the observable $O$.
While it is well known that
\begin{equation}
    \sum_{Q} \left| \alpha_Q \right|^2 = \Tr(O^2) / 2^n \leq \norm{O}_\infty^2,
\end{equation}
there do not seem to be many known results characterizing $\sum_{Q} \left| \alpha_Q \right|$.
To understand the Pauli $1$-norm,
we prove the following theorem.

\vspace{0.5em}
\begin{theorem}[Pauli $1$-norm bound]
\label{thm:main-normineq}
Let $O  = \sum_{Q \in \{I, X, Y, Z\}^{\otimes n}} \alpha_Q Q$ be an observable that can be written as a sum of geometrically local observables. We have,
\begin{equation}
    \sum_Q |\alpha_Q| \leq C \norm{O}_\infty,
\end{equation}
for some constant $C$.
\end{theorem}
\noindent A series of related norm inequalities are also established in \cite{huang2022learning}.
However, the techniques used in this work differ significantly from those in~\cite{huang2022learning}.

{\renewcommand\addcontentsline[3]{}\subsection{Prediction error bound for the ML algorithm}}

\noindent Using the construction of the local function $f_P(x_{c}, c \in I_P)$ given in Eq.~\eqref{eq:fP-def} and the vector $\mathbf{w}'$ defined in Eq.~\eqref{eq:w-prime-def}, we can show that
\begin{equation}
    \norm{\mathbf{w}'}_1 \leq \max_{P \in S^{\mathrm{(geo)}}} \left| X_P \right| \left( \sum_{Q} \left| \alpha_Q \right| \right) \leq \left(1 + \frac{2}{\delta_2}\right)^{\mathrm{poly}(\delta_1)} \left( \sum_{Q} \left| \alpha_Q \right| \right).
\end{equation}
The second inequality follows by bounding the size of our discrete subset $X_P$ and noticing that $|I_P| = \mathrm{poly}(\delta_1)$.
The norm inequality in Theorem~\ref{thm:main-normineq} then implies
\begin{equation}
\norm{\mathbf{w}'}_1 \leq C \norm{O}_\infty \left(1 + \frac{2}{\delta_2}\right)^{\mathrm{poly}(\delta_1)}  \leq 2^{\mathrm{poly} \log(1 / \epsilon)},
\end{equation}
because $\norm{O}_\infty \leq 1$ and $\delta_1 = \Theta(\log^2(1 / \epsilon)), \delta_2 = \Theta(1 / \epsilon)$.
This shows that there exists a vector $\mathbf{w}'$ that has a bounded $\ell_1$-norm and achieves a small training error.
The existence of $\mathbf{w}'$ guarantees that the vector $\mathbf{w}^*$ found by the optimization problem with the hyperparameter $B \geq \norm{\mathbf{w}'}_1$ will yield an even smaller training error.
Using the norm bound on $\mathbf{w}'$, we can choose the hyperparameter $B$ to be $B = 2^{\mathrm{poly} \log(1/\epsilon)}$.
Using standard learning theory \cite{tibshirani1996regression, mohri2018foundations}, we can thus obtain
\begin{equation}
    \E_{x \sim \mathcal{D}} \left| h^*(x) - \Tr(O \rho(x)) \right|^2 \leq \frac{1}{N} \sum_{\ell=1}^N \left| \mathbf{w}^* \cdot \phi(x_\ell) - y_\ell \right|^2 + \mathcal{O}\left( B \sqrt{\frac{\log(m_{\phi} / \delta)}{N}} \right)
\end{equation}
with probability at least $1 - \delta$.
The first term is the training error for $\mathbf{w}^*$, which is smaller than the training error of $0.53 \epsilon$ for $\mathbf{w}'$ from Lemma~\ref{lemma:main-training}.
Thus, the first term is bounded by $0.53 \epsilon$.
The second term is determined by $B$ and $m_\phi$, where we know that $m_\phi \leq |S^{\mathrm{(geo)}}| \left(1 + \frac{2}{\delta_2}\right)^{\mathrm{poly}(\delta_1)}$ and $|S^{\mathrm{(geo)}}| = \mathcal{O}(n)$.
Hence, with a training data size of
\begin{equation}
N = \mathcal{O}\left( \log(n / \delta) 2^{\mathrm{polylog}(1 / \epsilon)} \right),
\end{equation}
we can achieve a prediction error of $\epsilon$ with probability at least $1 - \delta$ for any distribution $\mathcal{D}$ over $[-1, 1]^m$.

\vspace{2em}
{\renewcommand\addcontentsline[3]{} \section{Numerical experiments}}

\begin{figure}[t]
    \centering
    \includegraphics[width=1.0\linewidth]{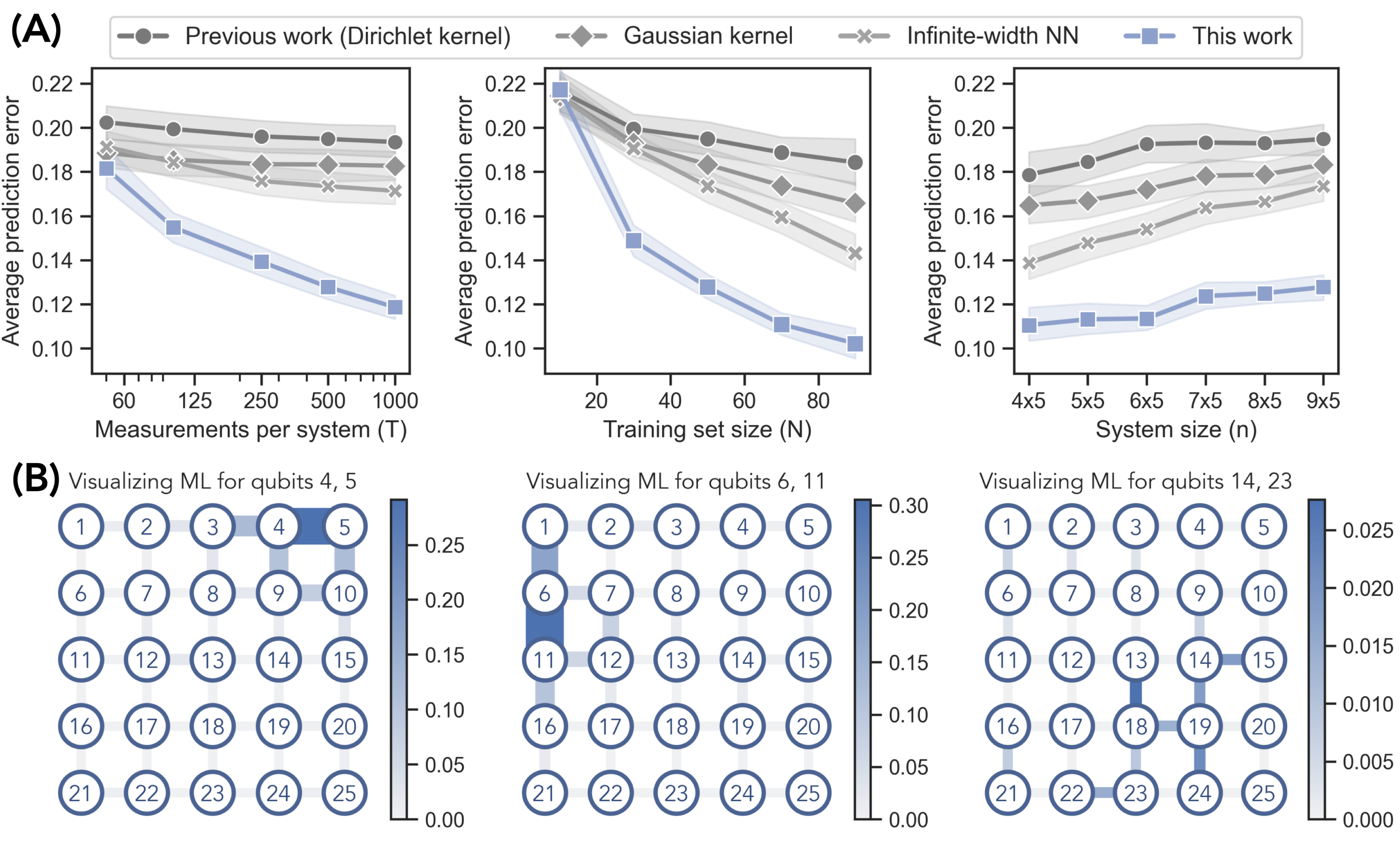}
    \caption{\textbf{Predicting ground state properties in 2D antiferromagnetic random Heisenberg models. (A)} Prediction error. Each point indicates the root-mean-square error for predicting the correlation function in the ground state (averaged over Heisenberg model instances and each pair of neighboring spins).
    Left figure fixes the training set size $N$ to be $50$ and system size $n$ to be $9 \times 5 = 45$. Center figure fixes the shadow size $T$ to be $500$ and $n = 45$. Right figure fixes $N = 50$ and $T = 500$.
    The shaded regions show the standard deviation over different spin pairs.
    \textbf{(B)} Visualization. We plot how much each coupling $J_{ij}$ contributes to the prediction of the correlation function over different pairs of qubits in the trained ML model. Thicker and darker edges correspond to higher contributions. We see that the ML model learns to utilize the local geometric structure. }
    \label{fig:heisenberg}
\end{figure}

In this section, we present numerical experiments to assess the performance of the classical ML algorithm in practice.
The results illustrate the improvement of the algorithm presented in this work compared to those considered in~\cite{huang2021provably}, the mild dependence of the sample complexity on the system size~$n$, and the inherent geometry exploited by the ML models.
We consider the classical ML models described in Section~\ref{subsec:ML-model}, utilizing a random Fourier feature map~\cite{rahimi2007random}.
While the indicator function feature map was a useful tool to obtain our rigorous guarantees, random Fourier features are more robust and commonly used in practice.
Furthermore, we determine the optimal hyperparameters using cross-validation to minimize the root-mean-square error (RMSE) and then evaluate the performance of the chosen ML model using a test set.
The models and hyperparameters are further detailed in Appendix~\ref{sec:numerics}.

For these experiments, we consider the two-dimensional antiferromagnetic random Heisenberg model consisting of $4\times 5 = 20$ to $9\times 5 = 45$ spins.
In this setting, the spins are placed on sites in a 2D lattice.
The Hamiltonian is
\begin{equation}
H = \sum_{\langle ij \rangle} J_{ij} (X_i X_j + Y_i Y_j + Z_i Z_j),
\end{equation}
where the summation ranges over all pairs $\langle ij \rangle$ of neighboring sites on the lattice and the couplings $\{J_{ij}\}$ are sampled uniformly from the interval $[0,2]$.
Here, the vector $x$ is a list of all couplings $J_{ij}$ so that the dimension of the parameter space is $m = O(n)$, where $n$ is the system size.

We trained a classical ML model using randomly chosen values of the parameter vector $x = \{J_{ij}\}$.
For each parameter vector of random couplings sampled uniformly from $[0,2]$, we approximated the ground state using the same method as in~\cite{huang2021provably}, namely with the density-matrix renormalization group (DMRG)~\cite{white1992density} based on matrix product states (MPS)~\cite{SCHOLLWOCK201196}.
The classical ML model was trained on a data set $\{x_\ell, \sigma_T(\rho(x_\ell))\}_{\ell = 1}^N$ with $N$ randomly chosen vectors $x$, where each $x$ corresponds to a classical representation $\sigma_T(\rho(x_\ell))$ created from $T$ randomized Pauli measurements \cite{huang2020predicting}.
The ML algorithm predicted the classical representation of the ground state for a new vector $x$. These predicted classical representations were used to estimate two-body correlation functions, i.e., the expectation value of
\begin{equation}
C_{ij} = \frac{1}{3}(X_i X_j + Y_i Y_j + Z_i Z_j),
\end{equation}
for each pair of qubits $\langle ij \rangle$ on the lattice.

In Figure~\ref{fig:heisenberg}A, we can clearly see that the ML algorithm proposed in this work consistently outperforms the ML models implemented in~\cite{huang2021provably}, which includes the rigorous polynomial-time learning algorithm based on Dirichlet kernel proposed in \cite{huang2021provably}, Gaussian kernel regression \cite{cortes1995support, murphy2012machine}, and infinite-width neural networks \cite{jacot2018neural, neuraltangents2020}.
Figure~\ref{fig:heisenberg}A (Left) and \ref{fig:heisenberg}A (Center) show that as the number $T$ of measurements per data point or the training set size $N$ increases, the prediction performance of the proposed ML algorithm improves faster than the other ML algorithms.
This observation reflects the improvement in the sample complexity dependence on prediction error $\epsilon$. The sample complexity in \cite{huang2021provably} depends exponentially on $1 / \epsilon$, but Theorem~\ref{thm:rigor-guarantee} establishes a quasi-polynomial dependence on $1 / \epsilon$.
From Figure~\ref{fig:heisenberg}A (Right), we can see that the ML algorithms do not yield a substantially worse prediction error as the system size $n$ increases.
This observation matches with the $\log(n)$ sample complexity in Theorem~\ref{thm:rigor-guarantee}, but not with the $\mathrm{poly}(n)$ sample complexity proven in \cite{huang2021provably}.

An important step for establishing the improved sample complexity in Theorem~\ref{thm:rigor-guarantee} is that a property on a local region $R$ of the quantum system only depends on parameters in the neighborhood of region $R$.
In Figure~\ref{fig:heisenberg}B, we visualize where the trained ML model is focusing on when predicting the correlation function over a pair of qubits.
A thicker and darker edge is considered to be more important by the trained ML model.
Each edge of the 2D lattice corresponds to a coupling $J_{ij}$.
For each edge, we sum the absolute values of the coefficients in the ML model that correspond to a feature that depends on the coupling $J_{ij}$.
We can see that the ML model learns to focus only on the neighborhood of a local region $R$ when predicting the ground state property.

\vspace{2em}
{\renewcommand\addcontentsline[3]{}\section{Outlook}}

The classical ML algorithm and the advantage over non-ML algorithms as proven in \cite{huang2021provably} illustrate the potential of using ML algorithms to solve challenging quantum many-body problems.
However, the classical ML model given in~\cite{huang2021provably} requires a large amount of training data.
Although the need for a large dataset is a common trait in contemporary ML algorithms \cite{brown2020language, deng2009imagenet, saharia2022photorealistic},
one would have to perform an equally large number of physical experiments to obtain such data.
This makes the advantage of ML over non-ML algorithms challenging to realize in practice.
The sample complexity $N = \mathcal{O}(\log n)$ of the ML algorithm proposed here illustrates that this advantage could potentially be realized after training with data from a small number of physical experiments.
The existence of a theoretically backed ML algorithm with a $\log(n)$ sample complexity raises the hope of designing good ML algorithms to address practical problems in quantum physics, chemistry, and materials science by learning from the relatively small amount of data that we can gather from real-world experiments.

Despite the progress in this work, many questions remain to be answered.
Recently, powerful machine learning models such as graph neural networks have been used to empirically demonstrate a favorable sample complexity when leveraging the local structure of Hamiltonians in the 2D random Heisenberg model~\cite{wang2022predicting, tran2022shadows}.
Is it possible to obtain rigorous theoretical guarantees for the sample complexity of neural-network-based ML algorithms for predicting ground state properties?
An alternative direction is to notice that the current results have an exponential scaling in the inverse of the spectral gap. Is the exponential scaling a fundamental nature of this problem? Or do there exist more efficient ML models that can efficiently predict ground state properties for gapless Hamiltonians?

We have focused on the task of predicting local observables in the ground state, but many other physical properties are also of high interest.
Can ML models predict low-energy excited state properties? Could we achieve a sample complexity of $N = \mathcal{O}(\log n)$ for predicting any observable~$O$?
Another important question is whether there is a provable quantum advantage in predicting ground state properties.
Could we design quantum ML algorithms that can predict ground state properties by learning from far fewer experiments than any classical ML algorithm? Perhaps this could be shown by combining ideas from adiabatic quantum computation \cite{farhi2000quantum,mizel2007simple,childs2001robustness,aharonov2008adiabatic,barends2016digitized,albash2018adiabatic,du2010nmr,wan2020fast} and recent techniques for proving quantum advantages in learning from experiments \cite{aharonov2021quantum, chen2022exponential, huang2022foundations, huang2021information, huang2022quantum}.
It remains to be seen if quantum computers could provide an unconditional super-polynomial advantage over classical computers in predicting ground state properties.

\vspace{2em}
{\renewcommand\addcontentsline[3]{} \subsection*{Acknowledgments:}}
{ The authors thank Chi-Fang Chen, Sitan Chen, Johannes Jakob Meyer, and Spiros Michalakis for valuable input and inspiring discussions.
We thank Emilio Onorati, Cambyse Rouz\'e, Daniel Stilck Fran\c ca, and James D. Watson for sharing a draft of their new results \cite{onorati2023learning} on efficiently predicting properties of states in thermal phases of matter with exponential decay of correlation and in quantum phases of matter with local topological quantum order.
LL is supported by Caltech Summer Undergraduate Research Fellowship (SURF), Barry M. Goldwater Scholarship, and Mellon Mays Undergraduate Fellowship.
HH is supported by a Google PhD fellowship and a MediaTek Research Young Scholarship. JP acknowledges funding from  the U.S. Department of Energy Office of Science, Office of Advanced Scientific Computing Research, (DE-NA0003525, DE-SC0020290), and the National Science Foundation (PHY-1733907). The Institute for Quantum Information and Matter is an NSF Physics Frontiers Center.  }

\vspace{2em}
\bibliographystyle{unsrt}
{\renewcommand\addcontentsline[3]{} \bibliography{refs}}

\newpage
\appendix
\appendixpage

\tableofcontents

\vspace{1em}

These appendices provide detailed proofs of the statements in the main text. We discuss our main contribution that $\Tr(O\rho)$ can be approximated by a machine learning model given training data scaling logarithmically in system size, where $O$ is an unknown observable and $\rho$ is the ground state of a Hamiltonian. The proof of this result has three main parts. The first two parts yield important results necessary for the design of the ML algorithm and its sample complexity.

We recommend that readers start with Appendix~\ref{sec:simple}, which derives a simpler form for the ground state property $\Tr(O\rho(x))$ that we wish to predict. In Appendix~\ref{sec:norminequality}, we give a norm inequality characterizing the Pauli coefficients of any observable that can be written as a sum of geometrically local observables.
The norm inequality reveals a structure of the ground state property $\Tr(O \rho(x))$ that we can use to design an ML algorithm that uses very few training data.
In Appendix~\ref{sec:algorithm}, we present our ML algorithm and prove its sample complexity using standard tools in ML theory, including known guarantees on the LASSO (least absolute shrinkage and selection operator) algorithm's performance.
Finally, in Appendix~\ref{sec:numerics}, we describe numerical experiments performed to assess the performance of the algorithm in practice.

\section{Simple form for ground state property}
\label{sec:simple}

This section is dedicated to deriving a simpler form for the ground state property $\Tr(O\rho(x))$ as a function of $x$. We consider the assumptions~(a)-(d) from Appendix F.5 of~\cite{huang2021provably}, with (b) and (d) adjusted for our setting, which we reproduce here for convenience:
\begin{enumerate}[(a)]
\item \textit{Physical system:} \label{assum:a} We consider $n$ finite-dimensional quantum
systems that are arranged at locations, or sites, in a $d$-dimensional space, e.g., a spin chain ($d=1$), a square lattice ($d = 2$), or a cubic lattice ($d=3$). Unless specified otherwise, our big-$\mathcal{O},\Omega,\Theta$ notation is with respect to the thermodynamic limit $n \to \infty$.
\item \textit{Hamiltonian:} \label{assum:b}$H(x)$ decomposes into a sum of geometrically local terms $H(x) = \sum_{j=1}^L h_j(\vec{x}_j)$, each of which only acts on an $\mathcal{O}(1)$ number of sites in a ball of $\mathcal{O}(1)$ radius. Here, $\vec{x}_j \in \mathbb{R}^{q}, q = \mathcal{O}(1)$ and $x$ is the concatenation of $L$ vectors $\vec{x}_1,\dots, \vec{x}_L$ with dimension $m = Lq = \mathcal{O}(n)$.
Individual terms $h_j(\vec{x}_j)$ obey $\norm{h_j(\vec{x}_j)}_\infty \leq 1$ and also have bounded directional derivative: $\norm{\partial h_j/\partial \hat{u}}_\infty \leq 1$, where $\hat{u}$ is a unit vector in parameter space.
\item \textit{Ground-state subspace:} \label{assum:c} We consider the ground state $\rho(x)$ for the Hamiltonian $H(x)$ to be defined as $\rho(x) = \lim_{\beta \to \infty} e^{-\beta H(x)}/\Tr(e^{-\beta H(x)})$. This is equivalent to a uniform mixture over the eigenspace of $H(x)$ with the minimum eigenvalue.
\item \textit{Observable:} \label{assum:d} $O$ can be written as a sum of few-body observables $O = \sum_{j} O_j$, where each $O_j$ only acts on an $\mathcal{O}(1)$ number of sites. Hence, we can also write $O = \sum_{P \in S^{\mathrm{(geo)}}} \alpha_P P$, where $P \in \{I, X, Y, Z\}^{\otimes n}$ and $S^{\mathrm{(geo)}}$ is the set of geometrically local Pauli observables (defined more precisely in Def.~\ref{def:setglpauli}).
The results in this section hold for any $O$ of the above form. However, we only focus on $O$ given as a sum of geometrically local observables $\sum_{j} O_j$, where each $O_j$ only acts on an $\mathcal{O}(1)$ number of sites in a ball of $\mathcal{O}(1)$ radius.
\end{enumerate}
Under these assumptions, we can prove that $\Tr(O\rho(x))$ can be approximated by a sum of weighted indicator functions, where the weights satisfy a $\ell_1$-norm bound.
A precise statement of this result is found in Appendix~\ref{sec:simpleformal}. We first show that $\Tr(O\rho(x))$ can be approximated by a sum of smooth local functions in Appendix~\ref{sec:approxsmooth}. Then, we prove that this sum of smooth local functions can be approximated by simple functions in Appendix~\ref{sec:discretization}. Finally, we put everything together in Appendix~\ref{sec:simpleformal}. Several technical lemmas for bounding integrals are needed throughout these proofs, which are compiled in Appendix~\ref{sec:techlemmas}.

\subsection{Approximation by a sum of smooth functions}
\label{sec:approxsmooth}

The key intermediate step is to approximate $\Tr(O\rho(x))$ by a sum of smooth local functions. The proof of this relies on the spectral flow formalism~\cite{bachmann2012automorphic} and Lieb-Robinson bounds~\cite{lieb1972finite}.

First, we review the tools necessary from spectral flow~\cite{bachmann2012automorphic,hastings2005quasiadiabatic,osborne2007simulating}. Let the spectral gap of $H(x)$ be lower bounded by a constant $\gamma$ over $[-1,1]^m$. Then, the directional derivative of an associated ground state in the direction defined by the parameter unit vector $\hat{u}$ is given by
\begin{equation}
\label{eq:spectralflow}
\frac{\partial \rho}{\partial \hat{u}}(x) = -i [D_{\hat{u}}(x), \rho(x)],
\end{equation}
where $D_{\hat{u}}(x)$ is given by
\begin{equation}
\label{eq:spectralflow2}
D_{\hat{u}}(x) = \int_{-\infty}^{+\infty} W_\gamma(t) e^{itH(x)} \frac{\partial H}{\partial \hat{u}}(x) e^{-itH(x)}\,dt.
\end{equation}
Here, $W_\gamma(t)$ is defined by
\begin{equation}
\label{eq:wgamma}
|W_\gamma(t)| \leq \begin{cases}
\frac{1}{2} & 0 \leq \gamma|t| \leq \theta,\\
35e^2 (\gamma|t|)^4 e^{-\frac{2}{7}\frac{\gamma|t|}{\log(\gamma|t|)^2}} & \gamma|t| > \theta,
\end{cases}
\end{equation}
where $\theta$ is chosen to be the largest real solution of $35e^2\theta^4 \exp(-\frac{2}{7}\frac{\theta}{\log^2(\theta)}) = 1/2$. Notice that $W_\gamma(t)$ has the property that $\sup_t |W_\gamma(t)| = 1/2$.

Next, we review the Lieb-Robinson bounds~\cite{hastings2010locality,lieb1972finite}.
Let the distance $d_{\mathrm{obs}}(X_1, X_2)$ between any two operators $X_1, X_2$ be defined as the minimum distance between all pairs of sites acted on by $X_1$ and $X_2$, respectively, in the $d$-dimensional space.
Formally, this is defined as
\begin{equation}
\label{eq:dinter}
d_{\mathrm{obs}}(X_1, X_2) \triangleq \min_{\substack{i \in \mathsf{dom}(X_1)\\ i' \in \mathsf{dom}(X_2)}} d_{\mathrm{qubit}}(i, i'),
\end{equation}
where $\mathsf{dom}(O)$ contains the qubits that the observable $O$ acts on and $d_{\mathrm{qubit}}(i, i')$ is the distance between two qubits $i$ and $i'$.
Furthermore, notice that for any operator $X$ acting on a single site, a ball of radius $r$ around $X$ contains $\mathcal{O}(r^d)$ local terms in $d$-dimensional space:
\begin{equation}
\label{eq:localtermsbound}
\sum_{j : d_{\mathrm{obs}}(X, h_j)\leq r} 1 \leq b_d + c_d r^d,
\end{equation}
where $h_j$ is an interaction term of the Hamiltonian $H = \sum_{j=1}^L h_j$. Here, this bound implies the existence of a Lieb-Robinson bound~\cite{bravyi2006lieb,hastings2010locality} such that for any two operators $X_1,X_2$ and any $t \in \mathbb{R}$,
\begin{align}
\begin{split}
&\norm{[\exp(itH(x)) X_1 \exp(-it H(x)), X_2]}_\infty \\
&\leq c_{\mathrm{lr}} \norm{X_1}_\infty \norm{X_2}_\infty |\mathsf{dom}(X_1)| \exp(-a_{\mathrm{lr}}(d_{\mathrm{obs}}(X_1, X_2) - v_{\mathrm{lr}}|t|)), \label{eq:lrbound}
\end{split}
\end{align}
where $a_{\mathrm{lr}}, c_{\mathrm{lr}}, v_{\mathrm{lr}} = \Theta(1)$ are constants. Having reviewed these tools, before stating our result formally, we need to define a quantity that we use throughout the proof.

\begin{definition}
\label{def:delta}
Let $1/e > \epsilon > 0$. Consider a family of Hamiltonians $\{H(x) : x \in [-1,1]^m\}$ in a $d$-dimensional space. Suppose that the spectral gap of $H(x)$ is lower bounded by a constant $\gamma$ over $[-1,1]^m$. Define $\delta_1$ as
\begin{equation}
\delta_1 \triangleq \max\left(C_{\mathrm{max}}\log^2(1/\epsilon), C_4, C_5, \frac{\max(5900, \alpha, 7(d + 11), \theta)}{b}\right),
\end{equation}
where we denote $b \triangleq \gamma/2v_{\mathrm{lr}}$ for convenience, and $v_{\mathrm{lr}}$ is the constant from the Lieb-Robinson bound in Eq.~\eqref{eq:lrbound}.
Here, $C_{\mathrm{max}} = \max(C_1, C_2, C_3)$, where $C_1, C_2, C_3$ are constants defined in Lemmas~\ref{lemma:c1},~\ref{lemma:c2},~\ref{lemma:c3}. Also, we define $C_4$ as a constant such that for all $\delta' \geq C_4$,
\begin{equation}
\label{eq:c4}
\frac{1}{1 - \frac{77\log^2(b(\delta' + 1))}{b(\delta' + 1)}} \leq 2.
\end{equation}
Similarly, define $C_5$ as a constant such that for all $\delta' \geq C_5$,
\begin{equation}
\label{eq:c5}
\frac{1}{1 - \frac{7(2d + 22)\log^2(b(\delta' + 1))}{2b(\delta' + 1)}} \leq 2.
\end{equation}
Moreover, $\alpha$ is defined such that for all $x \geq b(\alpha + 1)$, $35\log^2 x < x - b$.
Finally, $\theta$ is chosen to be the largest real solution of
\begin{equation}
35e^2\theta^4\exp\left(-\frac{2}{7}\frac{\theta}{\log^2(\theta)}\right) = \frac{1}{2}.
\end{equation}
\end{definition}

The existence of $C_4, C_5$ is guaranteed by noting that as $\delta'$ goes to infinity, the inequalities become less than or equal to $2$. Similarly, the existence of $\alpha$ is guaranteed by considering $x \rightarrow \infty$. Using the quantity $\delta_1$ defined above, we also define the parameters ``close to'' a given Pauli term $P$.

\begin{definition}
\label{def:ip}
Given $\delta_1$ from Definition~\ref{def:delta} and an observable $O = \sum_{P \in S^{\mathrm{(geo)}}} \alpha_P P$, for each Pauli term $P \in S^{\mathrm{(geo)}}$, we define
\begin{equation}
I_P \triangleq \left\{c \in \{1,\dots, m\}: d_{\mathrm{obs}}(h_{j(c)}, P) \leq \delta_1 \right\},
\end{equation}
as in Eq.~\eqref{eq:ip}.
\end{definition}

Now, we are ready to present the precise statement that the ground state property $\Tr(O\rho(x))$ can be approximated by a sum of smooth local functions. First, we consider the simpler case where our observable $O = \alpha_P P$ is a single Pauli term, which easily generalizes to the general case via triangle inequality.

\begin{lemma}[Approximation using smooth local functions; simple case]
\label{lemma:approxlocal}
Consider a class of local Hamiltonians $\{H(x) : x \in [-1,1]^m\}$ satisfying assumptions~\ref{assum:a}-\ref{assum:c}, and an observable $O = \alpha_P P$, where $P$ acts on at most $\mathcal{O}(1)$ qubits. Then, there exists a constant $C > 0$ such that for any $1/e > \epsilon > 0$,
\begin{equation}
|\alpha_P \Tr(P \rho(x)) - f_P(x)| \leq C |\alpha_P|\epsilon,
\end{equation}
where $f_P(x) \triangleq \alpha_P \Tr(P \rho(\chi_P(x)))$ is a smooth function that only depends on parameters $x_{c} \in [-1, 1]$ for coordinates $c \in I_P$, the restriction function $\chi_P: [-1, 1]^m \mapsto [-1, 1]^m$ is defined as
\begin{equation} \label{eq:restriction-chi}
    \chi_P(x)_c = \begin{cases}
        x_c, & c \in I_P,\\
        0, & c \not\in I_P,
    \end{cases}
    \quad \forall c \in \{1, \ldots, m\},
\end{equation}
and the set $I_P$ of coordinates is given in Definition~\ref{def:ip}.
The function $f_P(x)$ is smooth in the sense that
\begin{equation}
\norm{\nabla_x f_P(x)}_2^2 \leq |\alpha_P|^2 C'
\end{equation}
for some constant $C' > 0$.
\end{lemma}

\begin{corollary}[Approximation using smooth local functions; general case]
\label{corollary:approxlocal}
Consider a class of local Hamiltonians $\{H(x) : x \in [-1,1]^m\}$ and an observable $O = \sum_{P \in \{I, X, Y, Z\}^{\otimes n}} \alpha_P P$ satisfying assumptions~\ref{assum:a}-\ref{assum:d}.
There exists a constant $C > 0$ such that for any $1/e > \epsilon > 0$,
\begin{equation}
|\Tr(O \rho(x)) - f(x)| \leq C\epsilon \left(\sum_{P} |\alpha_P| \right),
\end{equation}
where $f(x) = \sum_{P \in S^{(\mathrm{geo})}} f_P(x)$ for $f_P(x)$ given in Lemma~\ref{lemma:approxlocal}.
\end{corollary}

\begin{figure}[t]
\centering
\includegraphics[width=0.9\linewidth]{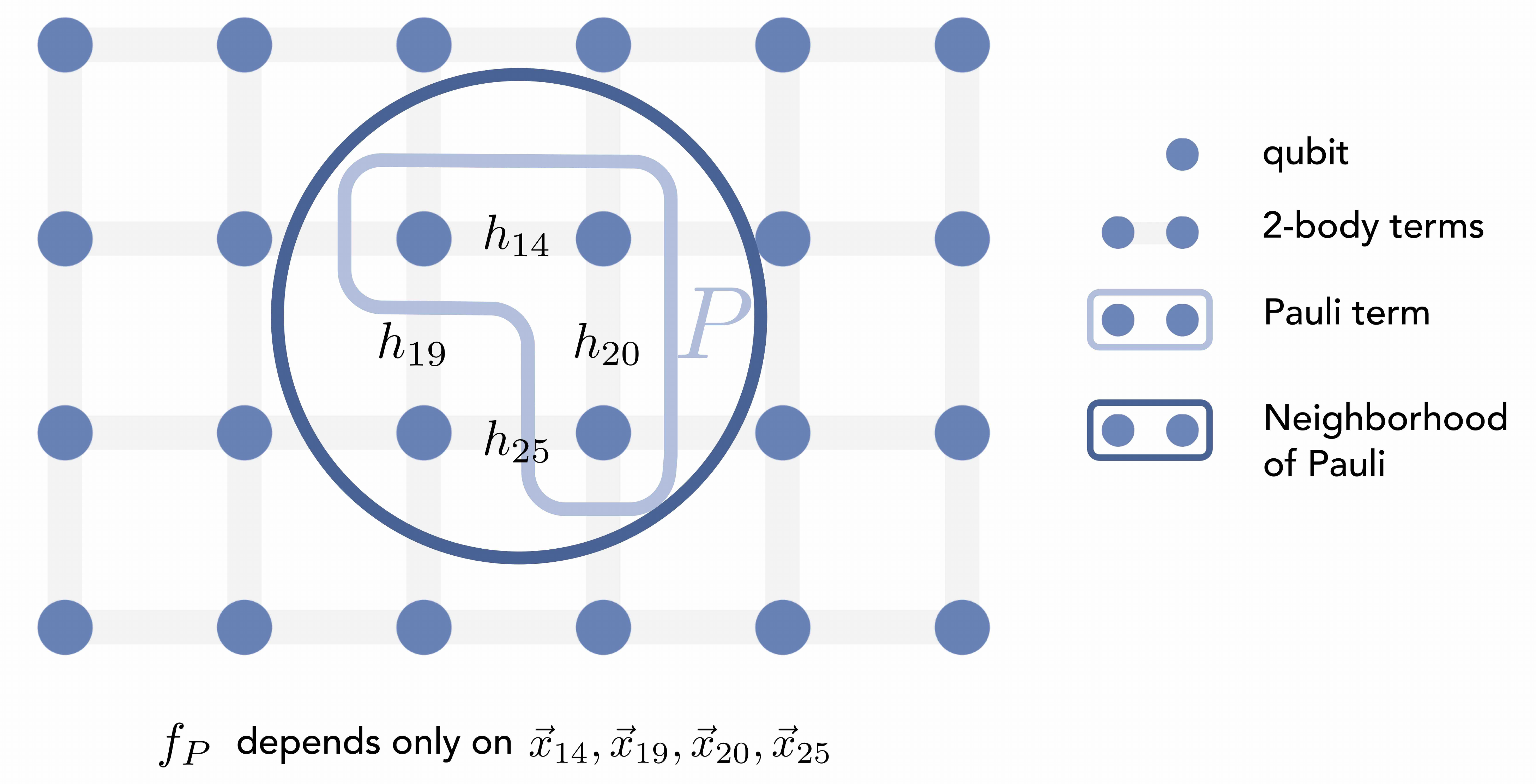}
\caption{\textbf{Intuition behind Lemma~\ref{lemma:approxlocal}}. The qubits (blue circles) are arranged in a two-dimensional lattice with local Hamiltonian terms (light gray shading) acting between all pairs of neighboring qubits. A Pauli term $P$ acts on a subset of these qubits indicated by the light blue region. The dark blue circle represents a neighborhood around the region on which $P$ acts. The idea of Lemma~\ref{lemma:approxlocal} is that when changing the parameters $x$, only $\vec{x}_j$ such that $h_j(\vec{x}_j)$ within the neighborhood around the region that $P$ acts on should significantly change $\Tr(P\rho(x))$. Hence, $f_P$ depends only on those parameters. It is implicit in the figure that $h_j$ depends on $\vec{x}_j$ for all $j$. Hence, $f_P$ depends only on the vectors $\vec{x}_{14}, \vec{x}_{19}, \vec{x}_{20}, \vec{x}_{25}$.}
\label{fig:approxlocal}
\end{figure}

We illustrate the intuition for Lemma~\ref{lemma:approxlocal} in Figure~\ref{fig:approxlocal}.
The proof of Lemma~\ref{lemma:approxlocal} requires several steps.
The main idea is that the function $f_P(x)$ is simply $\alpha_P \Tr(P \rho(\chi_P(x)))$ such that $\chi_P(x)_c = x_c$ for $c \in I_P$ and $\chi_P(x)_c = 0$ for coordinates $c \notin I_P$.
Thus, we need to show that changing coordinates outside of $I_P$ does not change $\alpha_P \Tr(P \rho(x))$ by much.
First, we change one coordinate outside of $I_P$ at a time and show that the directional derivative of $\alpha_P \Tr(P \rho(x))$ in the direction changing this coordinate is bounded.
Next, we use this to prove that $|\alpha_P \Tr(P \rho(x)) - \alpha_P \Tr(P \rho(x'))|$ is bounded, where $x$ and $x'$ differ in this one coordinate.
Finally, we show that the difference is bounded for the case where $x$ and $x'$ differ for all coordinates outside of $I_P$, which concludes the proof of Lemma~\ref{lemma:approxlocal}.
We separate these results into lemmas.
Throughout the proofs of these lemmas, we also need several technical lemmas for showing the existence of certain constants and bounding integrals, proofs of which we relegate to Appendix~\ref{sec:techlemmas}.
In the rest of this section, and in Appendix~\ref{sec:techlemmas}, we use the notation $b \triangleq \gamma/(2v_{\mathrm{lr}})$ and $\Delta(j, P) \triangleq d_{\mathrm{obs}}(h_{j(c)}, P)$ for convenience.

\begin{lemma}[Change one coordinate; directional derivative]
\label{lemma:oneparamderiv}
Consider a class of local Hamiltonians $\{H(x) : x \in [-1,1]^m\}$ satisfying assumptions~\ref{assum:a}-\ref{assum:c}, and an observable $O = \alpha_P P$, where $P$ acts on at most $\mathcal{O}(1)$ qubits. Suppose that some $x, x' \in [-1,1]^m$ only differ in one coordinate, say the coordinate $c^*$ such that $c^* \notin I_P$ and only one $h_j$ depends on $x_{c^*}$. Let $\hat{u}$ be a unit vector in the direction that moves from $x$ to $x'$ along the $c^*$th coordinate. Then, there exist constants $c_1,c_2$ such that
\begin{align}
\begin{split}
   & |\alpha_P|\left|\hat{u} \cdot \nabla_x \Tr(P\rho(x))\right| \\&\leq |\alpha_P|\left(c_1 e^{-\frac{a_{\mathrm{lr}}\Delta(j,P)}{2}} + c_2\left(\frac{1}{1-\frac{35\log^2(b\Delta(j, P))}{b\Delta(j,P)}}\right)\Delta(j,P)^{10}\exp\left(-\frac{2}{7}\frac{b\Delta(j,P)}{\log^2(b\Delta(j,P))}\right)\right).
\end{split}
\end{align}
\end{lemma}

\begin{proof}
For the direction $\hat{u}$, we can write the directional derivative of $\rho(x)$ in two ways. First, we have the standard definition:
\begin{equation}
\frac{\partial \rho}{\partial \hat{u}}(x) = \hat{u} \cdot \nabla_x \rho(x).
\end{equation}
Then, from spectral flow, we also have Eq.~\eqref{eq:spectralflow}.
When evaluated on an observable $O = \alpha_P P$, this establishes the following correspondence:
\begin{equation}
\label{eq:dirderiv}
\alpha_P (\hat{u} \cdot \nabla_x \Tr(P \rho(x))) = i\alpha_P\Tr(P[D_{\hat{u}}(x),\rho(x)]) = i\alpha_P\Tr([P, D_{\hat{u}}(x)]\rho(x)).
\end{equation}
Expanding $D_{\hat{u}}(x)$ according to Eq.~\eqref{eq:spectralflow2} and applying the triangle inequality to
\begin{equation}
H(x) = \sum_{j=1}^L h_j(\vec{x}_j),
\end{equation}
we have
\begin{equation}
|\alpha_P||\Tr([P, D_{\hat{u}}(x)]\rho(x))| \leq |\alpha_P|\int_{-\infty}^{+\infty} W_\gamma(t) \sum_{j=1}^L \norm{\left[P, e^{itH(x)} \frac{\partial h_j}{\partial \hat{u}} e^{-itH(x)}\right]}_\infty\,dt.
\end{equation}
Here, since $x_{c^*}$ only affects $h_j$ for one $j$ and $\hat{u}$ is in the direction where only the coordinate $c^*$ changes, then
\begin{equation}
\frac{\partial h_{j'}}{\partial \hat{u}} = 0
\end{equation}
for all $j' \neq j$. Thus, we are left with
\begin{equation}
\label{eq:f60}
|\alpha_P||\Tr([P, D_{\hat{u}}(x)]\rho(x))| \leq |\alpha_P|\int_{-\infty}^{+\infty} W_{\gamma}(t) \norm{\left[P, e^{itH(x)} \frac{\partial h_j}{\partial \hat{u}} e^{-itH(x)}\right]}_\infty\,dt.
\end{equation}
We bound this integral using Lieb-Robinson bounds and the inequality on $W_\gamma(t)$ that
\begin{equation}
\sup_t |W_\gamma(t)| = 1/2.
\end{equation}
We first need to split the integral into cases.
This is because Lieb-Robinson bounds only apply outside of the lightcone, i.e., when $\Delta(j,P) > v_{\mathrm{lr}}|t|$.
Then, when $\Delta(j,P) \leq v_{\mathrm{lr}}|t|$, we can instead use the commutator norm bound $\norm{[A, B]}_\infty \leq 2\norm{A}_\infty \norm{B}_\infty$.
Thus, we define $t^* = \Delta(j,P)/(2v_{\mathrm{lr}})$ and split up the integration into two parts: $t \in [-t^*, t^*]$ and $t\notin [-t^*, t^*]$ so that we have
\begin{align}
\label{eq:integralsum}
\begin{split}
|\alpha_P||\Tr([P, D_{\hat{u}}(x)]\rho(x))| &\leq |\alpha_P|\int_{-t^*}^{t^*} W_{\gamma}(t) \norm{\left[P, e^{itH(x)} \frac{\partial h_j}{\partial \hat{u}} e^{-itH(x)}\right]}_\infty\,dt\\
&+ |\alpha_P|\int_{t^*}^{+\infty} W_{\gamma}(t) \norm{\left[P, e^{itH(x)} \frac{\partial h_j}{\partial \hat{u}} e^{-itH(x)}\right]}_\infty\,dt\\
&+ |\alpha_P|\int_{-\infty}^{-t^*} W_{\gamma}(t) \norm{\left[P, e^{itH(x)} \frac{\partial h_j}{\partial \hat{u}} e^{-itH(x)}\right]}_\infty\,dt.
\end{split}
\end{align}
Notice that the first integral corresponds to the case when we are outside of the lightcone, i.e., $\Delta(j,P) > 2v_{\mathrm{lr}}|t|$ while the other two integrals correspond to the case when we are inside of the light cone.

First, we bound the first integral using the Lieb-Robinson bound.
Applying Eq.~\eqref{eq:lrbound} to the commutator norm, we have
\begin{subequations}
\begin{align}
    \norm{\left[P, e^{itH(x)}\frac{\partial h_j}{\partial \hat{u}} e^{-itH(x)}\right]}_\infty &\leq c_{\mathrm{lr}}\norm{P}_\infty \norm{\frac{\partial h_j}{\partial \hat{u}}}_\infty |\mathsf{dom}(h_j)| e^{-a_{\mathrm{lr}}(\Delta(j,P) - v_{\mathrm{lr}}|t|)}\\
    &\leq c_{\mathrm{lr}}c_h e^{-a_{\mathrm{lr}}(\Delta(j,P) - v_{\mathrm{lr}}|t|)},
\end{align}
\end{subequations}
where in the last inequality, we are using assumption~\ref{assum:b} that $\norm{\partial h_j/\partial \hat{u}}_\infty \leq 1$ and $|\mathsf{dom}(h_j)| \leq c_h$ for a constant $c_h$.
Plugging this into the integral, we have
\begin{subequations}
\begin{align}
|\alpha_P|\int_{-t^*}^{t^*} W_\gamma(t) \norm{\left[P, e^{itH(x)} \frac{\partial h_j}{\partial \hat{u}}e^{-itH(x)}\right]}_\infty \,dt &\leq |\alpha_P|c_{\mathrm{lr}}c_h e^{-a_{\mathrm{lr}}\Delta(j,P)}\int_{-t^*}^{t^*}|W_\gamma(t)|e^{a_{\mathrm{lr}}v_{\mathrm{lr}}|t|}\,dt\\
&\leq \frac{1}{2}|\alpha_P|c_{\mathrm{lr}}c_h e^{-a_{\mathrm{lr}}\Delta(j,P)} \int_{-t^*}^{t^*} e^{a_{\mathrm{lr}}v_{\mathrm{lr}}|t|}\,dt\\
&= |\alpha_P|c_{\mathrm{lr}}c_h e^{-a_{\mathrm{lr}}\Delta(j,P)} \int_0^{t^*} e^{a_{\mathrm{lr}}v_{\mathrm{lr}}t}\,dt\\
&= |\alpha_P|\frac{c_{\mathrm{lr}}c_h e^{-a_{\mathrm{lr}} \Delta(j,P)}}{a_{\mathrm{lr}}v_{\mathrm{lr}}}\left(e^{a_{\mathrm{lr}}v_{\mathrm{lr}}t^*} - 1\right)\\
&= |\alpha_P|\frac{c_{\mathrm{lr}}c_h}{a_{\mathrm{lr}}v_{\mathrm{lr}}} e^{-a_{\mathrm{lr}}\Delta(j,P)} \left(e^{a_{\mathrm{lr}}\Delta(j,P)/2} - 1\right)\\
&= |\alpha_P|\frac{c_{\mathrm{lr}}c_h}{a_{\mathrm{lr}}v_{\mathrm{lr}}} \left(e^{-a_{\mathrm{lr}} \Delta(j,P)/2} - e^{-a_{\mathrm{lr}} \Delta(j,P)}\right)\\
&\leq |\alpha_P|\frac{c_{\mathrm{lr}}c_h}{a_{\mathrm{lr}}v_{\mathrm{lr}}} e^{-a_{\mathrm{lr}} \Delta(j,P)/2},
\end{align}
\end{subequations}
where in the second line, we used the fact that $\sup_t |W_\gamma(t)| = 1/2$, and in the fifth line, we substituted back in $t^* = \Delta(j,P)/(2v_{\mathrm{lr}})$.

We can also bound the other integrals using the commutator norm bound
\begin{equation}
    \norm{[A, B]}_\infty \leq 2\norm{A}_\infty \norm{B}_\infty
\end{equation}
to obtain:
\begin{subequations}
\begin{align}
|\alpha_P|\int_{t^*}^{+\infty} W_{\gamma}(t) \norm{\left[P, e^{itH(x)} \frac{\partial h_j}{\partial \hat{u}} e^{-itH(x)}\right]}_\infty\,dt &\leq 2|\alpha_P| \int_{t^*}^{+\infty} |W_\gamma(t)| \norm{P}_\infty \norm{\frac{\partial h_j}{\partial \hat{u}}}_\infty \,dt\\
&\leq 2|\alpha_P| \int_{t^*}^\infty |W_\gamma(t)|\,dt,
\end{align}
\end{subequations}
where in the second line, we used assumption~\ref{assum:b} that $\norm{\partial h_j/\partial \hat{u}}_\infty \leq 1$.
To bound the resulting integral, we use the definition of $W_\gamma(t)$ in Eq.~\eqref{eq:wgamma}.
Note that by our definition of $t^*$, $\gamma t^* > \theta$, so we only need to consider this case in the upper bound on $W_\gamma(t)$.
This is because we chose
\begin{equation}
\delta_1 = \max\left(C_{\mathrm{max}} \log^2(1/\epsilon), C_4, C_5, \frac{\max(5900, \alpha, 7(d + 11), \theta)}{b} \right),
\end{equation}
and here we consider $\Delta(j,P) > \delta_1$. Thus, we have
\begin{equation}
\label{eq:gammat}
\gamma t^* = \frac{\gamma \Delta(j,P)}{2v_{\mathrm{lr}}} > \frac{\gamma \delta_1}{2v_{\mathrm{lr}}} \geq \max(5900, \alpha, 7(d + 11), \theta) \geq \theta.
\end{equation}
Hence, we can bound the integral:
\begin{equation}
\int_{t^*}^{+\infty} |W_\gamma(t)|\,dt \leq 35e^2\int_{t^*}^{+\infty}(\gamma t)^4 e^{-\frac{2}{7}\frac{\gamma t}{\log^2(\gamma t)}}\,dt = 35e^2\gamma^{-1}\int_{x=\gamma t^*}^{+\infty} x^4 e^{-\frac{2}{7} \frac{x}{\log^2(x)}}\,dx.
\end{equation}
In the inequality, we used the definition of $W_\gamma(t)$ and in the equality, we used the substitution $x = \gamma t$.
We can bound this integral using Lemma~\ref{lemma:integralbound}. Set $a=2/7$ and $k = 4$.
We have chosen $t^*$ and $\delta_1$ such that all of the assumptions of Lemma~\ref{lemma:integralbound} are satisfied.
In particular, from Eq.~\eqref{eq:gammat}, we see that $t = \gamma t^* > \max(5900, \alpha, 7(d + 11), \theta) \geq 5900$.
Furthermore, we have $ at/\log^2(t) > 2k+2$, because if $t \geq 5900$, then it is clear that $at/\log^2(t) > 10$.
Now, applying Lemma~\ref{lemma:integralbound}, we have
\begin{equation}
\int_{t^*}^{+\infty}|W_\gamma(t)|\,dt \leq \frac{245}{2}e^2\gamma^{-1} \left(\frac{1}{1-\frac{35\log^2(\gamma t^*)}{\gamma t^*}}\right) (\gamma t^*)^{10} e^{-\frac{2}{7}\frac{\gamma t^*}{\log^2(\gamma t^*)}}.
\end{equation}
The last integral can be bounded in exactly the same way.
Plugging these bounds into Eq.~\eqref{eq:integralsum}, we have
\begin{subequations}
\begin{align}
&|\alpha_P||\Tr([P, D_{\hat{u}}(x)]\rho(x))|\\
&\leq |\alpha_P|\left(\frac{c_{\mathrm{lr}}c_h}{a_{\mathrm{lr}}v_{\mathrm{lr}}} e^{-a_{\mathrm{lr}}\Delta(j,P)/2} + 4\left(\frac{245}{2}e^2\gamma^{-1}\left(\frac{1}{1-\frac{35\log^2(\gamma t^*)}{\gamma t^*}}\right) (\gamma t^*)^{10} e^{-\frac{2}{7}\frac{\gamma t^*}{\log^2(\gamma t^*)}}\right)\right)\\
&= |\alpha_P|\left(c_1 e^{-a_{\mathrm{lr}}\Delta(j,P)/2} + c_2\left(\frac{1}{1-\frac{35\log^2(b\Delta(j,P))}{b\Delta(j,P)}}\right)\Delta(j,P)^{10}\exp\left(-\frac{2}{7}\frac{b\Delta(j,P)}{\log^2(b\Delta(j,P))}\right)\right),
\end{align}
\end{subequations}
where in the second line, we defined constants
\begin{equation}
c_1 = \frac{c_{\mathrm{lr}}c_h}{a_{\mathrm{lr}}v_{\mathrm{lr}}},\;\;\; c_2 = \frac{245e^2b^9}{v_{\mathrm{lr}}}.
\end{equation}
Thus, we have proven that if we only change one coordinate outside of $I_P$, then the directional derivative changing this coordinate is small. This is exactly the claim of the lemma.
\end{proof}

\noindent An immediate consequence of this is that we can integrate the directional derivative to obtain a bound on the distance between $\Tr(P\rho(x))$ and $\Tr(P\rho(x'))$.

\begin{lemma}[Change one coordinate; distance]
\label{lemma:oneparamdist}
Consider a class of local Hamiltonians $\{H(x) : x \in [-1,1]^m\}$ satisfying assumptions~\ref{assum:a}-\ref{assum:c}, and an observable $O = \alpha_P P$, where $P$ acts on at most $\mathcal{O}(1)$ qubits. Suppose that some $x, x' \in [-1,1]^m$ only differ in one coordinate, say the coordinate $c^*$ such that $c^* \notin I_P$ and only one $h_j$ depends on $x_{c^*}$. Let $\hat{u}$ be a unit vector in the direction that moves from $x$ to $x'$ along the $c^*$th coordinate. Then, there exist constants $c_1', c_2'$ such that
\begin{align}
\begin{split}
    &|\alpha_P||\Tr(P\rho(x)) - \Tr(P\rho(x'))|\\
    &\leq |\alpha_P|\left(c_1' e^{-\frac{a_{\mathrm{lr}}\Delta(j,P)}{2}} + c_2'\left(\frac{1}{1-\frac{35\log^2(b\Delta(j,P))}{b\Delta(j,P)}}\right)\Delta(j,P)^{10}\exp\left(-\frac{2}{7}\frac{b\Delta(j,P)}{\log^2(b\Delta(j,P))}\right)\right).
\end{split}
\end{align}
\end{lemma}

\begin{proof}
By Lemma~\ref{lemma:oneparamderiv}, we have a bound on the directional derivative of $\alpha_P\Tr(P\rho(x))$ in the direction of $\hat{u}$.
In this lemma, we want a bound on the distance
\begin{equation}
|\alpha_P||\Tr(P\rho(x)) - \Tr(P\rho(x'))| = |\alpha_P||\Tr(P\rho(x_1,\dots, x_{c^*},\dots, x_m)) - \Tr(P\rho(x_1,\dots, x_{c^*}',\dots, x_m))|.
\end{equation}
To this end, we can obtain the distance by integrating the directional derivative:
\begin{subequations}
\begin{align}
&|\alpha_P||\Tr(P\rho(x)) - \Tr(P\rho(x'))|\\
&= |\alpha_P|\left|\int_{x_{c^*}}^{x_{c^*}'} \frac{\partial \Tr(P\rho(x_1,\dots, t, \dots, x_m))}{\partial \hat{u}}\,dt \right|\\
&\leq |\alpha_P| \int_{x_{c^*}}^{x_{c^*}'} \left|\frac{\partial \Tr(P\rho(x_1,\dots, t, \dots, x_m))}{\partial \hat{u}}\right|\,dt\\
&= |\alpha_P|\int_{x_{c^*}}^{x_{c^*}'} |\Tr([P, D_{\hat{u}}(x_1,\dots, t, \dots, x_m)]\rho(x_1,\dots,t,\dots,x_m))|\,dt,
\end{align}
\end{subequations}
where in the last line, we used the correspondence from Eq.~\eqref{eq:dirderiv}.
Now, the integrand is exactly what we bounded in Lemma~\ref{lemma:oneparamderiv}, so we have
\begin{subequations}
\begin{align}
&|\alpha_P||\Tr(P\rho(x)) - \Tr(P\rho(x'))|\\
&\leq |\alpha_P|\int_{x_{c^*}}^{x_{c^*}'} \left(c_1 e^{-a_{\mathrm{lr}}\Delta(j,P)/2} + c_2\left(\frac{1}{1-\frac{35\log^2(b\Delta(j,P))}{b\Delta(j,P)}}\right)\Delta(j,P)^{10}\exp\left(-\frac{2}{7}\frac{b\Delta(j,P)}{\log^2(b\Delta(j,P))}\right)\right)\,dt\\
&\leq 2|\alpha_P|\left(c_1 e^{-a_{\mathrm{lr}}\Delta(j,P)/2} + c_2\left(\frac{1}{1-\frac{35\log^2(b\Delta(j,P))}{b\Delta(j,P)}}\right)\Delta(j,P)^{10}\exp\left(-\frac{2}{7}\frac{b\Delta(j,P)}{\log^2(b\delta_1)}\right)\right),
\end{align}
\end{subequations}
where in the last line, we can bound this integral because $x_{c^*}, x_{c^*}' \in [-1,1]$, so their difference is at most $2$.
Taking $c_1' = 2c_1$ and $c_2' = 2c_2$, we arrive at the claim.
\end{proof}

\noindent With these two results, we can prove Lemma~\ref{lemma:approxlocal}.

\begin{proof}[Proof of Lemma~\ref{lemma:approxlocal}]
It remains to show that if we change multiple coordinates outside of $I_P$, the difference $|\Tr(P\rho(x)) - \Tr(P\rho(x'))|$ is still bounded.
Then, taking $f_P(x)$ to be $\alpha_P \Tr(P\rho(\chi_P(x)))$ with $\chi_P(x) \in [-1,1]^m$ equal to $x_c$ for coordinates $c \in I_P$ and $0$ for coordinates outside of $I_P$ gives the desired result.
Thus, we want to bound
\begin{equation}
|\alpha_P||\Tr(P\rho(x)) - \Tr(P\rho(x'))| = |\alpha_P||\Tr(P\rho(x_1,\dots, x_m)) - \Tr(P\rho(x_1',\dots,x_m'))|,
\end{equation}
where $x_{c'}' = x_{c'}$ if and only if $c' \in I_P$. We can bound this using the triangle inequality
\begin{align}
\begin{split}
|\alpha_P||\Tr(P \rho(x)) - \Tr(P \rho(x'))| &\leq |\alpha_P||\Tr(P \rho(x_1,x_2,\dots, x_m)) - \Tr(P\rho(x_1',x_2,\dots, x_m))|\\
&+ |\alpha_P||\Tr(P \rho(x_1',x_2,\dots, x_m)) - \Tr(P\rho(x_1',x_2',\dots, x_m))|\\
&+ \cdots\\
&+ |\alpha_P||\Tr(P \rho(x_1',\dots, x_{m-1}', x_m)) - \Tr(P \rho(x_1',\dots, x_m'))|
\end{split}
\end{align}
Here, recall that we are only changing coordinates outside of $I_P$, i.e., $x_c$ such that $h_j$ depends on $x_c$ and $\Delta(j,P) > \delta_1$ for $\delta_1$ in Definition~\ref{def:delta}.
Moreover, by assumption~\ref{assum:b}, each local term $h_j$ depends on at most a constant number $q$ of parameters.
Then, we can upper bound this sum by
\begin{subequations}
\begin{align}
|\alpha_P||\Tr(P\rho(x)) - \Tr(P\rho(x'))| &\leq q|\alpha_P| \sum_{j : \Delta(j,P) > \delta_1} |\Tr(P\rho(y_{k^*})) - \Tr(P\rho(y_{k^*}'))|,
\end{align}
\end{subequations}
where $k^*$ is defined as
\begin{equation}
    k^* \triangleq \argmax_{1 \leq k \leq m} |\Tr(P \rho(y_k)) - \Tr(P\rho(y_k'))|,
\end{equation}
and $y_k, y_k' \in [-1,1]^m$ denote parameter vectors that only differ in the $k$th coordinate.
Each of these terms in the summand can be bounded using Lemma~\ref{lemma:oneparamdist}:
\begin{subequations}
\begin{align}
&|\alpha_P||\Tr(P\rho(x)) - \Tr(P\rho(x'))|\\
&\leq q|\alpha_P| \sum_{j : \Delta(j,P) > \delta_1} \left(c_1' e^{-\frac{a_{\mathrm{lr}}\Delta(j,P)}{2}} + c_2'\left(\frac{1}{1-\frac{35\log^2(b\Delta(j,P))}{b\Delta(j,P)}}\right)\Delta(j,P)^{10}\exp\left(-\frac{2}{7}\frac{b\Delta(j,P)}{\log^2(b\Delta(j,P))}\right)\right)\\
\begin{split}
&\leq q|\alpha_P| \sum_{r=0}^{\infty} \sum_{j : \Delta(j,P) \in [\delta_1 + r, \delta_1 + r + 1]} \left(c_1' e^{-\frac{a_{\mathrm{lr}}\Delta(j,P)}{2}} \right.\\
&\hspace{100pt}\left.+ c_2'\left(\frac{1}{1-\frac{35\log^2(b\Delta(j,P))}{b\Delta(j,P)}}\right)\Delta(j,P)^{10}\exp\left(-\frac{2}{7}\frac{b\Delta(j,P)}{\log^2(b\Delta(j,P))}\right)\right).
\end{split}
\end{align}
\end{subequations}
Now, we want to upper bound this inner sum over $j$ such that $\Delta(j,P) \in [\delta_1 + r, \delta_1 + r + 1]$.
Using Eq.~\eqref{eq:localtermsbound}, we see that there are at most $|\mathsf{dom}(P)|(b_d + c_d(\delta_1 + r + 1)^d)$ interaction terms $h_j$ such that $\Delta(j,P) \in [\delta_1 + r, \delta_1 + r + 1]$.
Moreover, because $P$ acts on only a constant number of sites (by assumption), then
\begin{equation}
|\mathsf{dom}(P)|(b_d + c_d(\delta_1 + r + 1)^d) \leq c_P(b_d + c_d(\delta_1 + r + 1)^d),
\end{equation}
for some constant $c_P$. Using this as well as upper bounding the sum over $r$ by an integral, we have
\begin{subequations}
\begin{align}
&|\alpha_P||\Tr(P\rho(x)) - \Tr(P\rho(x'))|\\
&\leq q c_P|\alpha_P| \int_{r=0}^{+\infty} (b_d + c_d(\delta_1 + r + 1)^d)\\
&\cdot\left(c_1' e^{-\frac{a_{\mathrm{lr}}(\delta_1 + r)}{2}} + c_2'\left(\frac{1}{1-\frac{35\log^2(b(\delta_1 + r + 1))}{b(\delta_1 + r)}}\right)(\delta_1 + r + 1)^{10}\exp\left(-\frac{2}{7}\frac{b(\delta_1 + r)}{\log^2(b(\delta_1 + r + 1))}\right)\right)\,dr.
\end{align}
\end{subequations}
It remains to integrate this to obtain our desired bound. Distributing, we can split this integral into four terms. We bound each of these individually.

First, we have
\begin{equation}
\int_{r=0}^{+\infty} b_d c_1' e^{-a_{\mathrm{lr}}(\delta_1 + r)/2}\,dr= c_1' b_d e^{-a_{\mathrm{lr}}\delta_1/2} \int_{r=0}^{+\infty} e^{-a_{{\mathrm{lr}}}r/2}\,dr = \frac{2c_1' b_d}{a_{\mathrm{lr}}}e^{-a_{\mathrm{lr}}\delta_1/2}.
\end{equation}
We also have
\begin{subequations}
\begin{align}
\int_{r=0}^{+\infty}c_dc_1'(\delta_1 + r +1)^d e^{-a_{\mathrm{lr}}(\delta_1 + r)/2}\,dr &= c_1'c_d e^{-a_{\mathrm{lr}}\delta_1/2}\int_{r=0}^{+\infty} (\delta_1 + r + 1)^d e^{-a_{\mathrm{lr}}r/2}\,dr\\
&= c_1'c_d e^{-a_{\mathrm{lr}}\delta_1/2} \sum_{k=0}^d \frac{d!\, 2^{d-k+1}}{k!\, a_{\mathrm{lr}}^{d-k+1}}(\delta_1 + 1)^k,
\end{align}
\end{subequations}
where in the last equality we used integration by parts.
For the other two integrals, we use Lemma~\ref{lemma:integralbound-0} to obtain
\begin{equation}
\int_{r=0}^{+\infty} b_d c_2' \left(\frac{1}{1 - \frac{35\log^2(b(\delta_1 + r + 1))}{b(\delta_1 + r)}}\right)(\delta_1 + r + 1)^{10}\exp\left(-\frac{2}{7}\frac{b(\delta_1 + r)}{\log^2(b(\delta_1 + r + 1))}\right)\,dr \leq c\epsilon,
\end{equation}
for some constant $c$. Similarly, for the last integral, by Lemma~\ref{lemma:integralbound-d}, we have
\begin{equation}
\int_{r=0}^{+\infty} c_d c_2' \left(\frac{1}{1 - \frac{35\log^2(b(\delta_1 + r + 1))}{b(\delta_1 + r)}}\right)(\delta_1 + r + 1)^{d + 10}\exp\left(-\frac{2}{7}\frac{b(\delta_1 + r)}{\log^2(b(\delta_1 + r + 1))}\right)\,dr \leq c'\epsilon,
\end{equation}
for some constant $c'$. Putting everything together, we have
\begin{subequations}
\begin{align}
&|\alpha_P||\Tr(P\rho(x)) - \Tr(P\rho(x'))|\\
&\leq q c_P |\alpha_P|\left(\frac{2c_1'b_d}{a_{\mathrm{lr}}} e^{-a_{\mathrm{lr}}\delta_1/2} + c_1'c_d e^{-a_{\mathrm{lr}}\delta_1/2} \sum_{k=0}^d \frac{d! 2^{d - k +1}}{k! a_{\mathrm{lr}}^{d-k+1}}(\delta_1 + 1)^k + c\epsilon + c'\epsilon\right).
\end{align}
\end{subequations}
Combining constants and simplifying, we have
\begin{equation}
|\alpha_P||\Tr(P\rho(x)) - \Tr(P\rho(x'))| \leq |\alpha_P| \left(e^{-a_{\mathrm{lr}}\delta_1/2} \sum_{k=0}^d c_k'' \delta_1^k + c''\epsilon\right).
\end{equation}
To obtain the final bound, we can use our choice of $\delta_1$ to write this bound in terms of $\epsilon$:
\begin{equation}
e^{-a_{\mathrm{lr}}\delta_1/2}\sum_{k=0}^d c_k''\delta_1^k = \sum_{k=0}^d c_k'' e^{-a_{\mathrm{lr}}\delta_1/2 + k\log\delta_1} \leq \left(\sum_{k=0}^d c_k''\right)\epsilon,
\end{equation}
where the last inequality follows from our choice of $\delta_1$ in Definition~\ref{def:delta} and $C_1$ in Lemma~\ref{lemma:c1}. Thus, we have
\begin{equation}
|\alpha_P||\Tr(P\rho(x)) - \Tr(P\rho(x'))| \leq C|\alpha_P|\epsilon,
\end{equation}
where we take
\begin{equation}
C = \sum_{k=0}^d c_k'' + c''.
\end{equation}
To complete the proof, recall that $f_P(x) = \alpha_P\Tr(P \rho(\chi_P(x)))$, where $\chi_P$ is defined in Eq.~\eqref{eq:restriction-chi}.
The function $f_P$ only depends on parameters in $I_P$ by definition.
By the previous analysis, since $\chi_P(x)$ and $x$ only differs in the coordinates outside of the set $I_P$, the function $f_P(x)$ should be close to $\alpha_P \Tr(P \rho(\chi_P(x)))$ in absolute value as required.
Moreover, $\Tr(P\rho(x))$ is smooth by Lemma 4 in~\cite{huang2021provably} in that
\begin{equation}
\norm{\nabla_x \Tr(P\rho(x))}_2^2 \leq C'\norm{P}_\infty^2 = C'
\end{equation}
for some constant $C' > 0$. Then, because $f_P$ is defined as $\alpha_P \Tr(P \rho(\chi_P(x)))$, we have
\begin{equation}
\norm{\nabla_x f_P(x)}_2^2 \leq |\alpha_P|^2 C',
\end{equation}
so $f_P$ is smooth as claimed.
\end{proof}

\subsection{Simplification using discretization}
\label{sec:discretization}

Now, we want to show that the sum of smooth local functions $f(x) = \sum_{P \in S^{\mathrm{(geo)}}} f_P(x)$ from Corollary~\ref{corollary:approxlocal} can be approximated by simple functions, i.e., linear combinations of indicator functions.
In order to do so, we discretize our parameter space and map each $x \in [-1,1]^m$ to some $x'$ with discrete values.
Our simple function is then $f$ evaluated on this discretized $x'$.
To state this more precisely, we first require some definitions.
An illustrative example of how each set is defined is given in Figure~\ref{fig:def3}.

\begin{definition}[Discretization]
\label{def:discretization}
Let $\epsilon > 0$. Let
\begin{equation}
\delta_2 \triangleq \frac{1}{\left\lceil \frac{\sqrt{C' |I_P|}}{\epsilon}\right\rceil},
\end{equation}
where $I_P$ is defined in Definition~\ref{def:ip} and $C'$ is as in Lemma~\ref{lemma:approxlocal}. Define the discretized parameter space as
\begin{equation}
X_P \triangleq \left.\begin{cases}
x \in [-1,1]^m : \text{if } c \notin I_P, x_{c} = 0\\
\hspace{62pt} \text{if } c \in I_P, x_{c} \in \left\{0, \pm \delta_2, \pm 2\delta_2,\dots, \pm 1\right\}
\end{cases}\right\}.
\end{equation}
Moreover, for each $x \in X_P$, define the thickened affine subspace close to the vector $x$ for coordinates in $I_P$ as
\begin{equation}
T_{x, P} \triangleq \left\{x' \in [-1,1]^m : -\frac{\delta_2}{2} < x_{c} - x_{c}' \leq \frac{\delta_2}{2}, \forall c \in I_P\right\}.
\end{equation}
\end{definition}

\begin{figure}[t]
\centering
\includegraphics[width=0.95\linewidth]{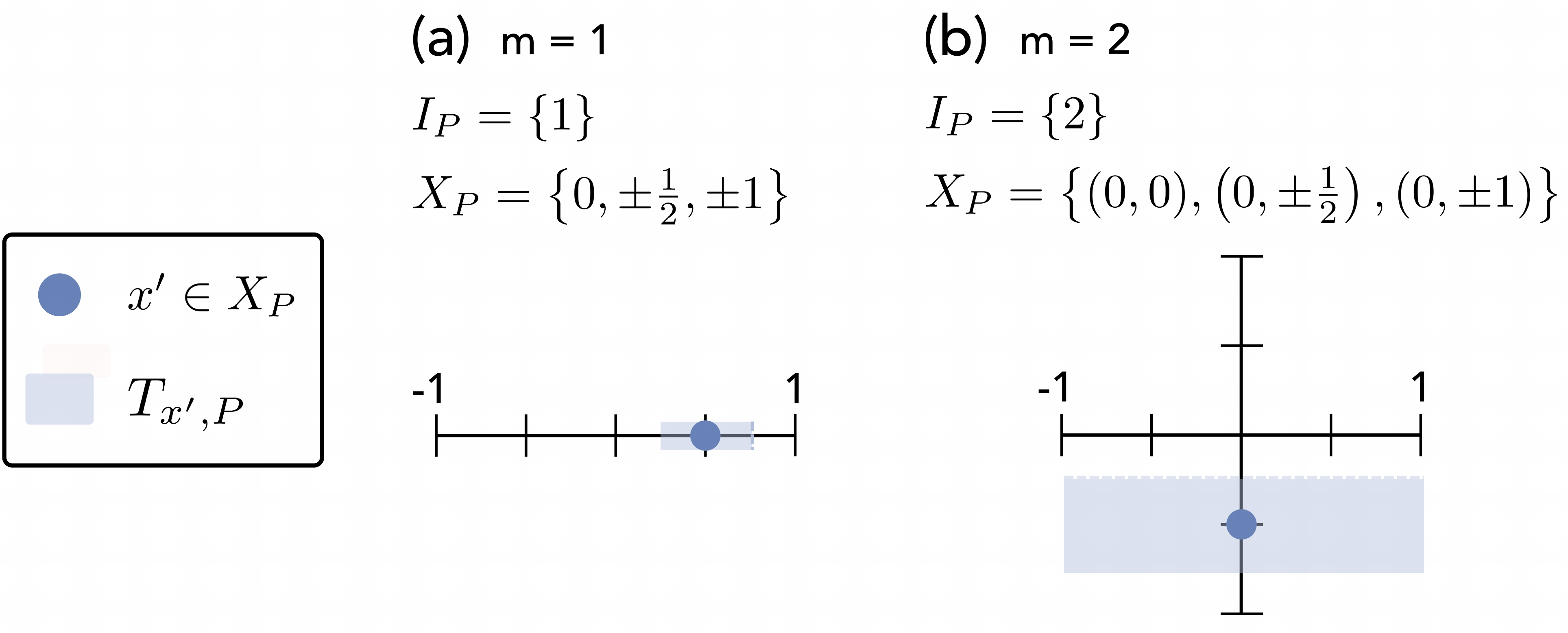}
\caption{\textbf{Example of Definition~\ref{def:discretization}}. Illustration of the set $T_{x', P}$ (light blue shading) for specific $x' \in X_P$ (blue circle), fixing $\delta_2 = 1/2$ for simplicity. \textbf{(a)} \textsc{Example for $m=1$.} $I_P$ is fixed to $\{1\}$ so that $X_P = \{0, \pm 1/2, \pm 1\}$ according to Def.~\ref{def:discretization}. $T_{x', P}$ is depicted for the chosen $x' = 1/2$. \textbf{(b)} \textsc{Example for $m = 2$.} $I_P$ is fixed to $\{2\}$, and $T_{x', P}$ is depicted for the chosen $x' = (0, -1/2)$.}
\label{fig:def3}
\end{figure}

\noindent With these definitions, the simple function that approximates $f$ is defined by
\begin{equation}
\label{eq:g}
g(x) \triangleq \sum_{P \in S^{\mathrm{(geo)}}} \left[ \, \sum_{x' \in X_P} f_P(x') \indicator[x \in T_{x',P}] \, \right] \triangleq \sum_{P \in S^{\mathrm{(geo)}}} g_P(x).
\end{equation}
In what follows, we prove that $g$ indeed approximates $f$ well. As in Appendix~\ref{sec:approxsmooth}, we first consider the simpler case where our observable $O = \alpha_P P$ is a single Pauli term, which easily generalizes to the general case via triangle inequality.

\begin{lemma}[Approximation using simple functions; simple case]
\label{lemma:approxsimple}
Let $\epsilon > 0$. Given this $\epsilon$ in Definition~\ref{def:discretization},
\begin{equation}
|g_P(x) - f_P(x)| < \epsilon|\alpha_P|
\end{equation}
for any $x$, where $f_P$ is as in Lemma~\ref{lemma:approxlocal} and $g_P$ is defined in Eq.~\eqref{eq:g}.
\end{lemma}

\begin{corollary}[Approximation using simple functions; general case]
\label{corollary:approxsimple}
Let $\epsilon > 0$. Given this $\epsilon$ in Definition~\ref{def:discretization}, then
\begin{equation}
|g(x) - f(x)| < \epsilon\left(\sum_{P \in S^{\mathrm{(geo)}}} |\alpha_P|\right)
\end{equation}
for any $x$, where $f$ is as in Corollary~\ref{corollary:approxlocal} and $g$ is defined in Eq.~\eqref{eq:g}.
\end{corollary}

\begin{proof}[Proof of Lemma~\ref{lemma:approxsimple}]
Consider some input $x \in [-1,1]^m$.
First, we want to argue that $x \in T_{x',P}$ for exactly one $x' \in X_P$.
Consider some variable $x_{c} \in [-1,1]$ of $x$ for $c \in I_P$.
It suffices to show that there exists $x_{c}' \in \{0,\pm \delta_2, \pm 2\delta_2,\dots, \pm 1\}$ such that $-\delta_2/2 < x_{c}' - x_{c} \leq \delta_2/2$.
This is clear because $\delta_2$ is defined as a fraction of the form $1/n$ for an integer $n$.
Moreover, there is at most one $x_{c}'$ such that this is true because each possible discrete value of $x_{c}'$ is separated by intervals of size $\delta_2$ while $x_{c}$ is within $\delta_2/2$ of $x_{c}'$, so there cannot be overlap for different values of $x_{c}'$.
Also, since $x_{c}$ is in a half-open interval of $x_{c}'$, this prevents points on the boundary (i.e., exactly $\delta_2/2$ away from $x_{c}'$) from being associated with two $x_{c}'$.
Finally, this half-open interval does not prevent the boundary case of $x_{c} = -1$ from being associated with an $x_{c}'$ because $-1$ is always a possible discrete value for $x_{c}'$.
This occurs again because of our choice of $\delta_2$ as a fraction of the form $1/n$ for an integer $n$. Thus, $x \in T_{x', P}$ for exactly one $x' \in X_P$.

With this, our goal is to show that
\begin{equation}
|g_P(x) - f_P(x)| = |f_P(x') - f_P(x)| < \epsilon|\alpha_P|.
\end{equation}
There are two parts to proving this. By definition of $T_{x', P}$ in Definition~\ref{def:discretization}, this means that $x'$ and $x$ are close for coordinates in $I_P$. However, for coordinates not in $I_P$, $x'$ and $x$ can be far away. Nevertheless, from our results in Lemma~\ref{lemma:approxlocal}, we know that $f_P$ does not change much when its input only differs for coordinates not in $I_P$. Thus, we can use this to obtain our bound.

To make this more clear, we introduce the notation
\begin{equation}
f_P(x) = f_P(x_{\mathrm{in}}; x_{\mathrm{out}}),
\end{equation}
where $x_{\mathrm{in}}$ denotes the variables $x_{c} \in [-1, 1]$ such that $c \in I_P$ and $x_{\mathrm{out}}$ denotes the variables $x_{c} \in [-1, 1]$ such that $c \notin I_P$.
With this, we can use the triangle inequality to treat the two cases separately:
\begin{subequations}
\label{eq:separateinout}
\begin{align}
|f_P(x') - f_P(x)| &= |f_P(x_{\mathrm{in}}'; x_{\mathrm{out}}') - f_P(x_{\mathrm{in}}; x_{\mathrm{out}})|\\
\begin{split}
&\leq |f_P(x_{\mathrm{in}}'; x_{\mathrm{out}}') - f_P(x_{\mathrm{in}}'; x_{\mathrm{out}})|\\
&+ |f_P(x_{\mathrm{in}}'; x_{\mathrm{out}}) - f_P(x_{\mathrm{in}}; x_{\mathrm{out}})|.
\end{split}
\end{align}
\end{subequations}
Here, in the first term, only the coordinates not in $I_P$ change while in the second term, only coordinates in $I_P$ change.
To bound the first term, we can use Lemma~\ref{lemma:approxlocal} with $\epsilon$ set to $\epsilon/(2C)$, where $C$ is the constant defined in Lemma~\ref{lemma:approxlocal}, to obtain
\begin{equation}
\label{eq:boundout}
|f_P(x_{\mathrm{in}}'; x_{\mathrm{out}}') - f_P(x_{\mathrm{in}}'; x_{\mathrm{out}})| \leq |\alpha_P| \frac{\epsilon}{2}.
\end{equation}

For the second term in Eq.~\eqref{eq:separateinout}, we bound this using the fact that $x'$ and $x$ are separated by at most $\delta_2$ for coordinates in $I_P$ and the smoothness condition on $f_P$ from Lemma~\ref{lemma:approxlocal}.
The key step here is that we can write this difference as the integral of the directional derivative of $f_P$ along the direction from $x_{\mathrm{in}}$ to $x'_{\mathrm{in}}$ given by a line. In particular, we can parameterize this line by $x_{\mathrm{in}}(t) = x_{\mathrm{in}} + (x_{\mathrm{in}}' - x_{\mathrm{in}})t$. Notice that at $t = 0$, this is equal to $x_{\mathrm{in}}$ while at $t = 1$, this is equal to $x_{\mathrm{in}}'$. Thus, suppressing the $x_{\mathrm{out}}$ parameters in our notation, we have
\begin{subequations}
\begin{align}
|f_P(x_{\mathrm{in}}') - f_P(x_{\mathrm{in}})| &= \left| \int_0^1 \frac{\partial f_P(x_{\mathrm{in}}(t))}{\partial t}\,dt \right|\\
&\leq \int_0^1 \left| \frac{\partial f_P(x_{\mathrm{in}}(t))}{\partial t}\right| \,dt\\
&= \int_0^1 \left|\frac{\partial f_P(x_{\mathrm{in}}(t))}{\partial x_{\mathrm{in}}(t)} \cdot \frac{\partial x_{\mathrm{in}}(t)}{\partial t} \right|\,dt\\
&= \int_0^1 \left|\nabla_{x_{\mathrm{in}}} f_P(x_{\mathrm{in}}) \cdot (x_{\mathrm{in}}'-x_{\mathrm{in}})\right|\,dt\\
&\leq \int_0^1 \norm{\nabla_{x_{\mathrm{in}}} f_P(x_{\mathrm{in}})}_2 \norm{x_{\mathrm{in}}'-x_{\mathrm{in}}}_2\,dt\\
&\leq \sqrt{C'}|\alpha_P| \norm{x_{\mathrm{in}}' - x_{\mathrm{in}}}_2\\
&\leq \sqrt{C'} |\alpha_P| \sqrt{|I_P|} \norm{x_{\mathrm{in}}' - x_{\mathrm{in}}}_\infty\\
&\leq \sqrt{C'}|\alpha_P| \sqrt{|I_P|} \frac{\delta_2}{2}\\
&\leq \frac{\epsilon}{2}|\alpha_P|.
\end{align}
\end{subequations}
Here, in the third line, we use the chain rule. In the fifth line, we use the Cauchy-Schwarz inequality. In the sixth line, we use the smoothness condition from Lemma~\ref{lemma:approxlocal} to bound the $\ell_2$-norm of the gradient. In the seventh line, we use the fact that $\norm{y}_2 \leq \sqrt{n}\norm{y}_\infty$ where $n$ is the number of elements in $y$. In the eighth line, we use the definition of $T_{x, P}$. Finally, in the last line, we use our choice of $\delta_2$ as
\begin{equation}
\delta_2 = \frac{1}{\left\lceil \frac{\sqrt{C' |I_P|}}{\epsilon}\right\rceil} \leq \frac{\epsilon}{\sqrt{C'|I_P|}}.
\end{equation}
Combining this bound with Eq.~\eqref{eq:boundout} and plugging into Eq.~\eqref{eq:separateinout}, we have
\begin{equation}
|f_P(x') - f_P(x)| < \frac{\epsilon}{2}|\alpha_P| + \frac{\epsilon}{2}|\alpha_P| = \epsilon|\alpha_P|,
\end{equation}
as required.
\end{proof}

\subsection{Simple form for ground state property}
\label{sec:simpleformal}

We can combine the results of the previous two sections to obtain the final result giving a simpler form for the ground state property $\Tr(O\rho(x))$.
The proof of this statement is simple given the previous results.

\begin{theorem}[Simple form for $\Tr(O\rho(x))$]
\label{thm:simple}
Let $1/e > \epsilon > 0$.
Given $\epsilon$, we define $\delta_1$ according to Definition~\ref{def:delta} with $\epsilon$ set to $\epsilon / (2C)$ for the constant $C$ defined in Corollary~\ref{corollary:approxlocal}, and define $\delta_2$ according to Definition~\ref{def:discretization} with $\epsilon$ set to $\epsilon / 2$.
The ground state property $\Tr(O\rho(x))$ can be approximated by a simple function, i.e.,
\begin{equation}
|\Tr(O\rho(x)) - g(x)| < \epsilon\left(\sum_{P \in S^{\mathrm{(geo)}}} |\alpha_P|\right),
\end{equation}
where $g$ is defined in Eq.~\eqref{eq:g}.
\end{theorem}

\begin{proof}
By the triangle inequality, we have
\begin{equation}
\label{eq:triangle}
|\Tr(O\rho(x)) - g(x)| \leq |\Tr(O\rho(x)) - f(x)| + |f(x) - g(x)|.
\end{equation}
Here, the first term can be bounded by Corollary~\ref{corollary:approxlocal} to obtain
\begin{equation}
\label{eq:triangle1}
|\Tr(O\rho(x)) - f(x)| \leq \frac{\epsilon}{2}\left(\sum_{P \in S^{\mathrm{(geo)}}} |\alpha_P|\right).
\end{equation}
Meanwhile, the second term in Eq.~\eqref{eq:triangle} can be bounded by Corollary~\ref{corollary:approxsimple} to obtain
\begin{equation}
\label{eq:triangle2}
|f(x) - g(x)| < \frac{\epsilon}{2}\left(\sum_{P \in S^{\mathrm{(geo)}}} |\alpha_P|\right).
\end{equation}
Combining Eq.~\eqref{eq:triangle1} and Eq.~\eqref{eq:triangle2} in Eq.~\eqref{eq:triangle}, we have
\begin{equation}
|\Tr(O\rho(x)) - g(x)| < \epsilon\left(\sum_{P \in S^{\mathrm{(geo)}}} |\alpha_P|\right)
\end{equation}
This concludes the proof.
\end{proof}

\subsection{Technical lemmas for finding constants and bounding integrals}
\label{sec:techlemmas}

In this section, we state and prove several technical lemmas for showing the existence of certain constants and bounding integrals of specific forms needed throughout Appendix~\ref{sec:simple}.
Throughout this section, we use the notation $b \triangleq \gamma/(2v_{\mathrm{lr}})$.
First, we show the existence of the constants utilized in Definition~\ref{def:delta}.

\begin{lemma}
\label{lemma:c1}
Given $a_{\mathrm{lr}}, b > 0$ and $d \geq 1$, there exists a constant $C_1$ large enough such that for all $1/e > \epsilon' > 0$ and for all $\delta_1' > C_1\log^2(1/\epsilon')$,
\begin{equation}
\label{eq:c1}
\frac{a_{\mathrm{lr}}}{2}\delta_1'- d\log(\delta_1') \geq \log\left(\frac{1}{\epsilon'}\right).
\end{equation}
Explicitly, such a constant $C_1$ can be given by
\begin{equation}
C_1 = \frac{(2d + \sqrt{4d^2+2a_{\mathrm{lr}}})^2}{a_{\mathrm{lr}}^2}.
\end{equation}
\end{lemma}

\begin{proof}
For simplicity, throughout this proof, let $x = \log(1/\epsilon')$. Because we assert that $1/e > \epsilon' > 0$, then $1 < x < \infty$.
First, we consider the monotonicity of $f(\delta_1') = \frac{a_{\mathrm{lr}}}{2} \delta_1'- d\log(\delta_1')$.
Taking the derivative of $f$ shows that $f(\delta_1')$ is monotonically increasing for $\delta_1' \geq 2d / a_{\mathrm{lr}}$.
Since $\delta_1' > C_1\log^2(1/\epsilon') = C_1x^2 \geq C_1 \geq 2d / a_{\mathrm{lr}}$ (note $d \geq 1$), it suffices to establish the claim for $\delta_1' = C_1x^2$, i.e.,
\begin{equation}
\frac{a_{\mathrm{lr}}}{2}C_1x^2 - d\log(C_1 x^2) \geq x
\end{equation}
for $x > 1$. We show that our choice of $C_1$ satisfies this inequality. First, using the inequality $\log(z) \leq 2(\sqrt{z} - 1)$ for $z > 0$, we can apply this with $z = C_1x^2$ to obtain
\begin{equation}
\frac{a_{\mathrm{lr}}}{2}C_1 x^2 - d\log(C_1x^2) \geq \frac{a_{\mathrm{lr}}}{2} C_1 x^2 - 2d\sqrt{C_1}x + 2.
\end{equation}
Bounding this trivially because $x^2 > x$ for $x > 1$, we have
\begin{equation}
\frac{a_{\mathrm{lr}}}{2}C_1 x^2 - d\log(C_1x^2) \geq \left(\frac{a_{\mathrm{lr}}}{2}C_1 - 2d\sqrt{C_1}\right)x.
\end{equation}
Plugging in our choice of $C_1$ and simplifying, we have
\begin{subequations}
\begin{align}
\frac{a_{\mathrm{lr}}}{2}C_1 x^2 - d\log(C_1x^2) &\geq \left(\frac{\left(2d+\sqrt{4d^2 + a_{\mathrm{lr}}}\right)^2}{2a_{\mathrm{lr}}} - 2d\left(\frac{2d + \sqrt{4d^2 + a_{\mathrm{lr}}}}{a_{\mathrm{lr}}}\right)\right)x\\
&= \left(\frac{8d^2 + 2a_{\mathrm{lr}} + 4d\sqrt{4d^2 + 2a_{\mathrm{lr}}}}{2a_{\mathrm{lr}}} - \frac{8d^2+4d\sqrt{4d^2 + 2a_{\mathrm{lr}}}}{2a_{\mathrm{lr}}}\right)x\\
&= x.
\end{align}
\end{subequations}
Hence, we obtain the desired inequality.
\end{proof}

\begin{lemma}
\label{lemma:c2}
Given $b > 0$, there exists a constant $C_2$ large enough such that for all $1/e > \epsilon' > 0$ and for all $\delta_1' > C_2\log^2(1/\epsilon')$,
\begin{equation}
\label{eq:c2}
\frac{2b\delta_1'}{7\log^2(b(\delta_1' + 1))} - 22\log(b(\delta_1'+1)) \geq \log\left(\frac{1}{\epsilon'}\right).
\end{equation}
Explicitly, such a constant $C_2$ can be given by
\begin{equation}
C_2 = \max\left(\frac{(18 b + (63 \cdot 22 / 2))^3}{b}, 1, \frac{2(7\cdot 16)^2}{b}, \frac{2^{3} (7 \cdot 22 \cdot 64)^4}{b}\right).
\end{equation}
\end{lemma}

\begin{proof}
For simplicity, throughout this proof, let $x =\log(1/\epsilon')$. Because we assert that $1/e > \epsilon' > 0$, then $1 < x < \infty$.
First, we consider the monotonicity of $f(\delta_1') = \frac{2b\delta_1'}{7\log^2(b(\delta_1' + 1))} - 22\log(b(\delta_1'+1))$.
Taking the derivative of $f$ shows that $f(\delta_1')$ is monotonically increasing for $\delta_1' \geq (18 b + (63 \cdot 22 / 2))^3 / b$.
For $\delta_1' \geq (18 b + (63 \cdot 22 / 2))^3 / b$, we can make use of $\log(z) \leq 3 z^{1/3}, \forall z > 0$ to show that
\begin{equation}
    \frac{2b (\delta_1' + 1)}{7 \log^2(b(\delta_1' + 1))} \geq \frac{4}{7} b + 22.
\end{equation}
Because $\delta_1' \geq (18 b + (63 \cdot 22 / 2))^3 / b \geq e / b$, we have
\begin{equation}
    \log^3(b(\delta_1'+ 1)) \geq 1.
\end{equation}
Together, we can show that for $\delta_1' \geq (18 b + (63 \cdot 22 / 2))^3 / b$,
\begin{equation}
    f'(\delta_1') = \frac{2b}{7 \log^2(b(\delta_1' + 1))} - \frac{4b}{7 (\delta_1' + 1) \log^3(b(\delta_1' + 1))} - \frac{22}{\delta_1' + 1} \geq 0.
\end{equation}

Since $\delta_1' > C_2 \log^2(1/\epsilon') = C_2 x^2 \geq C_2 \geq (18 b + (63 \cdot 22 / 2))^3 / b$, it suffices to establish the claim for $\delta_1' = C_2 x^2$, i.e.,
\begin{equation}
\frac{2bC_2x^2}{7\log^2(bC_2x^2 + b)} - 22\log(bC_2x^2 + b) \geq x
\end{equation}
for $x > 1$. We show that our choice of $C_2$ satisfies this. First, notice that it suffices to show the following two inequalities
\begin{equation}
\label{eq:1}
\frac{bC_2x^2}{7\log^2(bC_2 x^2 + b)} \geq x
\end{equation}
and
\begin{equation}
\label{eq:2}
\frac{bC_2x^2}{7\log^2(bC_2x^2 + b)} \geq 22\log(bC_2x^2 + b).
\end{equation}
Since $C_2 \geq 1$ and $x > 1$, then $C_2x^2 \geq 1$ and $bC_2 x^2 + b \leq 2bC_2x^2$. Then, in Eq.~\eqref{eq:1}, we have
\begin{equation}
\frac{bC_2x^2}{7\log^2(bC_2x^2 + b)} \geq \frac{bC_2x^2}{7\log^2(2bC_2x^2)} \geq \frac{\sqrt{C_2b}}{7\cdot 16\sqrt{2}}x \geq x,
\end{equation}
where the second inequality follows using the inequality $\log(z) \leq 4z^{1/4}$ for $z > 0$, applied with $z = 2bC_2x^2$, and the last inequality follows from our choice of $C_2$.
This proves Eq.~\eqref{eq:1}. Now, to prove Eq.~\eqref{eq:2}, notice that it suffices to show that
\begin{equation}
\label{eq:3}
\frac{bC_2x^2}{7} \geq 22(4 (2bC_2x^2)^{1/4})^3 = 22 \cdot 64 (2bC_2)^{3/4} x^{3/2}.
\end{equation}
This is because, again using the inequality $\log(z) \leq 4z^{1/4}$ with $z = 2bC_2x^2$, then
\begin{equation}
\frac{bC_2x^2}{7} \geq 22 \cdot 64 (2bC_2)^{3/4} x^{3/2} \geq 22\log^3(2bC_2x^2)
\end{equation}
so Eq.~\eqref{eq:3} implies Eq.~\eqref{eq:2}.
Thus, it remains to prove Eq.~\eqref{eq:3}, which is equivalent to
\begin{equation}
    C_2 \geq \frac{(7 \cdot 22 \cdot 64)^4 \cdot 2^{3}}{b}.
\end{equation}
Because $x > 1$, our choice of $C_2$ satisfies the above inequality.
\end{proof}

\begin{lemma}
\label{lemma:c3}
Given $a_{\mathrm{lr}}, b > 0$ and $d \geq 1$, there exists a constant $C_3$ large enough such that for all $1/e > \epsilon' > 0$ and for all $\delta_1' > C_3\log^2(1/\epsilon')$,
\begin{equation}
\label{eq:c3}
\frac{2b\delta_1'}{\log^2(b(\delta_1' + 1))} - (d + 22)\log(b(\delta_1'+1)) \geq \log\left(\frac{1}{\epsilon'}\right)
\end{equation}
Explicitly, such a constant $C_3$ can be given by
\begin{equation}
C_3 = \max\left(\frac{(18b + 63(d+22)/2)^3}{b}, 1, \frac{2(7\cdot 16)^2}{b}, \frac{2^3(7 \cdot (d+22) \cdot 64)^4}{b}\right).
\end{equation}
\end{lemma}

\begin{proof}
The proof is the same as that of Lemma~\ref{lemma:c2} after replacing $22$ by $d+22$.
\end{proof}

\noindent Next, we begin the integral bounds portion of this section and reprove a variant of the lemma introduced in~\cite{bachmann2012automorphic}.

\begin{lemma}[Variant of Lemma 2.5 in~\cite{bachmann2012automorphic}, Lemma 5 in~\cite{huang2021provably}]
\label{lemma:integralbound}
For $a > 0$ define
\[u_a(x) = e^{-a\frac{x}{\log^2(x)}}\]
on the domain $x \in (1,\infty)$, where $\log$ denotes the natural logarithm. For all integers $k \geq 0$ and $t \geq 5504$ such that
\begin{equation}
a\frac{t}{\log^2 t} > 2k + 2,
\end{equation}
we have the bound
\begin{equation}
\int_t^{+\infty} x^k u_a(x)\,dx \leq \frac{1}{a\left(1 - \frac{2k+2}{\tau(t)}\right)} t^{2k+2} u_a(t),
\end{equation}
where $\tau(x) \triangleq ax/\log^2(x)$.
\end{lemma}

\noindent To prove this, we need a bound on the upper incomplete Gamma function:
\begin{lemma}[Proposition 2.7 in~\cite{pinelis2020exact}]
Take any real $n \geq 0$. Then,
\label{lemma:incompletegammabound}
\begin{equation}
\Gamma(n + 1, z) \leq \frac{1}{1- \frac{n}{z}} z^n e^{-z},
\end{equation}
for all real $z > n$.
\end{lemma}

\begin{proof}[Proof of Lemma~\ref{lemma:integralbound}]
Define the function
\begin{equation}
\tau(x) \triangleq a\frac{x}{\log^2(x)}.
\end{equation}
Here, because we are considering the domain $x \in (1,\infty)$, then this function is well-defined and differentiable. Moreover, it is always positive because $\log^2(x) \geq \log(x) > 0$ for $x > 1$. Also, consider the derivative
\begin{equation}
\frac{d\tau}{dx} = a\left(\frac{\log(x) - 2}{\log^3(x)}\right).
\end{equation}
Again, this is well-defined because $\log^3(x) > 0$ for $x > 1$. Furthermore, we see that if $x \geq e^2$, then $\frac{d\tau}{dx} > 0$. Thus, for $x \geq e^2$, $\tau(x)$ is monotone increasing. Ultimately, our goal is to bound the integral
\begin{equation}
\int_t^{+\infty} x^k u_a(x) \,dx = \int_t^{+\infty} x^k e^{-\tau(x)}\,dx
\end{equation}
by using a substitution $\tau = \tau(x), d\tau = \frac{d\tau}{dx}dx$. Substituting in for $x$, we use the inverse $x = x(\tau)$ and for the differential $dx$, we use $dx = \frac{dx}{d\tau}d\tau$ to obtain
\begin{equation}
\label{eq:substitution}
\int_{\tau(t)}^{+\infty}(x(\tau))^k e^{-\tau} \frac{dx}{d\tau}\,d\tau.
\end{equation}
We want to get this into the form of the upper incomplete Gamma function:
\begin{equation}
\Gamma(n+1,z) = \int_z^{+\infty} \tau^n e^{-\tau}\,d\tau = n! e^{-z} \sum_{k=0}^n \frac{z^k}{k!}.
\end{equation}
Thus, we want to find bounds on $\frac{dx}{d\tau}$ and $x(\tau)$ in terms of $\tau$ (and constants). Since we define $x(\tau)$ as the inverse of $\tau(x)$, then we know that
\begin{equation}
\frac{dx}{d\tau} = \frac{1}{\frac{d\tau}{dx}} = \frac{1}{a}\left(\frac{\log^3(x)}{\log(x) - 2}\right).
\end{equation}
We notice here that if $x \geq 28$, then
\begin{equation}
\frac{dx}{d\tau} = \frac{1}{a}\left(\frac{\log^3(x)}{\log(x) - 2}\right) \leq \frac{x}{a}.
\end{equation}
If we further require $x \geq 5504$, then
\begin{equation}
x \leq \left(\frac{x}{\log^2x}\right)^2 = \frac{\tau^2}{a^2}.
\end{equation}
Using these together, we have that
\begin{equation}
\frac{dx}{d\tau} \leq \frac{x}{a} \leq \frac{\tau^2}{a^3}.
\end{equation}
Plugging these into Eq.~\eqref{eq:substitution}, we can upper bound our integral
\begin{subequations}
\begin{align}
\int_t^{+\infty} x^k e^{-\tau(x)}\,dx &= \int_{\tau(t)}^{+\infty}(x(\tau))^k e^{-\tau} \frac{dx}{d\tau}\,d\tau\\
&\leq \int_{\tau(t)}^{+\infty} \frac{\tau^{2k}}{a^{2k}} e^{-\tau} \frac{\tau^2}{a^3}\,d\tau\\
&= \frac{1}{a^{2k+3}} \int_{\tau(t)}^{+\infty}\tau^{2k+2}e^{-\tau}\,d\tau\\
&= \frac{1}{a^{2k+3}}\Gamma(2k+3,\tau(t)).
\end{align}
\end{subequations}
Now, applying Lemma~\ref{lemma:incompletegammabound}, we can further bound this:
\begin{equation}
\int_t^{+\infty} x^k e^{-\tau(x)}\,dx \leq \frac{1}{a^{2k+3}} \frac{1}{1-\frac{2k+2}{\tau(t)}} (\tau(t))^{2k+2}e^{-\tau(t)}
\end{equation}
for $\tau(t) > 2k+2$. Finally, since $\tau(t) \leq at$ for $t \geq e$, then we have
\begin{equation}
\int_t^{+\infty} x^k e^{-\tau(x)}\,dx \leq \frac{1}{a\left(1 - \frac{2k+2}{\tau(t)}\right)} t^{2k+2} e^{-\tau(t)}.
\end{equation}
\end{proof}

\noindent We use this to obtain another integral bound, which is as follows.

\begin{lemma}
\label{lemma:integralbound-0}
Let $\delta_1, \epsilon$ be as in Definition~\ref{def:delta}. Then, there exists a constant $c$ such that
\begin{equation}
I = \int_{r=0}^{+\infty} \left(\frac{1}{1-\frac{35\log^2(b(\delta_1 + r + 1))}{b(\delta_1 + r)}}\right)(\delta_1 + r + 1)^{10} \exp\left(-\frac{2}{7}\frac{b(\delta_1 + r)}{\log^2(b(\delta_1 +r + 1))}\right)\,dr \leq c\epsilon.
\end{equation}
\end{lemma}

\begin{proof}
Using the substitution $x = b(\delta_1 + r + 1)$, this integral transforms into
\begin{equation}
I = \left(\frac{1}{b}\right)^{n + 11} \int_{x = b(\delta_1 + 1)}^{+\infty} \left(\frac{1}{1-\frac{35\log^2x}{x-b}}\right)x^{10}\exp\left(-\frac{2}{7} \frac{x-b}{\log^2(x)}\right)\,dx.
\end{equation}
Here, we can show that for our choice of $\delta_1$, $e^{2b/7\log^2(x)}$ and $1/(1-(35\log^2x)/(x-b))$ are both monotonically decreasing in $x$. The derivative of the exponential term is
\begin{equation}
\frac{d}{dx} e^{\frac{2}{7} \frac{b}{\log^2x}} = -\frac{4b e^{\frac{2}{7}\frac{b}{\log^2x}}}{7x\log^3x}.
\end{equation}
To show that the exponential term is monotonically decreasing, we need to show that this derivative is less than $0$ for $x \geq b(\delta_1 + 1)$. We see that $e^{2b/(7\log^2x)}$ is always nonnegative and $\log^3 x$ is positive as long as $x > 1$ (in which case $x > 0$ as well). Thus, we only require $x > 1$ for this derivative to be less than $0$.

Similarly, for the other term, we have the derivative
\begin{equation}
\frac{d}{dx}\frac{1}{1-\frac{35\log^2 x}{x -b}} = -\frac{35\log x (2b - 2x  +x\log x)}{x(b - x + 35\log^2 x)^2}.
\end{equation}
In order for this to be less than $0$, we see that $(b - x + 35\log^2x)^2$ is always nonnegative and $\log x$ is positive as long as $x > 1$ (in which case $x > 0$ as well). Then, the only term left is $2b - 2x + x \log x$, which is positive as long as $\log x > 2(x-b)/x = 2- 2b/x$. This is satisfied is $\log x > 2$ instead, which follows when $x > e^2$.

Putting everything together, we see that both of these terms are monotonically decreasing in $x$ for $x > e^2$. In our integral, we have $x \geq b(\delta_1 + 1)$. However, by our choice of $\delta_1$ in Definition~\ref{def:delta}, we have that $\delta_1 \geq 5900/b$ so that $b(\delta_1 + 1) > b\delta_1 \geq 5900 > e^2$. Hence, the condition for these terms to be monotonically decreasing is satisfied for the bounds of the integral.

Thus, because these terms are monotonically decreasing, we can upper bound the integral by
\begin{equation}
I \leq \left(\frac{1}{b}\right)^{11} \exp\left(\frac{2}{7}\frac{b}{\log^2(b(\delta_1 + 1))}\right)\left(\frac{1}{1-\frac{35\log^2(b(\delta_1 + 1))}{b\delta_1}}\right) \int_{x = b(\delta_1 + 1)}^{+\infty} x^{10} e^{-\frac{2}{7}\frac{x}{\log^2(x)}}\,dx.
\end{equation}
Now, we can use Lemma~\ref{lemma:integralbound} to bound this final integral using $k = 10, a = 2/7$:
\begin{equation}
\int_{x=b(\delta_1 + 1)}^{+\infty} x^{10}e^{-\frac{2}{7}\frac{x}{\log^2(x)}}\,dx \leq \frac{7}{2} \frac{1}{1-\frac{7(22)\log^2(b(\delta_1 + 1))}{2b(\delta_1 + 1)}} (b(\delta_1 + 1))^{22} \exp\left(-\frac{2}{7}\frac{b(\delta_1 + 1)}{\log^2(b(\delta_1 + 1))}\right).
\end{equation}
Here, we note that the conditions are satisfied because
\begin{equation}
t = \frac{\gamma(\delta_1 + 1)}{2v_{\mathrm{lr}}} \geq \frac{\gamma\delta_1}{2v_{\mathrm{lr}}} \geq \max(5900, \alpha, 7(d + 11), \theta) \geq 5900,
\end{equation}
and it is clear that for $t \geq 5900$ that $at/\log^2t > 22$.
Let $c_1 = 7(1/b)^{11}/2$, and we can combine these bounds:
\begin{equation}
I \leq c_1 \exp\left(-\frac{2}{7}\frac{b\delta_1}{\log^2(b(\delta_1 + 1))}\right)\left(\frac{1}{1 - \frac{35\log^2(b(\delta_1 + 1))}{b\delta_1}}\right)\left(\frac{1}{1-\frac{7(22)\log^2(b(\delta_1 + 1))}{2b(\delta_1 + 1)}}\right)(b(\delta_1 + 1))^{22}.
\end{equation}
We can further bound this by
\begin{equation}
I \leq 4c_1 \exp\left(\frac{-2\gamma\delta_1 + 14(22)v_{\mathrm{lr}}\log^3(b(\delta_1 + 1))}{14v_{\mathrm{lr}}\log^2(b(\delta_1 + 1))}\right).
\end{equation}
Here, this is because of Eq.~\eqref{eq:c4}. This follows because
\begin{equation}
\frac{1}{1-\frac{35\log^2(b(\delta_1 + 1))}{b\delta_1}} \leq \frac{1}{1-\frac{77\log^2(b(\delta_1 + 1))}{b(\delta_1 + 1)}},
\end{equation}
Now, by our choice of $\delta_1$ and Lemma.~\ref{lemma:c2}, then we have
\begin{equation}
I \leq 4c_1 e^{-\log(1/\epsilon)} = 4c_1 \epsilon.
\end{equation}
Taking $c = 4c_1$, we arrive at our claim.
\end{proof}

\begin{lemma}
\label{lemma:integralbound-d}
Let $\delta_1, \epsilon$ be as in Definition~\ref{def:delta}. Then, there exists a constant $c'$ such that
\begin{equation}
I = \int_{r=0}^{+\infty} \left(\frac{1}{1-\frac{35\log^2(b(\delta_1 + r + 1))}{b(\delta_1 + r)}}\right)(\delta_1 + r + 1)^{d + 10} \exp\left(-\frac{2}{7}\frac{b(\delta_1 + r)}{\log^2(b(\delta_1 +r + 1))}\right)\,dr \leq c'\epsilon.
\end{equation}
\end{lemma}

\begin{proof}
The proof is the same as that of Lemma~\ref{lemma:integralbound-0} after replacing $x^{10}$ by $x^{d + 10}$.
Moreover, in the final steps, instead of using Eq.~\eqref{eq:c4} and Lemma~\ref{lemma:c2}, we use Eq.~\eqref{eq:c5} and Lemma~\ref{lemma:c3}, respectively.
\end{proof}

\section{Norm inequality for observables}
\label{sec:norminequality}

The efficiency of learning depends strongly on the complexity of the target functions we would like to learn.
One way to characterize the complexity of the target function is to consider an appropriate norm of the function.
Given an observable $O = \sum_P \alpha_P P$ specified by the Pauli coefficients $\alpha_P$, Theorem~\ref{thm:simple} shows that having a smaller $\ell_1$-norm $\sum_P |\alpha_P|$ on the Pauli coefficients implies that the ground state property $\Tr(O \rho(x))$ can be better approximated by a simple function.
This motivates the derivation of bounds on $\sum_P |\alpha_P|$.

A technical contribution of this work is to develop a norm inequality relating the $\ell_1$-norm of the Pauli coefficients  $\sum_P |\alpha_P|$ to the spectral norm $\norm{O}_\infty$ (the largest singular value).
To state this result precisely, we first present some formal definitions.
Throughout the remainder of this section, we consider labelling the $n$ qubits in a $d$-dimensional lattice with a $d$-tuple, $\ell = (\ell_1,\dots, \ell_d)$, where each $\ell_k \in \{1,\dots, \lfloor \sqrt[d]{n} \rfloor\}$.

\begin{definition}[Domain of an observable]
\label{def:domain}
Let $O$ be an arbitrary observable in a finite $d$-dimensional space. Then, define the \emph{domain} $\mathsf{dom}(O) \subseteq \{1,\dots, \lfloor \sqrt[d]{n} \rfloor\}^d$ of $O$ to be the set of qubits that $O$ acts nontrivially on.
\end{definition}

\begin{definition}[geometrically local with range $R$]
\label{def:gl}
Let $O$ be an arbitrary observable in a finite $d$-dimensional space and let $\mathsf{dom}(O) \subseteq \{1,\dots, \lfloor \sqrt[d]{n}\rfloor\}^d$ be its domain. Moreover, let $\mathsf{dom}(O)_k = \pi_k(\mathsf{dom}(O)) \subseteq \{1,\dots, \lfloor \sqrt[d]{n}\rfloor\}$, where $\pi_k:\mathbb{Z}^d \to \mathbb{Z}$ is the projection map onto the $k$th coordinate.
Let $R_{O,k} \triangleq \max (\mathsf{dom}(O)_k) - \min (\mathsf{dom}(O)_k)$.
The observable $O$ is \emph{geometrically local with range $R$} if $R_{O,k} \leq R_k$, for all $k = 1,\dots,d$ and
\begin{equation}
R \triangleq \prod_{k=1}^d R_k.
\end{equation}
In cases when the range $R = \mathcal{O}(1)$ is unimportant, we simply say that $O$ is \emph{geometrically local}.
\end{definition}

\noindent We can now properly state the norm inequality relating the Pauli-1 norm to the spectral norm.

\begin{theorem}[Detailed restatement of Theorem~\ref{thm:main-normineq}]
\label{thm:normineq}
Given an observable $O = \sum_{P} \alpha_P P$ that can be written as a sum of geometrically local observables with range $R$ in a finite $d$-dimensional space, we have
\begin{equation}
\sum_{P} |\alpha_{P}| \leq 2^d R \cdot 4^R \norm{O}_\infty.
\end{equation}
\end{theorem}

\noindent If we additionally require that $\norm{O}_\infty = \mathcal{O}(1)$, we have the following corollary.

\begin{corollary}
\label{corollary:normineq}
Given an observable $O = \sum_{P} \alpha_P P$ with $\norm{O}_\infty = \mathcal{O}(1)$ that can be written as a sum of geometrically local observables in a finite $d$-dimensional space with $R = \mathcal{O}(1)$, we have $\sum_{P} |\alpha_{P}| = \mathcal{O}(1).$
\end{corollary}

In order to establish the above norm inequality, we consider an explicit algorithm for constructing a state $\rho$ satisfying $\sum_{Q} \left| \alpha_Q \right| \leq C \Tr(O \rho)$.
In this way, bounding $\Tr(O\rho)$ above by $\norm{O}_\infty$ gives the desired inequality.
We briefly discuss the idea of the algorithm.
First, we consider the set of all geometrically local blocks over the $n$ qubits.
Then, we consider all Pauli observables $Q$ with nonzero $\alpha_Q$ and the qubits that $Q$ acts on.
For each block, if the qubits that $Q$ acts on are all inside that block, we put $Q$ inside of this block.
If there are multiple such blocks, we choose an arbitrary one to put $Q$ in so that each Pauli observable $Q$ is in exactly one block.
After that, we separate all blocks into a few disjoint layers of blocks.
Each layer contains many blocks that are sufficiently far from one another, and each block contains some Pauli observables.
We select the layer that has the largest $\sum_Q |\alpha_Q|$, where this sum is over all Pauli observables inside that layer.
To construct the state $\rho$, we let $\rho$ be the maximally mixed state on qubits outside of the selected layer.
For each block in the selected layer, we choose $\rho$ to be a state that maximizes the sum of the Pauli terms in the block.
With a careful analysis, the constructed state $\rho$ satisfies the desired norm inequality.

\subsection{Facts and lemmas}

Before proving Theorem~\ref{thm:normineq}, we give a few definitions, facts and lemmas.

\begin{definition}[geometrically local Pauli observables]
\label{def:setglpauli}
Throughout the appendix, we consider $S^{(\mathrm{geo})}$ to be the set of all geometrically local Pauli observables with a constant range $R = \mathcal{O}(1)$.
\end{definition}

\noindent The following fact can be easily shown by considering the Pauli decomposition of each geometrically local observable in the sum.

\begin{fact}
Any observable $O$ that can be written as a sum of geometrically local observables can also be written as a sum of geometrically local Pauli observables. Thus, we can write $O = \sum_P \alpha_P P$, where $\alpha_P = 0$ for all $P \notin S^{(\mathrm{geo})}$.
\end{fact}

\noindent A construction of the mixed state that we are going to use throughout the proof is the following.
The key idea is that $|\Tr(P_i \rho)| = 1/k$ for $i = 1, \ldots, k$, and $|\Tr(P \rho)| = 0$ for any $P \in \{I, X, Y, Z\}^{\otimes n} \setminus \{I, P_1, \ldots, P_k\}$.

\begin{lemma}
\label{lemma:normalizationpsd}
Let $P_1,\dots, P_k \in \{I, X, Y, Z\}^{\otimes n}$. Suppose that $P_i \neq I^{\otimes n}$ for all $i=1,\dots,k$. Then
\begin{equation}
\rho = \frac{I + \frac{\pm P_1 \pm \cdots \pm P_k}{k}}{2^n}
\end{equation}
is a mixed state, i.e., it is positive semidefinite and has unit trace.
\end{lemma}

\begin{proof}
First, we can easily show that $\rho$ has unit trace. Let $P_i = \bigotimes_{j=1}^n P_{i,j}$ for all $i=1,\dots, k$, where $P_{i,j} \in \{I, X, Y, Z\}$. Then, we have
\begin{subequations}
\begin{align}
\Tr(\rho) &= \frac{1}{2^n}\left(\Tr(I) \pm \frac{1}{k}(\Tr(P_1)\pm \cdots \pm \Tr(P_k))\right)\\
&= \frac{1}{2^n}\left(2^n \pm \frac{1}{k}\left(\prod_{j=1}^n \Tr(P_{1,j}) \pm \cdots \pm \prod_{j=1}^n \Tr(P_{k,j})\right)\right)\\
&= 1,
\end{align}
\end{subequations}
where the last equality follows because the trace of a nonidentity Pauli matrix is $0$, and we assume that $P_i \neq I^{\otimes n}$ so that the $P_{i,j}$ are not all identity.
To show that $\rho$ is positive semidefinite, it suffices to prove that the eigenvalues of $(\pm P_1 \pm \cdots \pm P_k)/k$ are between $-1$ and $1$.
Then, when this is summed with the identity matrix which has eigenvalue $+1$, the eigenvalues are nonnegative. We see this using the spectral norm
\begin{equation}
\norm{\frac{\pm P_1 \pm \cdots \pm P_k}{k}}_\infty \leq \frac{1}{k}\left(\norm{P_1}_\infty + \cdots + \norm{P_k}_\infty\right) = 1,
\end{equation}
which concludes our proof.
\end{proof}

\noindent Now, we want to define an operation that is useful throughout the proof.

\begin{definition}[Restriction of a Pauli operator]
\label{def:restriction}
Let $P \in \{I, X, Y, Z\}^{\otimes n}$. Write $P = \bigotimes_{\ell \in \{1,\dots, \lfloor\sqrt[d]{n}\rfloor\}^d} P_\ell$ for $P_\ell \in \{I, X, Y, Z\}$. Let $S \subseteq \{1,\dots, \lfloor \sqrt[d]{n}\rfloor\}^d$ be a subset of qubits. The \emph{restriction of $P$ to the subset of qubits $S$} is the substring of Paulis that act on $S$:
\begin{equation}
\mathsf{restrict}(P; S) \triangleq P_S \triangleq \bigotimes_{\ell\in S} P_\ell.
\end{equation}
\end{definition}

\noindent In Definition~\ref{def:restriction}, the subscript notation is used to be consistent with the more standard notation of $P_k$ to denote a Pauli acting on qubit $k$.

\subsection{Proof of Theorem~\ref{thm:normineq}}

The key idea is to upper bound $\sum_P |\alpha_P|$ by a constant times $\Tr(O\rho)$ for some test state $\rho$.
We construct such a $\rho$ with a similar form to that seen in Lemma~\ref{lemma:normalizationpsd}.
Then, because $\rho$ is positive semidefinite and has unit trace by Lemma~\ref{lemma:normalizationpsd}, $\Tr(O\rho) \leq \norm{O}_\infty$.
Putting everything together, we have
\begin{equation}
\sum_P |\alpha_P| \leq 2^d R \cdot 4^R \Tr(O\rho) \leq 2^d R \cdot 4^R\norm{O}_\infty,
\end{equation}
as required.
Thus, it suffices to consider this intermediate step of finding a quantum state $\rho$ such that $2^d R \cdot 4^R \Tr(O\rho) \geq \sum_P |\alpha_P|$.
To this end, we consider dividing our space of all Pauli observables into different sets and focus on one set, which educates our choice of $\rho$.

Consider some Pauli observable $P \in S^{(\mathrm{geo})}$, where $S^{(\mathrm{geo})}$ is the set of all geometrically local Pauli observables.
Since $P$ is geometrically local, by Definition~\ref{def:gl}, there exist constants $R_k$ for $k = 1,\dots, d$ that serve as the maximum range of qubits that a Pauli observable covers in the $k$th dimension.
We want to divide our $d$-dimensional space into \emph{blocks} of $R_k$ qubits in each dimension.
These blocks of qubits are
\begin{equation}
B_{(\vec{i}, \vec{j})} \triangleq \{\text{qubits}\;\;\ell = (\ell_1,\dots, \ell_d) : \ell_k \in [(2i_k - 2)R_k + j_k + 1, (2i_k-1)R_k + j_k],\;\;\forall k \in \{1,\dots, d\}\},
\end{equation}
where $\vec{i} = (i_1,\dots, i_d)$ and $\vec{j} = (j_1,\dots, j_d)$.
We construct these blocks for $i_k = 1,\dots, \lfloor \frac{\lfloor \sqrt[d]{n}\rfloor -j_k + R_k}{2R_k} \rfloor$ and $j_k = 0,\dots, 2R_k-1$ for $k = 1,\dots, d$.
Here, we are dividing the $d$-dimensional space into blocks of $R = \prod_{k=1}^d R_k$ qubits, where each block is index by $\vec{i}$ and is separated from the next by $R_k$ qubits in the $k$th dimension.
We refer to this gap between the blocks as the \emph{buffer}.
Denote the buffer as
\begin{equation}
B_{\vec{j}}' \triangleq \{1,\dots, \lfloor \sqrt[d][n] \rfloor\}^d \setminus \left(\bigcup_{\vec{i}} B_{(\vec{i},\vec{j})}\right),
\end{equation}
where the union is over all possible vectors $\vec{i}$ such that $i_k$ ranges from $1$ to $\lfloor \frac{\lfloor \sqrt[d]{n} \rfloor - j_k + R_k}{2R_k}\rfloor$.
This separation using the buffer region is so that no Pauli term can act on qubits in two blocks at once, which we use later.
Moreover, we are considering possible shifts of these blocks by $j_k$ qubits in each of the dimensions.
Notice that there are only $2R_k$ possible shifts in each dimension until the blocks align with the original positioning of another block.
Consider the related set consisting of the Pauli terms that act only on qubits in a given block
\begin{equation}
S_{(\vec{i}, \vec{j})} \triangleq \{P : \text{for all qubits $k\in \mathsf{dom}(P)$, then $k \in B_{(\vec{i}, \vec{j})}$}\} \setminus \left(\bigcup_{(\vec{i}',\vec{j}') \leq (\vec{i},\vec{j})} S_{(\vec{i}',\vec{j}')}\right),
\end{equation}
where we define $(\vec{i}',\vec{j}') \leq (\vec{i},\vec{j})$ using the standard lexicographical order, i.e.,
\begin{equation}
(\vec{i}',\vec{j}') = ((i_1',\dots, i_d'), (j_1',\dots, j_d')) \leq (\vec{i},\vec{j}) = ( (i_1,\dots, i_d), (j_1,\dots, j_d))
\end{equation}
if and only if $\vec{i}' < \vec{i}$, or $\vec{i}' = \vec{i}$ and $\vec{j}' \leq \vec{j}$.
Here, $\vec{i}' \leq \vec{i}$ if and only if $i_1' < i_1$, or $i_1' = i_1$ and $i_2' < i_2$, or, etc.
Thus, we create these sets $S_{(i,j)}$ sequentially according to this ordering.
We remove previous sets so that each $S_{(\vec{i},\vec{j})}$ is disjoint from other sets $S_{(\vec{i}', \vec{j}')}$.

Now, taking a union over all $\vec{i}$, we can consider the Pauli terms acting on these blocks together.
The resulting sets then only differ based on the shift of $j_k$ qubits in each dimension.
\begin{equation}
\label{eq:uj}
U_{\vec{j}} \triangleq \bigcup_{\vec{i}} S_{(\vec{i},\vec{j})}.
\end{equation}

\begin{figure}[t]
\centering
\includegraphics[width=0.95\linewidth]{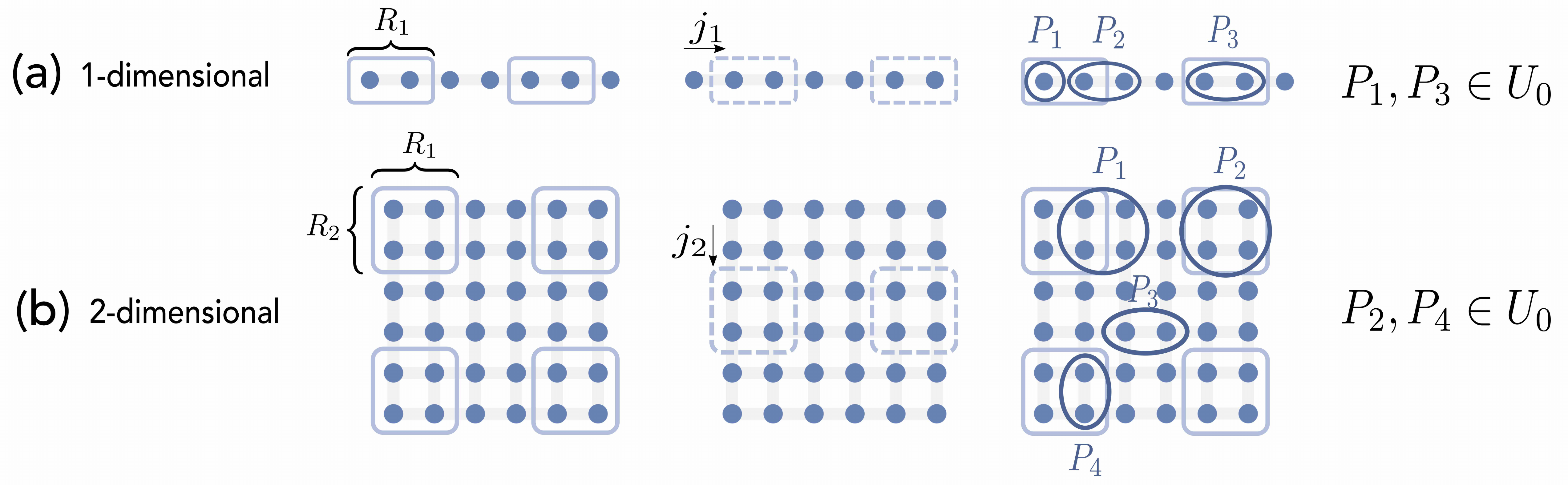}
\caption{\textbf{Intuition behind proof construction of Theorem~\ref{thm:normineq} for the cases of $d = 1$ \textbf{(a)} and $d = 2$ \textbf{(b)}}. In both cases, the idea is to divide our qubits (blue circles) in $d$-dimensional space into blocks (light blue boxes), and consider the quantity we wish to bound in these blocks. Note that all qubits not highlighted are in the buffer region. The first column in the figure depicts the unshifted blocks, i.e., $\vec{j} = 0$. The second column displays an example of shifted blocks (dashed boxes). Finally, the last column considers Pauli terms (dark blue circles) acting on the qubits circled and indicates if they are contained in $U_{0}$, defined in Eq.~\eqref{eq:uj}.}
\label{fig:normineq}
\end{figure}

Figure~\ref{fig:normineq} illustrates all these definitions.
We now consider $\sum_{P \in U_{\vec{j}}} |\alpha_P|$, where these $\alpha_P$ are the 
coefficients in $O = \sum_P \alpha_P P$.
We want to pick the set $U_{\vec{j}}$ such that this sum is largest, i.e.
\begin{equation}
\vec{j}^* \triangleq \argmax_{\substack{0 \leq j_1 \leq 2R_1 - 1\\ \vdots \\ 0 \leq j_d \leq 2R_d - 1 }} \sum_{P \in U_{\vec{j}}} |\alpha_P|.
\end{equation}
We now focus on the set $U_{\vec{j}^*}$. To justify this choice, we can think of each of the sets shifted by $\vec{j}$ as breaking up the sum $\sum_P |\alpha_P|$ into different disjoint sums. This is a result of our earlier choice for $U_{\vec{j}}$ to contain a disjoint set of Pauli terms. Then, the maximum over all shifts $\vec{j}$ of $\sum_{P \in U_{\vec{j}}} |\alpha_P|$ is greater than the average over all shifts. In other words, we have
\begin{equation}
\sum_{P \in U_{\vec{j}^*}} |\alpha_P| \geq \frac{1}{2^d R} \sum_P |\alpha_P|,
\end{equation}
where recall that $R = \prod_{k=1}^d R_k$.
Relating back to our original goal, it remains to find a test state $\rho$ such that
\begin{equation}
\Tr(O\rho) \geq \frac{1}{4^R} \sum_{P \in U_{\vec{j}^*}} |\alpha_P|.
\end{equation}
Once we have this, we can conclude that
\begin{equation}
2^d R \cdot 4^R \norm{O}_\infty \geq 2^d R \cdot 4^R \Tr(O\rho) \geq 2^dR \sum_{P \in U_{\vec{j}^*}} |\alpha_P| \geq \sum_P |\alpha_P|,
\end{equation}
proving our claim, where the first inequality follows because $\rho$ is positive semidefinite and has unit trace from Lemma~\ref{lemma:normalizationpsd}.
In what follows, we aim to define this $\rho$ based on the set $U_{\vec{j}^*}$ and show that this inequality holds.

The idea is to have $\rho$ as the maximally mixed state on qubits in the buffer region $B'_{\vec{j}^*}$ and be a state of the form in Lemma~\ref{lemma:normalizationpsd} for qubits in $\bigcup_{\vec{i}} B_{(\vec{i},\vec{j}^*)}$.
In this way, when we take $\Tr(O\rho)$, any Pauli terms not in $U_{\vec{j}^*}$ contribute $0$ while Pauli terms $P$ in $U_{\vec{j}^*}$ contribute a constant times $|\alpha_P|$.
Explicitly, we define $\rho$ as
\begin{equation}
\rho \triangleq \bigotimes_{\vec{i}} \left(\frac{1}{2^R}\left(I_{B_{(\vec{i},\vec{j}^*)}} + \frac{1}{|S_{(\vec{i},\vec{j}^*)}|} \sum_{Q \in S_{(\vec{i},\vec{j}^*)}} (-1)^{\delta_{\{\alpha_Q < 0\}}} Q_{B_{(\vec{i},\vec{j}^*)}}\right)\right) \bigotimes_{\ell \in B'_{\vec{j}^*}} \frac{I_\ell}{2} \triangleq \bigotimes_{\vec{i}} \rho_{\vec{i}} \bigotimes_{\ell \in B_{\vec{j}^*}'} \frac{I_\ell}{2},
\end{equation}
where $\delta_{\{\alpha_Q < 0\}}$ is $1$ when $\alpha_Q < 0$ and $0$ otherwise.
Also, the tensor product is again over all possible vectors $\vec{i}$ such that each entry $i_k$ ranges from $1$ to $\lfloor \frac{\lfloor \sqrt[d]{n} \rfloor - j_k + R_k}{2R_k} \rfloor$.
Here, we are using the notation from Definition~\ref{def:restriction} to denote quantum operations restricted to their action on a given set of qubits.
By Lemma~\ref{lemma:normalizationpsd}, $\rho_{\vec{i}}$ is a proper quantum state that is positive semidefinite and has unit trace; hence $\rho$ is a quantum state.
Now, we want to calculate $\Tr(O\rho)$. Recall that $O = \sum_P \alpha_P P$.
Taking the trace, we have
\begin{equation}
\Tr(O\rho) = \sum_P \alpha_P \Tr(P\rho).
\end{equation}
There are four cases that can occur regarding $\mathsf{dom}(P)$.
\begin{enumerate}
\item \label{itm:case1} $P$ acts nontrivially on some qubits in the buffer region, i.e., $\mathsf{dom}(P) \subseteq B_{(\vec{i}_P,\vec{j}^*)} \cup B'_{\vec{j}^*}$ for some $\vec{i}_P$.
\item \label{itm:case2} $P$ acts trivially on all qubits in the buffer region, but $P$ acts nontrivially on qubits in two or more blocks, i.e., $\mathsf{dom}(P) \subseteq B_{(\vec{i}_{P,1}, \vec{j}^*)} \cup B_{(\vec{i}_{P,2}, \vec{j}^*)}$ for some $\vec{i}_{P,1}, \vec{i}_{P,2}$.
\item \label{itm:case3} $P$ acts nontrivially only on qubits in a single block but $P$ is not in the set $U_{\vec{j}^*}$, i.e., $\mathsf{dom}(P) \subseteq B_{(\vec{i}_P, \vec{j}^*)}$ for some $\vec{i}_P$, but $P \notin S_{(\vec{i}_P, \vec{j}^*)}$.
\item \label{itm:case4} $P$ acts nontrivially only on qubits in a single block and $P$ is in the set $U_{j^*}$, i.e., $\mathsf{dom}(P) \subseteq B_{(\vec{i}_P, \vec{j}^*)}$ for some $\vec{i}_P$ and $P \in S_{(\vec{i}_P, \vec{j}^*)}$.
\end{enumerate}
We compute $\Tr(P \rho)$ for each of these cases.
We note that $\Tr(\rho_{\vec{i}}) = 1$ in all these calculations.

For Case~\ref{itm:case1}, it suffices to consider the case where $P$ acts on only one qubit in the buffer region, i.e. there exists a qubit $\ell^* \in B'_{\vec{j}^*}$ such that $\ell^* \in \mathsf{dom}(P)$ and $\mathsf{dom}(P) \setminus \{\ell^*\} \subseteq B_{(\vec{i}_P, \vec{j}^*)}$.
Then, the state $P \rho$ is as follows
\begin{subequations}
\begin{align}
P\rho &= \left(\bigotimes_{\vec{i}\neq \vec{i}_P} \rho_{\vec{i}} \bigotimes_{\substack{\ell \in B_{\vec{j}^*}'\\\ell\neq \ell^*}} \frac{I_\ell}{2}\right)\\
&\otimes \left(\frac{1}{2^R}\left(P_{B_{(\vec{i}_P, \vec{j}^*)}} + \frac{1}{|S_{(\vec{i}_P, \vec{j}^*)}|}\sum_{Q \in S_{(\vec{i}_P,\vec{j}^*)}} (-1)^{\delta_{\{\alpha_Q < 0\}}} P_{B_{(\vec{i}_P, \vec{j}^*)}} Q_{B_{(\vec{i}_P, \vec{j}^*)}}\right) \otimes \frac{P_{\ell^*}}{2}\right),
\end{align}
\end{subequations}
where we are again using the notation from Definition~\ref{def:restriction}.
Taking the trace of this state, since the trace of $I/2$ is $1$, we have
\begin{subequations}
\begin{align}
\Tr(P\rho) &= \frac{\Tr\left(\frac{P_{\ell^*}}{2}\right)}{2^R} \Biggr(\Tr(P_{(B_{(\vec{i}_P, \vec{j}^*)})})\\
&+ \frac{1}{|S_{(\vec{i}_P,\vec{j}^*)}|}\sum_{Q \in S_{(\vec{i}_P, \vec{j}^*)}} (-1)^{\delta_{\{\alpha_Q < 0\}}} \Tr(P_{B_{(\vec{i}_P, \vec{j}^*)}} Q_{B_{(\vec{i}_P, \vec{j}^*)}})\Biggr)\\
&= 0,
\end{align}
\end{subequations}
where the last equality follows because the trace of a nonidentity Pauli string is $0$ and
\begin{equation}
    \Tr(P_{B_{(\vec{i}_P, \vec{j}^*)}}Q_{B_{(\vec{i}_P, \vec{j}^*)}}) = 2^R\delta_{\{P_{B_{(\vec{i}_P, \vec{j}^*)}} = Q_{B_{(\vec{i}_P, \vec{j}^*)}}\}}.
\end{equation} Here, $P_{B_{(\vec{i}_P, \vec{j}^*)}} \neq Q_{B_{(\vec{i}_P, \vec{j}^*)}}$ because $Q \in S_{(\vec{i}_P, \vec{j}^*)}$ so that $Q$ acts nontrivially only on qubits in $B_{(\vec{i}_P, \vec{j}^*)}$ while $P$ acts nontrivially on $\ell^* \notin B_{(\vec{i}_P,\vec{j}^*)}$.
Thus, Case~\ref{itm:case1} contributes $0$ to $\Tr(O\rho)$.

Next, we consider Case~\ref{itm:case2}.
In Case~\ref{itm:case2}, we consider what happens if $P$ acts nontrivially on qubits in more than one block, i.e., $\mathsf{dom}(P) \subseteq B_{(\vec{i}_{P,1}, \vec{j}^*)} \cup B_{(\vec{i}_{P,2}, \vec{j}^*)}$.
However, this case is in fact not possible by construction because the buffer region between $B_{(\vec{i}_{P,1}, \vec{j}^*)}$ and $B_{(\vec{i}_{P,2}, \vec{j}^*)}$ is of size $R_k$ in each of the dimensions.
Recall that $R_k$ is the largest distance between two qubits that any $P$ acts on in the $k$th dimension.
Thus, it is not possible for $P$ to span across the buffer region, so this case cannot occur. Hence, it trivially contributes $0$ to $\Tr(O\rho)$.

Now, we consider Case~\ref{itm:case3}.
From the previous two cases, we see that $P$ can only act nontrivially on qubits in a single block $B_{(\vec{i}_P, \vec{j}^*)}$ to contribute to $\Tr(O\rho)$.
However, by construction of the sets $S_{(\vec{i},\vec{j})}$, in order to make them disjoint, it is possible that $P \notin S_{(\vec{i}_P, \vec{j}^*)}$ despite it acting on the correct block of qubits.
We show that this also contributes $0$ to $\Tr(O\rho)$.
Taking the trace, we have
\begin{equation}
\Tr(P \rho) = \frac{1}{2^R} \left(\Tr(P_{(B_{(\vec{i}_P, \vec{j}^*)})}) + \frac{1}{|S_{(\vec{i}_P,\vec{j}^*)}|}\sum_{Q \in S_{(\vec{i}_P, \vec{j}^*)}} (-1)^{\delta_{\{\alpha_Q < 0\}}} \Tr(P_{B_{(\vec{i}_P, \vec{j}^*)}} Q_{B_{(\vec{i}_P, \vec{j}^*)}})\right) = 0,
\end{equation}
where the last equality follows because the trace of a nonidentity Pauli string is $0$ and
\begin{equation}
\Tr(P_{B_{(\vec{i}_P, \vec{j}^*)}}Q_{B_{(\vec{i}_P, \vec{j}^*)}}) = 2^R\delta_{\{P_{B_{(\vec{i}_P, \vec{j}^*)}} = Q_{B_{(\vec{i}_P, \vec{j}^*)}}\}}.
\end{equation}
Here, $P_{B_{(\vec{i}_P, \vec{j}^*)}} \neq Q_{B_{(\vec{i}_P, \vec{j}^*)}}$ because $Q \in S_{(\vec{i}_P, \vec{j}^*)}$ while we know from this case that $P \notin S_{(\vec{i}_P, \vec{j}^*)}$.

Finally, we consider Case~\ref{itm:case4}.
From the previous cases, we see that the only remaining possibility is that $P$ acts nontrivially on qubits in a single block $B_{(\vec{i}_P, \vec{j}^*)}$ and is also contained in a set $S_{(\vec{i}_P, \vec{j}^*)}$.
Computing the trace, we have
\begin{subequations}
\begin{align}
\Tr(P\rho) &= \frac{1}{2^R} \left(\Tr(P_{(B_{(\vec{i}_P, \vec{j}^*)})}) + \frac{1}{|S_{(\vec{i}_P,\vec{j}^*)}|}\sum_{Q \in S_{(\vec{i}_P, \vec{j}^*)}} (-1)^{\delta_{\{\alpha_Q < 0\}}} \Tr(P_{B_{(\vec{i}_P, \vec{j}^*)}} Q_{B_{(\vec{i}_P, \vec{j}^*)}})\right)\\
&= \frac{1}{2^R} \frac{1}{|S_{(\vec{i}_P, \vec{j}^*)}|} \sum_{Q \in S_{(\vec{i}_P, \vec{j}^*)}} (-1)^{\delta_{\{\alpha_Q < 0\}}} \Tr(P_{B_{(\vec{i}_P, \vec{j}^*)}} Q_{B_{(\vec{i}_P, \vec{j}^*)}})\\
&= \frac{1}{|S_{(\vec{i}_P, \vec{j}^*)}|} (-1)^{\delta_{\{\alpha_P < 0\}}}.
\end{align}
\end{subequations}
Here, we are using that the trace of a nonidentity Pauli string is $0$ and $\Tr(P_{B_{(\vec{i}_P, \vec{j}^*)}}Q_{B_{(\vec{i}_P, \vec{j}^*)}}) = 2^R\delta_{\{P_{B_{(\vec{i}_P, \vec{j}^*)}} = Q_{B_{(\vec{i}_P, \vec{j}^*)}}\}}$.
Because $P \in S_{(\vec{i}_P,\vec{j}^*)}$, there exists a $Q \in S_{(\vec{i}_P, \vec{j}^*)}$ such that $P = Q$ so that the sum over $S_{(\vec{i}_P, \vec{j}^*)}$ then collapses to this $P$.
Thus, for this case, $P$ contributes a nonzero amount to $\Tr(O \rho)$.
Summing over all $P \in U_{\vec{j}^*}$, we have a total contribution of
\begin{equation}
\sum_{P \in U_{\vec{j}^*}} \alpha_P \Tr(P\rho) = \sum_{P \in U_{\vec{j}^*}} \frac{1}{|S_{(\vec{i}_P, \vec{j}^*)}|}(-1)^{\delta_{\{\alpha_P < 0\}}} \alpha_P = \sum_{P \in U_{\vec{j}^*}} \frac{1}{|S_{(\vec{i}_P, \vec{j}^*)}|} |\alpha_P|.
\end{equation}
Thus, putting everything together, only Case~\ref{itm:case4} contributed a nonzero amount to $\Tr(O\rho)$, so we have
\begin{equation}
\Tr(O \rho) = \sum_{P \in U_{\vec{j}^*}} \frac{1}{|S_{(\vec{i}_P, \vec{j}^*)}|} |\alpha_P|.
\end{equation}
Here, we know that
\begin{equation}
|S_{(\vec{i}_P, \vec{j}^*)}| \leq 4^R
\end{equation}
because $|B_{(\vec{i}_P, \vec{j}^*)}| = R$ and on $R = \prod_k R_k$ qubits, there are $4^R$ possible Pauli terms (i.e., $I, X, Y, Z$ on each of the qubits). Then, we have
\begin{equation}
\Tr(O\rho) \geq \frac{1}{4^R}\sum_{P \in U_{\vec{j}^*}} |\alpha_P|.
\end{equation}
As explained previously, this suffices to conclude the proof.\qed

\section{ML algorithm and sample complexity}
\label{sec:algorithm}

In this section, we present our machine learning algorithm and prove that it can approximate $\Tr(O\rho(x))$ given training data size $N$ scaling logarithmically in system size $n$.
To do so, we leverage the results in Appendix~\ref{sec:simple} and Appendix~\ref{sec:norminequality} heavily.

Recall that we consider an unknown family of $n$-qubit geometrically local Hamiltonians $\{H(x) : x \in [-1,1]^m\}$ in a finite $d$-dimensional space such that $H(x) = \sum_{j=1}^L h_j(\vec{x}_j)$, where $\vec{x}_j \in \mathbb{R}^q, q = \mathcal{O}(1)$, and $x$ is the concatenation of the $L$ vectors $\vec{x}_1,\dots,\vec{x}_L$.
We also assume that the spectral gap of $H(x)$ is lower bounded by a constant $\gamma$ over $[-1,1]^m$ and $\rho(x)$ is the ground state of $H(x)$.
We also consider an unknown observable $O$ with $\norm{O}_\infty = \mathcal{O}(1)$ that can be written as a sum of geometrically local observables and an arbitrary unknown distribution $\mathcal{D}$ over $[-1, 1]^m$.

In what follows, we first present a full description of the proposed ML algorithm in Appendix~\ref{sec:ml-algorithm}.
In Appendix~\ref{sec:rigor-guarantee}, we then state the rigorous guarantee achieved by this ML algorithm.
Next, we find a bound required to utilize $\ell_1$-regularized regression in Appendix~\ref{sec:1norm}.
Then, we bound the in-sample error on the training data in Appendix~\ref{sec:training} by showing that the function $g(x)$ for approximating $\Tr(O \rho(x))$ as defined in Eq.~\eqref{eq:g} of Appendix~\ref{sec:simple} achieves small training error.
Finally, we use standard results in machine learning theory to bound the prediction error in Appendix~\ref{sec:prediction}.

\subsection{ML algorithm}
\label{sec:ml-algorithm}

This section is dedicated to describing the ML algorithm in detail.
Let $1/e > \epsilon_1, \epsilon_2, \epsilon_3 > 0$.
The ML algorithm is also given training data $\{(x_\ell, y_\ell)\}_{\ell=1}^N$ consisting of parameters $x_\ell$ sampled from an arbitrary unknown distribution $\mathcal{D}$ over $[-1,1]^m$ along with an estimator $y_\ell$ of $\Tr(O\rho(x_\ell))$ such that $|y_\ell - \Tr(O\rho(x_\ell))| \leq \epsilon_2$.

Given this, we first redefine several notions from Appendix~\ref{sec:simple} in terms of $\epsilon_1$.
We utilize these definitions in the remainder of Appendix~\ref{sec:algorithm}.
We begin by redefining $\delta_1$ and $I_P$, originally defined in Def.~\ref{def:delta},~\ref{def:ip}, respectively.
Define $\delta_1$ as
\begin{equation}
\label{eq:delta1}
\delta_1 \triangleq \max\left(C_{\mathrm{max}}\log^2(2C/\epsilon_1), C_4, C_5, \frac{\max(5900, \alpha, 7(d + 11), \theta)}{b}\right),
\end{equation}
where all constants $b, v_{\mathrm{lr}}, C_{\mathrm{max}}, C_4, C_5, \alpha, \theta$ are defined as in Def.~\ref{def:delta} and $C$ is defined in Lemma~\ref{lemma:approxlocal}.
Using this definition of $\delta_1$, let $I_P$ be defined as
\begin{equation}
    I_P \triangleq \{c \in \{1,\dots, m\} : d_{\mathrm{obs}}(h_{j(c)}, P) \leq \delta_1\},
\end{equation}
as in Eq.~\eqref{eq:ip}.
Now, we can redefine the quantities from Def.~\ref{def:discretization} used to approximate the ground state property as a sum of discretized functions.
Let $\delta_2$ be given by
\begin{equation}
    \label{eq:delta2}
    \delta_2 \triangleq \frac{1}{\left\lceil \frac{2\sqrt{C'|I_P|}}{\epsilon_1} \right\rceil},
\end{equation}
where $C'$ is defined in Lemma~\ref{lemma:approxlocal}.
From this, we can define the discretized parameter space $X_P$, which contains parameter vectors that are $0$ outside of $I_P$ and take on discrete values inside of $I_P$:
\begin{equation}
    X_P \triangleq \left.\begin{cases}
    x \in [-1,1]^m : \text{if } c \notin I_P, x_{c} = 0\\
    \hspace{62pt} \text{if } c \in I_P, x_{c} \in \left\{0, \pm \delta_2, \pm 2\delta_2,\dots, \pm 1\right\}
    \end{cases}\right\}.
\end{equation}
Furthermore, for each discretized vector $x' \in X_P$, let $T_{x,P}$ be the set of vectors close to $x'$ for coordinates in $I_P$:
\begin{equation}
    T_{x, P} \triangleq \left\{x' \in [-1,1]^m : -\frac{\delta_2}{2} < x_c - x_c' \leq \frac{\delta_2}{2} ,\;\forall c \in I_P\right\}.
\end{equation}
Finally, we define an additional hyperparameter $B > 0$ as
\begin{equation}
    B \triangleq 2^{\mathcal{O}(\mathrm{polylog}(1/\epsilon_1))}.
\end{equation}

With these definitions in place, we can discuss the ML algorithm. At a high level, the algorithm first maps the parameter space into a high-dimensional feature space.
Then, the ML algorithm learns a linear function in this feature space using $\ell_1$-regularized regression.

In particular, the feature map $\phi$ maps $x \mapsto \phi(x)$, where $x \in [-1,1]^m$ is an $m$-dimensional vector while $\phi(x) \in \mathbb{R}^{m_\phi}$ is an $m_\phi$-dimensional vector with
\begin{equation}
    m_\phi \triangleq \sum_{P \in S^{\mathrm{geo}}} |X_P|.
\end{equation}
Here, $S^{(\mathrm{geo})}$ denotes the set of all geometrically local Pauli observables as in Def.~\ref{def:setglpauli}.
Each coordinate of $\phi(x)$ is indexed by $x' \in X_P, P \in S^{(\mathrm{geo})}$ and is defined as
\begin{equation}
\phi(x)_{x', P} \triangleq \indicator[x \in T_{x', P}].
\end{equation}

The hypothesis class for our proposed ML algorithm consists of linear functions in this feature space, i.e., functions of the form $h(x) = \mathbf{w} \cdot \phi(x)$.
The classical ML model learns such a function using $\ell_1$-regularized regression (LASSO)~\cite{doi:10.1137/0907087,tibshirani1996regression,mohri2018foundations} over the feature space.
Namely, given the hyperparameter $B > 0$ defined above, we utilize LASSO to find an $m_\phi$-dimensional vector $\mathbf{w}^*$ from the following optimization problem that minimizes the training error $\frac{1}{N} \sum_{\ell=1}^N \left| \mathbf{w} \cdot \phi(x_\ell) - y_\ell \right|^2$,
\begin{equation}
\min_{\substack{\mathbf{w} \in \mathbb{R}^{m_\phi}\\ \norm{\mathbf{w}}_1 \leq B} } \, \frac{1}{N} \sum_{\ell=1}^N \left| \mathbf{w} \cdot \phi(x_\ell) - y_\ell \right|^2,
\end{equation}
where $y_\ell$ approximates $\Tr(O\rho(x_\ell))$.
We denote the learned function by $h^*(x) = \mathbf{w}^* \cdot \phi(x)$.
Importantly, this learned function does not need to achieve the minimum training error.
In the following, we consider the vector $\mathbf{w}^*$ to yield a training error that is larger than the minimum training error by at most $\epsilon_3/2$.

\subsection{Rigorous guarantee}
\label{sec:rigor-guarantee}
Given these definitions and the ML algorithm, we prove the following theorem.
The theorem stated in the main text corresponds to $\epsilon_1 = 0.2 \epsilon$, $\epsilon_2 = \epsilon$, and $\epsilon_3 = 0.4 \epsilon$. Hence $(\epsilon_1 + \epsilon_2)^2 \leq 1.44 \epsilon^2 \leq 0.53 \epsilon$ and $(\epsilon_1 + \epsilon_2)^2 + \epsilon_3 \leq \epsilon$.

\begin{theorem}
\label{thm:algorithm}
Let $1/e > \epsilon_1, \epsilon_2, \epsilon_3 > 0$ and $\delta > 0$.
Given training data $\{ (x_\ell, y_\ell) \}_{\ell=1}^N$ of size
\begin{equation}
N = \log(n / \delta) 2^{\mathcal{O}(\log(1 / \epsilon_3) + \mathrm{polylog}(1/\epsilon_1))},
\end{equation}
where $x_\ell$ is sampled from $\mathcal{D}$ and $y_\ell$ is an estimator of $\Tr(O\rho(x_\ell))$ such that $|y_\ell - \Tr(O\rho(x_\ell))| \leq \epsilon_2$,
the ML algorithm can produce $h^*(x)$ that achieves prediction error
\begin{equation}
\E_{x \sim \mathcal{D}} |h^*(x) - \Tr(O\rho(x))|^2 \leq (\epsilon_1 + \epsilon_2)^2 + \epsilon_3
\end{equation}
with probability at least $1 - \delta$.
The training time for constructing the hypothesis function $h$ and the prediction time for computing $h^*(x)$ are upper bounded by $\mathcal{O}(nN) = n\log(n / \delta) 2^{\mathcal{O}(\log(1 / \epsilon_3) + \mathrm{polylog}(1/\epsilon_1))}$.
\end{theorem}

In the ML problem formulated in Sec.~\ref{subsec:ML-model} and Appendix~\ref{sec:ml-algorithm}, the training data $\{x_\ell, y_\ell\}_{\ell=1}^N$ corresponds to a fixed and unknown observable $O$.
However, we may be interested in training an ML model that can predict $\Tr(O \rho(x))$ for a wide range of observables $O$.
In this setting, one could consider a classical dataset $\{x_\ell, \sigma_T(\rho(x_\ell))\}_{\ell = 1}^N$ generated by performing classical shadow tomography \cite{huang2020predicting, elben2020mixed, elben2022randomized, wan2022matchgate, bu2022classical} on the ground state $\rho(x_\ell)$ for each $x_\ell$ in $\ell = 1, \ldots, N$.
This is achieved by repeatedly performing $T$ randomized Pauli measurements on each state $\rho(x_\ell)$.
Using the classical shadow dataset, we can obtain the following corollary for predicting ground state representations.

\begin{corollary}
\label{corollary:ground-state-rep-full}
Let $1/e > \epsilon_1, \epsilon_2, \epsilon_3 > 0$ and $\delta > 0$.
Given a training data set $\{x_\ell, \sigma_T(\rho(x_\ell))\}_{\ell = 1}^N$ of size
\begin{equation}
N = \log(n / \delta) 2^{\mathcal{O}(\log(1/\epsilon_3) + \mathrm{polylog}(1/\epsilon_1))},
\end{equation}
where $x_\ell$ is sampled from an unknown distribution $\mathcal{D}$ and $\sigma_T(\rho(x_\ell))$ is the classical shadow representation of the ground state $\rho(x_\ell)$ using $T$ randomized Pauli measurements. For $T = \mathcal{O}(\log(nN / \delta)/\epsilon_2^2) = \tilde{\mathcal{O}}(\log(n / \delta) / \epsilon_2^2)$, the proposed ML algorithm can learn a ground state representation $\hat{\rho}_{N, T}(x)$ that achieves
\begin{equation}
\E_{x \sim \mathcal{D}}|\Tr(O\hat{\rho}_{N, T}(x)) - \Tr(O\rho(x))|^2 \leq (\epsilon_1 + \epsilon_2)^2 + \epsilon_3
\end{equation}
for any observable $O$ with eigenvalues between $-1$ and $1$ that can be written as a sum of geometrically local observables with probability at least $1-\delta$.
\end{corollary}
\begin{proof}
For any observable $O$ with eigenvalues between $-1$ and $1$ that can be written as a sum of geometrically local observables, we have $O = \sum_{P \in S^{(\mathrm{geo})}} \alpha_P P$, where $S^{(\mathrm{geo})}$ is the set of all geometrically local Pauli observables.
From Corollary~\ref{corollary:normineq}, we have
\begin{equation} \label{eq:alphaPbound}
    \sum_{P \in S^{(\mathrm{geo})}} |\alpha_P| \leq C
\end{equation}
for a constant $C$. We are going to use the constant $C$ to set the training data size $N$ and the number of randomized measurements $T$.
In particular, we are going to consider
\begin{align}
    N &= \log(n / (\delta / (2 |S^{(\mathrm{geo})}|) ) ) 2^{\mathcal{O}\left(\log(C^2 / \epsilon_3) + \mathrm{polylog}(C / \epsilon_1)\right)}, \\
    T &= \mathcal{O}\left(\log(nN / (\delta / (2 |S^{(\mathrm{geo})}|) ) ) / (\epsilon_2 / C)^2\right).
\end{align}
For any geometrically local Pauli observables $P \in S^{(\mathrm{geo})}$, we can use the classical shadow dataset $\{x_\ell, \sigma_T(\rho(x_\ell))\}_{\ell = 1}^N$ to estimate the expectation value of $P$ in the ground states $\rho(x_\ell)$ for all $\ell$.
This creates a dataset $\{x_\ell, y_\ell^{(P)}\}_{\ell = 1}^N$.
Under the specified $T$, Lemma 1 in~\cite{huang2021provably} guarantees that
\begin{equation} \label{eq:y-is-good}
    \left| y_\ell^{(P)} - \Tr(P \rho(x_\ell)) \right| < \frac{\epsilon_2}{C},
\end{equation}
for all $\ell = 1, \ldots, N$ and $P \in S^{(\mathrm{geo})}$ with probability at least $1 - (\delta / 2)$.
For each $P \in S^{(\mathrm{geo})}$, we consider $h^*_P(x)$ to be the function produced from Theorem~\ref{thm:algorithm}.
From Theorem~\ref{thm:algorithm}, we have
\begin{equation} \label{eq:prediction-error-repre}
    \E_{x \sim \mathcal{D}} | h^*_{P}(x) - \Tr(P \rho(x)) |^2 \leq \frac{1}{C^2}\left[ (\epsilon_1 + \epsilon_2)^2 + \epsilon_3 \right]
\end{equation}
for all $P \in S^{(\mathrm{geo})}$ with probability at least $1 - (\delta / 2)$ conditioned on the event given in Eq.~\eqref{eq:y-is-good} occurs.
Using the union bound to combine the two events considered in Eq.~\eqref{eq:y-is-good} and Eq.~\eqref{eq:prediction-error-repre}, we can ensure that Eq.~\eqref{eq:prediction-error-repre} holds with probability at least $1 - \delta$.

We define the ground state representation produced by the ML algorithm to be
\begin{equation}
    \hat{\rho}_{N, T}(x) \triangleq \sum_{P \in S^{(\mathrm{geo})}} h^*_P(x) \left(\frac{P}{2^n}\right).
\end{equation}
For the observable $O$, we have
\begin{equation}
    |\Tr(O\hat{\rho}_{N, T}(x)) - \Tr(O\rho(x))| \leq \sum_{P \in S^{(\mathrm{geo})}} |\alpha_P| | h^*_P(x) - \Tr(P \rho(x)) |.
\end{equation}
By the Cauchy-Schwarz inequality, Eq.~\eqref{eq:alphaPbound}, and Eq.~\eqref{eq:prediction-error-repre}, we have
\begin{align}
    &\E_{x \sim \mathcal{D}}|\Tr(O\hat{\rho}_{N, T}(x)) - \Tr(O\rho(x))|^2\\
    &\leq \sum_{P_1, P_2 \in S^{(\mathrm{geo})}} |\alpha_{P_1}| |\alpha_{P_2}| \E_{x \sim \mathcal{D}} | h^*_{P_1}(x) - \Tr(P_1 \rho(x)) | | h^*_{P_2}(x) - \Tr(P_2 \rho(x)) |\\
    &\leq \sum_{P_1, P_2 \in S^{(\mathrm{geo})}} |\alpha_{P_1}| |\alpha_{P_2}| \sqrt{\E_{x \sim \mathcal{D}} | h^*_{P_1}(x) - \Tr(P_1 \rho(x)) |^2} \sqrt{\E_{x \sim \mathcal{D}} | h^*_{P_2}(x) - \Tr(P_2 \rho(x)) |^2} \\
    &\leq \left(\sum_{P \in S^{(\mathrm{geo})}} |\alpha_{P}|\right)^2 \frac{1}{C^2} \left[ (\epsilon_1 + \epsilon_2)^2 + \epsilon_3\right] \leq (\epsilon_1 + \epsilon_2)^2 + \epsilon_3.
\end{align}
This concludes the proof of the corollary.
\end{proof}

\subsection{$\ell_1$-Norm bound on coefficients of linear hypothesis}
\label{sec:1norm}

We now justify our choice of the hyperparameter $B > 0$ such that $\norm{\mathbf{w}}_1 \leq B$.
From Appendix~\ref{sec:simple}, we constructed a function that approximates the ground state property. Explicitly, this function is defined as
\begin{equation}
g(x) \triangleq \sum_{P \in S^{\mathrm{(geo)}}} \sum_{x' \in X_P} f_P(x') \indicator[x \in T_{x',P}] = \mathbf{w}' \cdot \phi(x),
\end{equation}
where in this case, the vector of coefficients $\mathbf{w}'$, indexed by $x' \in X_P, P \in S^{(\mathrm{geo})}$, is defined as
\begin{equation}
\label{eq:w'}
\mathbf{w}'_{x', P} \triangleq f_P(x').
\end{equation}
Thus, we see that the ML model, which learns functions of this form, has the capacity to approximate the target ground state property $\Tr(O\rho(x))$.
The actual function we learn, $h^*(x) = \mathbf{w}^* \cdot \phi(x)$ could differ significantly from $g(x) = \mathbf{w}' \cdot \phi(x)$ because $\mathbf{w}'$ is unknown.
Nevertheless, we can utilize an upper bound $\norm{\mathbf{w}'}_1 \leq B$ to restrict the hypothesis set of the ML algorithm to functions of the form $h(x) = \mathbf{w} \cdot \phi(x)$ such that $\norm{\mathbf{w}}_1 \leq B$.
Thus, we find an upper bound on $\norm{\mathbf{w}'}_1$ in the following lemma.

\begin{lemma}[$\ell_1$-Norm bound]
\label{lemma:1norm}
Let $\mathbf{w}'$ be the vector of coefficients defined in Eq.~\eqref{eq:w'}. Then, we have the following bound on $\norm{\mathbf{w}'}_1$:
\begin{equation}
\norm{\mathbf{w}'}_1 = \sum_{P \in S^{\mathrm{(geo)}}} \sum_{x' \in X_P} |f_P(x')| = 2^{\mathcal{O}(\mathrm{polylog}(1/\epsilon_1))}.
\end{equation}
\end{lemma}

\begin{proof}
First, we can analyze the $|f_P(x')|$ term. Recall that $f_P$ is just $\alpha_P \Tr(P\rho(\chi_P(x)))$, where $\chi_P(x) \in [-1, 1]^m$ sets parameters outside of $I_P$ to $0$. Thus, we can bound its absolute value by
\begin{equation}
|f_P(x)| = |\alpha_P| |\Tr(P\rho(\chi_P(x)))| \leq |\alpha_P|.
\end{equation}
Plugging this into the $\ell_1$-norm of $\mathbf{w}'$, we have
\begin{subequations}
\begin{align}
    \norm{\mathbf{w}'}_1 &= \sum_{P \in S^{\mathrm{(geo)}}} \sum_{x' \in X_P} |f_P(x')|\\
    &\leq \sum_{P \in S^{\mathrm{(geo)}}} |X_P| |f_P(x')|\\
    \label{eq:1normunfin}
    &\leq \max_{P \in S^{\mathrm{(geo)}}} |X_P| \sum_{Q \in S^{(\mathrm{geo})}}|\alpha_Q|.
\end{align}
\end{subequations}
Thus, it suffices to count the number of elements in $X_P$.
Recall in Definition~\ref{def:discretization} that $X_P$ is defined such that the parameter values for $c' \notin I_P$ are fixed to $0$ while for $c' \in I_P$, $x_{c'}$ can be any value in $\{0, \pm \delta_2, \pm 2\delta_2,\dots, \pm 1\}$.
Hence, it is clear that
\begin{equation}
|X_P| \leq \left|\left\{0, \pm \delta_2, \pm 2\delta_2,\dots, \pm 1\right\}\right|^{|I_P|} \leq \left(\frac{2}{\delta_2} + 1\right)^{|I_P|}.
\end{equation}
Moreover, by our choice of $\delta_2$ in Eq.~\eqref{eq:delta2}, we have
\begin{equation}
\label{eq:xpsizeunfin}
|X_P| \leq \left(2\left\lceil \frac{2\sqrt{C'|I_P|}}{\epsilon_1} \right\rceil + 1\right)^{|I_P|}.
\end{equation}
Now, it remains to bound the size of $I_P$, defined in Eq.~\eqref{eq:ip}.
This size is simply the number of parameters that $h_j$ depends on for some $h_j$ with $d_{\mathrm{obs}}(h_j, P) \leq \delta_1$.
By Eq.~\eqref{eq:localtermsbound}, we can bound the number of such $h_j$:
\begin{equation}
\sum_{j : d_{\mathrm{obs}}(h_j, P) \leq \delta_1} 1 \leq b_d + c_d\delta_1^d.
\end{equation}
Moreover, we assume that each $h_j$ depends on $\mathcal{O}(1)$ parameters. Suppose that each $h_j$ depends on at most $q$ parameters. Then, we can bound the size of $I_P$ by
\begin{equation}
|I_P| \leq q(b_d + c_d\delta_1^d).
\end{equation}
Utilizing this bound in Eq.~\eqref{eq:xpsizeunfin}, we obtain
\begin{equation}
\label{eq:boundxpsize}
|X_P| \leq \left(2\left\lceil \frac{2\sqrt{C'q(b_d + c_d\delta_1^d)}}{\epsilon_1} \right\rceil + 1\right)^{q(b_d + c_d\delta_1^d)}.
\end{equation}
Plugging this into our $\ell_1$-norm bound from Eq.~\eqref{eq:1normunfin}, we have
\begin{subequations}
\begin{align}
\norm{\mathbf{w}'}_1 &= \sum_{P \in S^{\mathrm{(geo)}}} \sum_{x' \in X_P} |f_P(x')|\\
&\leq \left(2\left\lceil \frac{2\sqrt{C'q(b_d + c_d\delta_1^d)}}{\epsilon_1} \right\rceil + 1\right)^{q(b_d + c_d\delta_1^d)} \sum_{Q \in S^{\mathrm{(geo)}}} |\alpha_Q| \\
&\leq D \left(2\left\lceil \frac{2\sqrt{C'q(b_d + c_d\delta_1^d)}}{\epsilon_1} \right\rceil + 1\right)^{q(b_d + c_d\delta_1^d)},
\end{align}
\end{subequations}
where the second inequality follows from Corollary~\ref{corollary:normineq}, taking $D$ as this constant. We can simplify this expression further by using that $\delta_1 = C_{\mathrm{max}}\log^2(2C/\epsilon_1)$ for sufficiently small $\epsilon_1$ according to Eq.~\eqref{eq:delta1}.
\begin{subequations}
\label{eq:1normbound}
\begin{align}
\norm{\mathbf{w}'}_1 &= \sum_{P \in S^{\mathrm{(geo)}}} \sum_{x' \in X_P} |f_P(x')|\\
&\leq D \left(2\left\lceil \frac{2\sqrt{C'q(b_d + c_d(C_{\mathrm{max}}\log^2(2C/\epsilon_1))^d)}}{\epsilon_1} \right\rceil + 1\right)^{q(b_d + c_d(C_{\mathrm{max}}\log^2(2C/\epsilon_1))^d)}\\
&= \left(\frac{\log^{2d}(1/\epsilon_1)}{\epsilon_1}\right)^{\mathcal{O}(\log^{2d}(1/\epsilon_1))}\\
&= \left(\frac{1}{\epsilon_1}\right)^{\mathcal{O}(\log^{2d}(1/\epsilon_1))}\left(\log^{2d}(1/\epsilon_1)\right)^{\mathcal{O}(\log^{2d}(1/\epsilon_1))}\\
&= 2^{\mathcal{O}(\log^{2d + 1}(1/\epsilon_1))}\\
&= 2^{\mathcal{O}(\mathrm{polylog}(1/\epsilon_1))},
\end{align}
\end{subequations}
which is the promised scaling.
\end{proof}

\subsection{Training error bound}
\label{sec:training}

Using the results in Appendix~\ref{sec:simpleformal}, we can derive a bound on the training error of $g(x) = \mathbf{w}' \cdot \phi(x)$ discussed in the previous section.
The existence of $\mathbf{w}'$ then guarantees that the function $h^*(x) = \mathbf{w}^* \cdot \phi(x)$ found by performing optimization to minimize training error will also yield a training error close to zero.
To prove this rigorously, we first write a precise definition of training error.

\begin{definition}[Training error]
\label{def:training}
Given a function $h(x)$ and a training dataset $\{(x_\ell, y_\ell)\}_{\ell=1}^N$. The \emph{training error} is defined as
\begin{equation}
\hat{R}(h) = \min_{\mathbf{w}} \frac{1}{N} \sum_{\ell=1}^N | h(x_\ell) - y_\ell|^2.
\end{equation}
\end{definition}

\noindent We can bound the training error in the following lemma.

\begin{lemma}[Detailed restatement of Lemma~\ref{lemma:main-training}]
\label{lemma:training}
The function
\begin{equation}
g(x) = \sum_{P \in S^{\mathrm{(geo)}}} \sum_{x' \in X_P} f_P(x') \indicator[x \in T_{x', P}] = \mathbf{w}' \cdot \phi(x),
\end{equation}
achieves training error
\begin{equation}
\hat{R}(g) \leq (\epsilon_1 + \epsilon_2)^2,
\end{equation}
where the training error is defined in Definition~\ref{def:training}.
\end{lemma}

\begin{proof}
This lemma follows directly from Theorem~\ref{thm:simple}. Let $\ell^*$ be defined as
\begin{equation}
\ell^* = \mathop{\mathrm{argmax}}_{1 \leq \ell \leq N} |g(x_\ell) - y_\ell|^2.
\end{equation}
Then, the training error can be bounded above by
\begin{equation}
\hat{R}(g) \leq |g(x_{\ell^*}) - y_{\ell^*}|^2 \leq \left(|g(x_{\ell^*}) - \Tr(O\rho(x_{\ell^*}))| + |\Tr(O\rho(x_{\ell^*})) - y_{\ell^*}|\right)^2,
\end{equation}
where the last inequality follows by triangle inequality. Here, the second term can be bounded by $\epsilon_2$ using definition of our training labels $y_\ell$. For the first term, let $D$ be a constant such that
\begin{equation}
\sum_{P \in S^{\mathrm{(geo)}}} |\alpha_P| \leq D,
\end{equation}
using Corollary~\ref{corollary:normineq}. Then, by Theorem~\ref{thm:simple}, we have
\begin{equation}
|g(x_{\ell^*}) - \Tr(O\rho(x_{\ell^*}))| \leq \epsilon_1 + \epsilon_2.
\end{equation}
Putting everything together, we have
\begin{equation}
\hat{R}(g) \leq (\epsilon_1 + \epsilon_2)^2,
\end{equation}
which is the claimed result.
\end{proof}

Now consider the function $h^*(x) = \mathbf{w}^* \cdot \phi(x)$, where $\mathbf{w}^*$ is obtained by minimizing the training error, such that the training error is larger than the minimum training error by at most $\epsilon_3 / 2$.
We can achieve this using an optimization algorithm described in Appendix~\ref{sec:comp-time}.
Formally, we have the following inequality,
\begin{equation}
    \hat{R}(h^*) \leq \frac{\epsilon_3}{2} + \min_{\substack{\mathbf{w} \in \mathbb{R}^{m_\phi}\\ \norm{\mathbf{w}}_1 \leq B} } \, \frac{1}{N} \sum_{\ell=1}^N \left| \mathbf{w} \cdot \phi(x_\ell) - \Tr(O \rho(x_\ell)) \right|^2.
\end{equation}
Because we have set $B = 2^{\mathrm{polylog}(1 / \epsilon_1)}$,
\begin{equation}
\norm{\mathbf{w}'}_1 = \sum_{P \in S^{\mathrm{(geo)}}} \sum_{x' \in X_P} |f_P(x')| \leq 2^{\mathcal{O}(\mathrm{polylog}(1/\epsilon_1))} = B.
\end{equation}
Therefore, the minimum training error must be at most $\hat{R}(g)$,
\begin{equation}
    \min_{\substack{\mathbf{w} \in \mathbb{R}^{m_\phi}\\ \norm{\mathbf{w}}_1 \leq B} } \, \frac{1}{N} \sum_{\ell=1}^N \left| \mathbf{w} \cdot \phi(x_\ell) - \Tr(O \rho(x_\ell)) \right|^2 \leq \hat{R}(g).
\end{equation}
Together, we have
\begin{equation} \label{eq:training-error-hstar}
    \hat{R}(h^*) \leq \hat{R}(g) + \frac{\epsilon_3}{2} \leq (\epsilon_1 + \epsilon_2)^2 + \frac{\epsilon_3}{2}.
\end{equation}
The last inequality follows from Lemma~\ref{lemma:training}.

\subsection{Prediction error bound}
\label{sec:prediction}

With this bound on the training error, it remains to find a bound on the prediction error of our hypothesis function. To this end, we can use a standard result from machine learning theory about the prediction error of $\ell_1$-norm-constrained linear hypotheses trained using the LASSO algorithm~\cite{doi:10.1137/0907087,mohri2018foundations,tibshirani1996regression}.

\begin{theorem}[Theorem 11.16 in~\cite{mohri2018foundations}]
\label{theorem:predictionerror}
Let $\mathcal{X} \subseteq \mathbb{R}^A$ and $\mathcal{H} = \{\mathbf{x} \in \mathcal{X} \mapsto \mathbf{w} \cdot \mathbf{x} : \norm{\mathbf{w}}_1 \leq B\}$. Let $S = ( (\mathbf{x}_1, y_1),\dots, (\mathbf{x}_N, y_N)) \in (\mathcal{X} \times \mathcal{Y})^N$. Let $\mathcal{D}$ denote a distribution over $\mathcal{X} \times \mathcal{Y}$ according to which the training data $S$ is drawn. Assume that there exists $r_\infty > 0$ such that for all $\mathbf{x} \in \mathcal{X}$, $\norm{\mathbf{x}}_\infty \leq r_\infty$ and $M > 0$ such that $|h(x) - y| \leq M$ for all $(x,y) \in \mathcal{X} \times \mathcal{Y}$. Then, for any $\delta > 0$, with probability at least $1-\delta$, each of the following inequalities holds for all $h \in \mathcal{H}$:
\begin{equation}
\mathop{\mathbb{E}}_{(x,y) \sim \mathcal{D}} |h(x) - y|^2 \triangleq R(h) \leq \hat{R}_S(h) + 2r_\infty B M \sqrt{\frac{2\log(2A)}{N}} + M^2 \sqrt{\frac{\log\frac{1}{\delta}}{2N}}
\end{equation}
where $R(h)$ is the prediction error for the hypothesis $h$ and $\hat{R}_S(h)$ is the training error of $h$ on the training data $S$.
\end{theorem}

\noindent We can use this theorem to prove the prediction error bound in Theorem~\ref{thm:algorithm}.

\begin{proof}[Proof of prediction error in Theorem~\ref{thm:algorithm}]
We utilize Theorem~\ref{theorem:predictionerror} as well as our established lemmas.

First, we demonstrate that we satisfy the conditions of the theorem in our setting.
Here, we view $h$ in Theorem~\ref{theorem:predictionerror} as a function of the higher-dimensional feature vector $\phi(x)$ rather than the $m$-dimensional vector $x \in [-1, 1]^m$ so that $h$ is a linear hypothesis.
In this perspective, our input space $\mathcal{X}$ is the feature space $\{0,1\}^{m_\phi} \subseteq \mathbb{R}^{m_\phi}$, as the indicator functions we are evaluating only take $0$-$1$ values.
In our case, the dimension $A$ is given by
\begin{equation}
A = m_\phi \triangleq \sum_{P \in S^{\mathrm{(geo)}}} |X_P|.
\end{equation}
Moreover, the training data we are given is $S = ( (\phi(x_1), y_1),\dots, (\phi(x_N), y_N)) \in (\mathcal{X} \times \mathcal{Y})^N$, where $y_\ell$ is such that
\begin{equation}
|y_\ell - \Tr(O\rho(x_\ell))| \leq \epsilon_2.
\end{equation}
Again, we are thinking of $h$ as a function that takes the input $\phi(x)$.
Furthermore, since $\phi(x_\ell) \in \{0,1\}^{m_\phi}$ for all $\ell = 1,\dots,N$, we can see that $\norm{\phi(x_\ell)}_\infty \leq 1 = r_\infty$.
Moreover, the hypothesis class $\mathcal{H}$ is given by the set of the functions of the same form as $h$, i.e., $\mathcal{H} = \{\phi(x) \in \mathcal{X} \mapsto \mathbf{w} \cdot \phi(x) : \norm{\mathbf{w}}_1 \leq B\}$ with $B = 2^{\mathrm{polylog}(1/\epsilon)}$.
By considering $M = 2^{\mathcal{O}(\mathrm{polylog}(1/\epsilon_1))}$, we also have $|h(x_\ell) - y_\ell| \leq M$ for all $\ell = 1,\dots, N$ because
\begin{equation}
|h(x_\ell) - y_\ell| \leq |\mathbf{w} \cdot \phi(x_\ell)| + |y_\ell| \leq \norm{\mathbf{w}}_1 \norm{\phi(x)}_\infty + 2 \leq 2^{\mathcal{O}(\mathrm{polylog}(1/\epsilon_1))} + 2 = 2^{\mathcal{O}(\mathrm{polylog}(1/\epsilon_1))},
\end{equation}
where the second inequality follows by H\"older's inequality.
Furthermore, by Eq.~\eqref{eq:training-error-hstar},
the learned model $h^*(x) = \mathbf{w}^* \cdot \phi(x)$ achieves $\hat{R}(h^*) \leq  (\epsilon_1 + \epsilon_2)^2 + (\epsilon_3 / 2)$. Thus, by Theorem~\ref{theorem:predictionerror},
\begin{equation}
R(h^*) \leq (\epsilon_1 + \epsilon_2)^2 + \frac{\epsilon_3}{2} + 2B M\sqrt{\frac{2\log(2m_\phi)}{N}} + M^2\sqrt{\frac{\log{\frac{1}{\delta}}}{2N}}
\end{equation}
with probability at least $1 - \delta$.
In order to bound the prediction error above by $(\epsilon_1 + \epsilon_2)^2 + \epsilon_3$, we need $N$ to be large enough such that
\begin{equation}
\label{eq:wtsprediction}
2 B M\sqrt{\frac{2\log(2m_\phi)}{N}} + M^2\sqrt{\frac{\log{\frac{1}{\delta}}}{2N}} \leq \frac{\epsilon_3}{2}.
\end{equation}
We can upper bound $m_\phi$ using the same approach as in the proof of Lemma~\ref{lemma:1norm}. Explicitly, using Eq.~\eqref{eq:boundxpsize} and Eq.~\eqref{eq:1normbound}, we have
\begin{subequations}
\begin{align}
m_\phi &= \sum_{P \in S^{\mathrm{(geo)}}} |X_P| \leq \sum_{P \in S^{\mathrm{(geo)}}} \left(2\left\lceil \frac{\sqrt{C'q(b_d + c_d\delta_1^d)}}{\epsilon_1} \right\rceil + 1\right)^{q(b_d + c_d\delta_1^d)}\\
&= 2^{\mathcal{O}(\mathrm{polylog}(1/\epsilon_1))}\mathcal{O}(n), \label{eq:featuresize}
\end{align}
\end{subequations}
where the last equality follows because $|S^{(\mathrm{geo})}| = \mathcal{O}(n)$. Plugging everything into the left hand side of Eq.~\eqref{eq:wtsprediction}, we have
\begin{subequations}
\begin{align}
&\hspace{10pt}2 B M\sqrt{\frac{2\log(2m_\phi)}{N}} + M^2\sqrt{\frac{\log{\frac{1}{\delta}}}{2N}}\\
&\leq 2\sqrt{2}\left(2^{\mathcal{O}(\mathrm{polylog}(1/\epsilon_1))}\right)^2 \sqrt{\frac{\log\left(2\cdot 2^{\mathcal{O}(\mathrm{polylog}(1/\epsilon_1))}\mathcal{O}(n)\right)} {N}} \\
&+ \frac{1}{\sqrt{2}}\left(2^{\mathcal{O}(\mathrm{polylog}(1/\epsilon_1))}\right)^2\sqrt{\frac{\log\frac{1}{\delta}}{N}}\\
&= 2^{\mathcal{O}(\mathrm{polylog}(1/\epsilon_1))}\frac{1}{\sqrt{N}}\left(\sqrt{\mathcal{O}(\mathrm{polylog}(1/\epsilon_1)) + \mathcal{O}(\log(n))} + \sqrt{\log \frac{1}{\delta}}\right).
\end{align}
\end{subequations}
To upper bounded the above by $\frac{\epsilon_3}{2}$, we choose
\begin{subequations}
\begin{align}
N &= \frac{4}{\epsilon_3^2} \left(2^{\mathcal{O}(\mathrm{polylog}(1/\epsilon_1))}\right)^2 \left(\sqrt{\mathcal{O}(\mathrm{polylog}(1/\epsilon_1)) + \mathcal{O}(\log(n))} + \sqrt{\log \frac{1}{\delta}}\right)^2 \\
&= 2^{\mathcal{O}(\log(1/\epsilon_3) + \mathrm{polylog}(1/\epsilon_1))} \log(n / \delta).
\end{align}
\end{subequations}
Together, the training data size $N$ given above guarantees that $R(h^*) \leq (\epsilon_1 + \epsilon_2)^2 + \epsilon_3$ with probability at least $1 - \delta$.
\end{proof}

\subsection{Computational time for training and prediction}
\label{sec:comp-time}

Finally, we find the computation time required for the ML algorithm's training and prediction. To this end, we utilize standard results about the training time of the LASSO algorithm~\cite{hazan2012linear}.

\begin{proof}[Proof of computational time in Theorem~\ref{thm:algorithm}]

The training time is dominated by the time required for $\ell_1$-regularized regression (LASSO) over the feature space defined by the feature map $\phi$.
It is well-known that to obtain a training error at most $(\epsilon_3 / 2)$ larger than the optimal function value, the LASSO algorithm on the feature space can be executed in time $\mathcal{O}\left(\frac{m_\phi \log m_\phi}{\epsilon_3^2}\right)$~\cite{hazan2012linear}, where $m_\phi$ is the dimension of the feature space.
By Eq.~\eqref{eq:featuresize}, we know that
\begin{equation}
m_\phi = \mathcal{O}(n) 2^{\mathcal{O}(\mathrm{polylog}(1/\epsilon_1))}
\end{equation}
Plugging this into the time required for LASSO, we have
\begin{subequations}
\begin{align}
\mathcal{O}\left(\frac{m_\phi \log m_\phi}{\epsilon_3^2}\right) &= \mathcal{O}\left(\frac{\mathcal{O}(n) 2^{\mathcal{O}(\mathrm{polylog}(1/\epsilon_1))} \log(\mathcal{O}(n)2^{\mathcal{O}(\mathrm{polylog}(1/\epsilon_1))})}{\epsilon_3^2}\right)\\
&= \mathcal{O}(n2^{\mathcal{O}(\log(1/\epsilon_3) + \mathrm{polylog}(1/\epsilon_1))}\left(\mathcal{O}(\mathrm{polylog(1/\epsilon_1)}) + \mathcal{O}(\log(n))\right))\\
&= n\log n \,2^{\mathcal{O}(\log(1 / \epsilon_3) + \mathrm{polylog(1/\epsilon_1)})}\\
&= \mathcal{O}(nN),
\end{align}
\end{subequations}
where the last equality follows by the definition of the training data size $N$.

The prediction time is the amount of time it takes to compute $h^*(x) = \mathbf{w}^* \cdot \phi(x_\ell)$, which takes time $\mathcal{O}(m_\phi)$. This can also be upper bounded by $\mathcal{O}(nN)$.
\end{proof}

\section{Details of numerical experiments}
\label{sec:numerics}

For the numerical experiments, we consider the two-dimensional antiferromagnetic Heisenberg model.
In this setting, spin-$1/2$ particles are placed on sites in a 2D lattice.
The Hamiltonian is
\begin{equation}
H = \sum_{\langle ij \rangle} J_{ij} (X_i X_j + Y_i Y_j + Z_i Z_j),
\end{equation}
where the summation ranges over all pairs $\langle ij \rangle$ of neighboring sites on the lattice and the couplings $\{J_{ij}\}$ are sampled uniformly from the interval $[0,2]$.
Here, the parameter $x$ is a list of all couplings $J_{ij}$ so that the dimension of the parameter space is $m = O(n)$, where $n$ is the system size.
We are interested in predicting ground state properties, which in this case are the two-body correlation functions for each pair of qubits on the lattice. In particular, this correlation function is the expectation value of
\begin{equation}
C_{ij} = \frac{1}{3}(X_i X_j + Y_i Y_j + Z_i Z_j),
\end{equation}
for each pair of qubits $\langle ij \rangle$.

We generated training and testing data for this model using the same method as~\cite{huang2021provably}.
For completeness, we briefly discuss this here.
For each parameter vector of random couplings sampled uniformly from $[0,2]$, we approximated the ground state using the density-matrix renormalization group (DMRG)~\cite{white1992density} based on matrix product states (MPS)~\cite{SCHOLLWOCK201196}.
We first consider an initial random MPS with bond dimension $\chi = 10$ and variationally optimize it using a singular value decomposition cutoff of $10^{-8}$.
We terminate the DMRG runs when the change in energy is less than $10^{-4}$.
After DMRG converges, we perform randomized Pauli measurements by locally rotating into the corresponding Pauli bases and sampling the rotated state~\cite{FerrisSampling}.
In this work, we utilize two different data sets: one which is the same as in~\cite{huang2021provably} and the other which is generated in the same way but contains more data points.

We consider classical machine learning models given by first performing a feature mapping $\phi$ on the input vector $x$ and then running $\ell_1$-regularized regression (LASSO) over the feature $\phi(x)$ space, as described in Appendix~\ref{sec:ml-algorithm}.
However, while the indicator function feature map was a useful tool to obtain our rigorous guarantees, it is often hard to discretize a high-dimension parameter space into $X_P$ in practice.
Thus, we instead utilize random Fourier features~\cite{rahimi2007random}.
One can think of this as a single layer of a randomly initialized neural network.
Explicitly, this feature map is
\begin{equation}
\phi: z \mapsto \begin{pmatrix}
\cos\left(\frac{\gamma}{\sqrt{l}}(\omega_1 \cdot z)\right)\\
\sin\left(\frac{\gamma}{\sqrt{l}}(\omega_1 \cdot z)\right)\\
\vdots\\
\cos\left(\frac{\gamma}{\sqrt{l}}(\omega_R \cdot z)\right)\\
\sin\left(\frac{\gamma}{\sqrt{l}}(\omega_R \cdot z)\right)
\end{pmatrix},
\end{equation}
where $l$ is the length of the vector $z$, $\gamma > 0$ and $R > 0$ are tunable hyperparameters, and $\omega_i$ are $l$-dimensional vectors sampled from a multivariate standard normal distribution.
Here, for each vector $z$, $\phi(z)$ is a $2R$-dimensional vector.
Thus, the hyperparameter $R$ determines the length of the feature vector.
We consider a set of different hyperparameters:
\begin{align}
R &\in \{5, 10, 20, 40\},\\
\gamma &\in \{0.4, 0.5, 0.6, 0.65, 0.7, 0.75\}.
\end{align}

Using this feature map, the ML algorithm is implemented as follows.
First, we decompose $x$ into several vectors corresponding to local regions of a given local term of the Hamiltonian.
This is analogous to the discretization of the parameter space using $X_P$.
Explicitly, the decomposition is performed in the following way.
First, recall that in the 2D antiferromagnetic Heisenberg model, qubits are placed on sites in a 2D lattice.
Thus, each local term can be viewed as an edge between neighboring sites on the lattice.
We construct a local region around this edge by including all edges within an $\ell_1$-distance $\delta_1$.
This is analogous to Eq.~\eqref{eq:ip}.
Now, for each vector resulting from the decomposition of $x$, we apply the feature map $\phi$ and concatenate all vectors together to obtain $\phi(x)$.
Finally, we run the LASSO algorithm using \texttt{scikit-learn}, a Python package~\cite{scikit-learn}.
Here, LASSO optimizes the objective function
\begin{equation}
\frac{1}{2N} \norm{y - Xw}_2^2 + \alpha \norm{w}_1,
\end{equation}
where $N$ is the amount of training data, $y$ is a vector of the training data labels $\{y_\ell\}_{\ell=1}^N$, $X$ is a matrix of the training data inputs $\{x_\ell\}_{\ell=1}^N$, $w$ is a vector of coefficients we want to learn, and $\alpha > 0$ is a regularization parameter.
We consider a set of different regularization parameters
\begin{equation}
\alpha \in \{2^{-8}, 2^{-7}, 2^{-6}, 2^{-5}\}.
\end{equation}
We consider several different classical ML models, corresponding to these choices of hyperparameters $R, \gamma, \alpha$.
Thus, we perform model selection to determine the optimal choice of these hyperparameters.
To this end, we consider $M$ different values of the parameter $x = \{J_{ij}\}$, where $M$ is roughly around $100$ across different system sizes\footnote{The data set size from~\cite{huang2021provably} varies slightly depending on system size. For lattices of sizes $4\times 5$, $5\times 5$, and $7\times 5$, $M = 100$. However, for $6\times 5$, $M = 97$, for $8\times 5$, $M = 92$ and for $9\times 5$, $M = 89$.}.
From these $M$ data points, we randomly choose half of these points as training data (i.e., $N = M/2$) and the remaining half is test data.
For each ground state property we want to predict, we choose one value of each of $R, \gamma, \alpha$ such that the root-mean-square error is minimized when performing $4$-fold cross-validation, which is also implemented using \texttt{scikit-learn}.
Finally, we test the performance of the ML model with these chosen hyperparameters using the test data.

For each vector $x$ that we tested on, we would predict the correlation functions for all pairs of qubits $\langle ij \rangle$.
Hence, the prediction error is averaged over $(M/2) \times (1.8 n - 5) \approx 1500$ to $3500$ predictions, i.e., over all of the test data and all pairs of qubits.
Despite $M / 2$ being only of around $50$, the prediction errors reported in the plots are statistically sound given the large total number of predictions.
The standard deviation of the exact correlation functions in the data varies slightly across different system sizes\footnote{The standard deviation for system size $4 \times 5$ is around $0.192$, $5 \times 5$ is around $0.199$, $6 \times 5$ is around $0.187$, $7 \times 5$ is around $0.193$, $8 \times 5$ is around $0.190$, $9 \times 5$ is around $0.187$}.
When the standard deviation is smaller, the prediction error will also be smaller.
To judge the difficulty to predict the correlation functions across different system sizes, we normalize the standard deviation to be the average standard deviation of $0.191$.
We also include experiments where we vary the training data size $N$ or the classical shadow size $T$, i.e., the number of randomized Pauli measurements used to approximate the ground state.
For a fixed training data size of $N = M/2$, we vary the classical shadow size with values in $T \in \{50, 100, 250, 500, 1000\}$.
Similarly, for a fixed shadow size of $T = 500$, we vary the training data size with values $N = pM$ for $p \in \{0.1, 0.3, 0.5, 0.7, 0.9\}$.
The numerical results of these experiments are summarized in Figure~\ref{fig:heisenberg}.

\end{document}